\documentclass[5p,times,twocolumn]{elsarticle}
\usepackage{amsmath,amssymb,amsthm}
\usepackage{graphicx}
\usepackage{booktabs}
\usepackage{multirow}
\usepackage{makecell}
\usepackage{pdflscape}
\usepackage{caption}

\usepackage{microtype}

\usepackage{dblfloatfix}
\usepackage{float}
\usepackage{longtable}
\usepackage{array}
\usepackage{tabularx}
\usepackage{threeparttable}
\renewcommand{\arraystretch}{2.0}
\usepackage{xcolor}
\usepackage[colorlinks=true,
linkcolor=blue,
citecolor=blue,
urlcolor=blue]{hyperref}

\let\oldref\ref
\renewcommand{\ref}[1]{\textcolor{blue}{[\oldref{#1}]}}
\journal{Insurance: Mathematics and Economics}

\usepackage{listings}
\usepackage{xcolor}

\lstdefinestyle{Rstyle}{
  language=R,
  basicstyle=\ttfamily\scriptsize,
  breaklines=true,
  breakatwhitespace=false,
  columns=fullflexible,
  keepspaces=true,
  showstringspaces=false,
  frame=single,
  xleftmargin=0pt,
  xrightmargin=0pt,
  aboveskip=0.7em,
  belowskip=0.7em
}

\newtheorem{proposition}{Proposition}

\begin{document}

\begin{frontmatter}

\title{A Copula-Based Framework for Multivariate Zero-Inflated Mixed Poisson Models}

\author[inst1]{Nguyen Quang Huy}
\address[inst1]{Actuarial Sciences Laboratory, College of Technology, National Economics University, Hanoi, Vietnam}
\author[inst1]{Pham Thi Hong Tham\corref{cor1}}
\ead{thamtkt@neu.edu.vn}
\cortext[cor1]{Corresponding author}

\author[inst2]{Tran Thi Thu Hien}
\address[inst2]{Faculty of Management and Tourism, Hanoi University, Hanoi, Vietnam}

\begin{abstract}
Multivariate count data often contain overdispersion, excess zeros, and complex dependence. Existing multivariate zero-inflated count models usually use a single dependence structure to jointly model structural zeros and count outcomes. This makes the two sources of dependence difficult to interpret separately and often leads to slow computation or intractable likelihoods. This paper proposes a general framework for multivariate zero-inflated mixed Poisson models that explicitly separates these two sources of dependence within a unified likelihood-based formulation. The proposed hierarchical model combines zero-inflated mixed Poisson marginals with separate dependence models for the structural-zero and latent-intensity components. Dependence among structural zeros is modeled by standard parametric copulas, while dependence among latent mixing variables is modeled by checkerboard copulas. Likelihood-based inference is carried out using an inference-functions-for-margins procedure, and the proposed likelihood also supports full maximum likelihood estimation when computationally feasible. Simulation studies show accurate parameter estimation and good finite-sample performance under different dependence structures and latent mixing distributions. Applications to benchmark datasets from existing statistical software and a healthcare utilization dataset demonstrate the flexibility, computational efficiency, and practical usefulness of the proposed framework.

\end{abstract}


\end{frontmatter}

\section{Introduction}
\subsection{Literature Review}
The analysis of count data has received much attention in actuarial science, insurance, healthcare, epidemiology, and risk management because of its wide range of applications. In many practical settings, count data exhibit overdispersion, excess zeros, and complex dependence at the same time. Developing statistical models that can handle these features remains an active research topic.

A large body of research has focused on modeling overdispersed count data using mixed Poisson distributions. In these models, the Poisson intensity is treated as a latent random variable, introducing additional heterogeneity into the count process. Classical examples include the Poisson--Inverse Gaussian model of \citet{Willmot1987} and the general family of mixed Poisson distributions reviewed by \citet{karlis2005mixed}. In actuarial science, mixed Poisson models are widely used for claim frequency modeling because they naturally account for unobserved risk heterogeneity while remaining analytically tractable. Comprehensive discussions are given by \citet{denuit2007actuarial}, while \citet{frees2009regression} describe their role in actuarial regression models. More recently, mixed Poisson and related latent variable models have been applied to insurance pricing, reserving, and risk classification \citep{garrido2016generalized}.

Another common feature of count data is the presence of excess zeros. In insurance portfolios, many policyholders do not report any claims during the observation period. Likewise, healthcare datasets often contain many individuals with no use of particular services. To model this feature, \citet{lambert1992zero} proposed the Zero-Inflated Poisson (ZIP) model, in which observed zeros arise from a mixture of a structural-zero component and a standard count-generating process. Later work extended the ZIP model to include random effects, overdispersion, and correlated observations \citep{ridout1998models,hall2000zero,min2005random}. These models are widely used because they separate the probability of structural zeros from the intensity of the count process.

A separate line of research has focused on multivariate count data. Dependence among multiple count variables arises naturally in many applications. For example, claim frequencies from different insurance coverages may be affected by common risk factors, while different types of healthcare utilization are often influenced by an individual's overall health condition. Early work was based on multivariate Poisson distributions and shared random-effects models \citep{kocherlakota2017bivariate,chib2001markov}. More flexible models based on finite mixtures and latent variable representations were later proposed by \citet{karlis2007finite}. Reviews of count regression models and multivariate count data are provided by \citet{cameron2013regression}.

Copula-based methods provide a flexible way to model multivariate count data by separating marginal distributions from the dependence structure. This allows different marginal models to be combined with a wide range of dependence models within a unified likelihood framework. Similar ideas have also been successfully used in other statistical problems, such as copula link-based additive models for censored time-to-event data \citep{petti2022copula}. Early work on copulas for discrete data includes the practitioner-oriented review of \citet{triverdi2005copula} and the overview of \citet{genest2007primer}. A comprehensive treatment of copula theory is given by \citet{joe2014dependence}. For multivariate count data, copula-based models have been developed by \citet{nikoloulopoulos2009finite} and \citet{nikoloulopoulos2013copula}, among others, and have been applied to healthcare utilization, insurance claim frequencies, and other types of multivariate count data.

More recently, latent-factor copula models have been proposed to model high-dimensional dependence. These models provide a parsimonious representation of dependence and avoid the large number of parameters required by fully parameterized copula models \citep{krupskii2013factor}. As a result, they have made copula-based methods more practical for high-dimensional applications.

Recent research has also highlighted the role of copulas for discrete data. In particular, \citet{panagiotelis2012pair} proposed a copula-based framework for multivariate discrete distributions and developed practical methods for likelihood evaluation when exact computation of discrete copula probabilities is difficult. These developments are especially useful for multivariate count data, where discrete outcomes and complex dependence often occur together.

Latent variable models have also extended the applicability of mixed Poisson models. For example, Poisson--Lognormal models provide a flexible way to capture dependence through latent intensities and have been successfully applied to multivariate count data \citep{chiquet2021poisson}. Other models based on compound Poisson and related mixed count distributions have also been proposed for highly zero-inflated data \citep{park2023zero}. Together, these studies show that latent-intensity models provide a natural way to account for both overdispersion and dependence.

Dependence modeling has also received increasing attention in actuarial science. For example, \citet{shi2014multivariate} proposed multivariate negative binomial models for insurance claim frequencies, while \citet{garrido2016generalized} studied generalized linear models for dependent claim frequencies and severities. Copula-based methods have also been used to model the dependence between claim frequency and claim severity \citep{shi2015dependent}. These studies demonstrate the importance of modeling heterogeneity and dependence jointly in insurance applications.

Modeling excess zeros and dependence simultaneously remains much more challenging than modeling either feature alone. Early zero-inflated models were extended to include random effects and correlated latent structures \citep{hall2000zero,min2005random}. More recent work has introduced copulas, latent variable models, and finite mixtures into zero-inflated count models. Examples include Gaussian copula models for zero-inflated count time series \citep{alqawba2019zero}, copula-based finite mixture regression models for insurance claim counts \citep{bermudez2022copula}, and bivariate zero-inflated count copula regression models implemented in the \texttt{bizicount} package \citep{niehaus2024bizicount}. Comprehensive reviews of recent developments are provided by \citet{young2022zero} and \citet{young2022zeroPartII}.

Overall, these studies have greatly improved the modeling of overdispersion, excess zeros, latent heterogeneity, and dependence in count data. However, most existing multivariate zero-inflated count models introduce dependence through a single copula, shared random effect, or latent factor. As a result, they do not distinguish dependence arising from the structural-zero component and dependence arising from latent count intensities. In addition, copula models for discrete count data often require evaluating multidimensional rectangle probabilities, which makes exact likelihood evaluation computationally expensive and sometimes infeasible in high dimensions.

\subsection{Research Gap and Motivation}

Although recent studies have improved the modeling of multivariate zero-inflated count data, important challenges remain. First, most existing models use a single dependence mechanism, such as a copula, a shared random effect, or a latent factor, to model both structural zeros and count outcomes. As a result, they do not distinguish dependence arising from the structural-zero component and dependence arising from latent count intensities. This makes the two sources of dependence difficult to interpret and estimate separately.

This distinction is important in actuarial applications. Some policyholders may have a low probability of making claims across several coverages, creating dependence among structural zeros. At the same time, policyholders who enter the claim-generating process may still have dependent claim frequencies because they share common latent risk factors. These two types of dependence have different interpretations and may lead to different conclusions for risk classification, portfolio segmentation, premium calculation, and risk assessment. Modeling them separately therefore provides a clearer understanding of the dependence structure.

A second challenge is computation. Existing copula-based models for multivariate count data often require the evaluation of high-dimensional rectangle probabilities, making likelihood evaluation slow and computationally expensive as the dimension increases. This limits the use of flexible copula models for high-dimensional count data.

To address these challenges, we propose a multivariate zero-inflated mixed Poisson regression framework that separates dependence in the structural-zero component from dependence among latent mixing variables. The framework allows different dependence models for these two components within a single likelihood framework. Dependence among structural zeros is modeled using standard parametric copulas, while dependence among latent mixing variables is modeled using checkerboard copulas. Gaussian factor representations are used to obtain efficient likelihood computation in moderate- and high-dimensional settings.

The main contributions of this paper are threefold. First, we propose a general multivariate zero-inflated mixed Poisson regression framework that models dependence in the structural-zero component and the latent-intensity process separately. This makes the two sources of dependence easier to interpret. Second, the proposed framework applies to a broad class of mixed Poisson distributions under a common modeling and estimation framework, so different mixing distributions can be used without changing the dependence model or the estimation procedure. Third, we develop a computationally efficient estimation procedure for likelihood-based inference, making the proposed framework practical for moderate- and high-dimensional multivariate count data.

\section{Multivariate Zero-Inflated Mixed Poisson Model}

\subsection{Hierarchical Model}\label{sec:hierarchical_model}

Let $\mathbf{Y}=(Y_1,\ldots,Y_k)^\top$ denote a $k$-dimensional vector of count responses.

To jointly accommodate excess zeros, overdispersion, and multivariate dependence, we employ a three-layer hierarchical model consisting of a structural-zero layer, a latent-intensity layer, and a conditional Poisson layer.

For each margin \(j=1,\ldots,k\), let \(B_j\in\{0,1\}\) be a latent Bernoulli indicator denoting whether an observation belongs to the structural-zero state. Let \(N_j\) denote the underlying count variable, and assume that \(B_j\) and \(N_j\) are independent. The observed response is defined as
\begin{equation}
Y_j=B_jN_j,
\label{eq:Y_definition}
\end{equation}
so that
\[
Y_j=
\begin{cases}
0,& B_j=0,\\
N_j,& B_j=1.
\end{cases}
\]

Conditional on a latent-intensity variable \(\Lambda_j\), we assume
\[
N_j\mid\Lambda_j=\lambda_j
\sim
\mathrm{Poisson}(\lambda_j),
\]
where
\begin{equation}
\Lambda_j
\sim
F_j(\mu_j,\phi_j),
\label{eq:mixing_distribution}
\end{equation}
and \(F_j\) denotes a positive continuous mixing distribution parameterized by a mean parameter \(\mu_j\) and a dispersion parameter \(\phi_j\).

Throughout the paper, we assume
\[
\mathbb{E}(\Lambda_j)=\mu_j,\qquad
\mathbb{V}(\Lambda_j)=\phi_jV_j(\mu_j),
\]
where \(V_j(\cdot)\) is the variance function of the mixing distribution (see \ref{app_distribution_for_lambda}). For many common mixing distributions, the variance function has the form
\[
V_j(\mu_j)=\mu_j,\qquad
V_j(\mu_j)=\mu_j^2,\qquad
\text{or}\qquad
V_j(\mu_j)=\mu_j^3.
\]
More general variance functions can also be used if needed. Unless stated otherwise, we use the natural variance function of each mixing distribution. In a few cases, we use a different variance function to match the parameterization used in existing software, allowing a direct comparison with competing methods.

The proposed framework can be used with many different mixing distributions, and different margins may follow different distributions. In this paper, we consider five commonly used choices: Gamma, Lognormal, Pareto, Weibull, and Inverse Gaussian. These give the corresponding Poisson--Gamma, Poisson--Lognormal, Poisson--Pareto, Poisson--Weibull, and Poisson--Inverse Gaussian margins. The proposed framework is not limited to these distributions and can be extended to other mixing distributions with suitable distributional and computational properties.

The hierarchical construction naturally separates the structural-zero component from the latent-intensity component, providing the foundation for introducing distinct dependence structures for the two latent layers.

\subsection{Marginal Model Specification}\label{sec:marginal_selection}
The previous subsection introduced the hierarchical structure of the proposed model. We now specify the regression models for the structural-zero probabilities and the latent-intensity means.

For each margin \(j=1,\ldots,k\), suppose we observe \(n\) independent observations. Let
\[
(Y_{ij},B_{ij},N_{ij},\Lambda_{ij},\mathbf X_i),
\qquad i=1,\ldots,n,
\]
denote the observed response, the latent variables, and the covariate vector for observation \(i\).

For the structural-zero component, let
\[
\pi_{ij}
=
\mathbb{P}(B_{ij}=1\mid \mathbf X_i=\mathbf x_i)
\]
be the probability that observation \(i\) belongs to the count-generating state in margin \(j\). We model this probability using a logistic regression,
\begin{equation}
\operatorname{logit}(\pi_{ij})
=
\mathbf x_i^\top
\boldsymbol\gamma_j,
\label{eq:logit_pi}
\end{equation}
where \(\boldsymbol\gamma_j\) is the regression coefficient vector.

For the latent-intensity component, we model the conditional mean of the mixing distribution as
\begin{equation}
\mu_{ij}
=
\exp\!\left(
\mathbf x_i^\top
\boldsymbol\beta_j
\right),
\label{eq:mean_link}
\end{equation}
where \(\boldsymbol\beta_j\) is the regression coefficient vector.

In general, the structural-zero and latent-intensity components may use different sets of covariates. Since variable selection is not the focus of this paper, we use the same covariate vector \(\mathbf X_i\) for both components to simplify the notation.

The proposed framework can be used with many positive continuous mixing distributions. In this paper, we consider five commonly used distributions: Gamma, Lognormal, Pareto, Weibull, and Inverse Gaussian. These distributions are used throughout the methodological development, simulation studies, and real-data applications.

For each mixing distribution, the corresponding marginal probability mass function is derived separately. \ref{app_distribution_for_lambda} gives the computational details, including closed-form expressions whenever available and numerical methods for the remaining cases.

\subsection{Dependence Modeling}\label{sec:depe_model}

The hierarchical model in the previous subsection naturally separates the observed count process into two latent components. The Bernoulli indicators determine whether an observation belongs to the structural-zero state, while the latent-intensity variables determine the latent count process. This representation allows the two sources of dependence to be modeled separately.

Let
\[
\mathbf B=(B_1,\ldots,B_k)^\top,
\qquad
\boldsymbol{\Lambda}=(\Lambda_1,\ldots,\Lambda_k)^\top,
\]
denote the vectors of structural-zero indicators and latent-intensity variables, respectively. Throughout this paper, we assume that \(\mathbf B\) and \(\boldsymbol{\Lambda}\) are mutually independent. Dependence among the observed responses is induced by the dependence within \(\mathbf B\) and within \(\boldsymbol{\Lambda}\). This assumption allows the two sources of dependence to be modeled and interpreted separately.

Dependence among the structural-zero indicators is modeled using a copula. Let
\[
\mathbf U=(U_1,\ldots,U_k)^\top
\]
be a vector of standard uniform random variables with joint distribution
\[
\mathbb{P}(U_1\le u_1,\ldots,U_k\le u_k)
=
C_{\mathbf B}(u_1,\ldots,u_k),
\]
where \(C_{\mathbf B}\) is a \(k\)-dimensional copula.
We define
\[
B_j=
\begin{cases}
1, & U_j \le \pi_j,\\
0, & U_j > \pi_j,
\end{cases}
\qquad j=1,\ldots,k.
\]
where
\[
\pi_j=\mathbb{P}(B_j=1).
\]
It follows that
\[
B_j\sim\operatorname{Bernoulli}(\pi_j),
\]
and the dependence among the structural-zero indicators is determined by the copula \(C_{\mathbf B}\).

For any
\(\mathbf b=(b_1,\ldots,b_k)^\top\in\{0,1\}^k\),
the joint probability mass function of
\(\mathbf B\)
is given by
\begin{equation}
\begin{aligned}
p_{\mathbf B}(\mathbf b)
&=
\mathbb P(\mathbf B=\mathbf b)
\\
&=
\sum_{\boldsymbol{\delta}\in\{0,1\}^{k}}
(-1)^{k-\sum_{j=1}^{k}\delta_j}
C_{\mathbf B}
\!\left(
v_1(\delta_1,b_1),
\ldots,
v_k(\delta_k,b_k)
\right),
\end{aligned}
\label{eq:pmf_B}
\end{equation}
where
\[
v_j(\delta,b)
=
(1-\delta)(1-b)\pi_j
+
\delta\left[(1-b)+b\pi_j\right],
\qquad
j=1,\ldots,k.
\]

Dependence among the latent-intensity variables is modeled using another copula. Specifically,
\[
\mathbb{P}(\Lambda_1\le\lambda_1,\ldots,\Lambda_k\le\lambda_k)
=
C_{\boldsymbol{\Lambda}}
\!\left(
F_1(\lambda_1),\ldots,F_k(\lambda_k)
\right),
\]
where \(C_{\boldsymbol{\Lambda}}\) is a \(k\)-dimensional copula and \(F_j\) is the marginal distribution of \(\Lambda_j\) introduced in Section~\ref{sec:marginal_selection}.

\subsection{Joint Distribution of the Latent Counts}

Conditional on the latent-intensity vector
\[
\boldsymbol{\Lambda}_i
=
(\Lambda_{i1},\ldots,\Lambda_{ik})^\top,
\]
the latent counts are assumed to be conditionally independent,
\[
\mathbf N_i \mid \boldsymbol{\Lambda}_i=\boldsymbol{\lambda}_i
\sim
\prod_{j=1}^{k}
\mathrm{Poisson}(\lambda_{ij}),
\]

The joint distribution of
\(\boldsymbol{\Lambda}_i\)
is determined by the marginal mixing distributions and the copula model introduced in Section~\ref{sec:depe_model}. The following proposition gives the joint probability mass function of the latent count vector after integrating out the latent-intensity variables.

\begin{proposition}
\label{prop_eqpmfN}

Fix an observation \(i\) and condition on
\(\mathbf X_i=\mathbf x_i\). Assume that:

\begin{enumerate}[(i)]
\item
For each margin \(j=1,\ldots,k\),
\[
F_j(\cdot;\mu_{ij},\phi_j)
\]
is continuous and strictly increasing on \((0,\infty)\), with quantile function
\[
Q_j(\cdot;\mu_{ij},\phi_j)
=
F_j^{-1}(\cdot;\mu_{ij},\phi_j).
\]

\item
The dependence among
\(\boldsymbol{\Lambda}_i\)
is modeled by an absolutely continuous copula
\(C_{\boldsymbol{\Lambda}}(\cdot;\boldsymbol{\theta}_{\Lambda})\)
with copula density
\(c_{\boldsymbol{\Lambda}}(\cdot;\boldsymbol{\theta}_{\Lambda})\) where \(\boldsymbol{\theta}_{\Lambda}\) denotes the copula parameter vector.

\end{enumerate}

Then
\begin{equation}
\begin{aligned}
&
\mathbb{P}(\mathbf N_i=\mathbf n_i \mid \mathbf X_i=\mathbf x_i) \\
&=
\int_0^1\cdots\int_0^1
\left[
\prod_{j=1}^{k}
\frac{
\exp\!\left\{-Q_j(u_j;\mu_{ij},\phi_j)\right\}
Q_j(u_j;\mu_{ij},\phi_j)^{n_{ij}}
}{
\Gamma(n_{ij}+1)
}
\right]
\\
&\hspace{3.6cm}\times
c_{\boldsymbol\Lambda}(u_1,\ldots,u_k;\boldsymbol\theta_\Lambda)
\,du_1\cdots du_k.
\end{aligned}
\label{eq:jointpmf_uniform}
\end{equation}

\end{proposition}

The proof is given in \ref{proof_of_pro1}.

Equation~(\ref{eq:jointpmf_uniform}) is valid for any absolutely continuous copula and any continuous mixing distribution with a known quantile function. For general copulas, however, the multidimensional integral usually has no closed-form solution and must be evaluated numerically. The computational cost therefore increases rapidly with the dimension.

To obtain a more efficient likelihood, we use a checkerboard copula. As shown in the next subsection, the checkerboard representation replaces the multidimensional integral in (\ref{eq:jointpmf_uniform}) by a finite weighted sum over checkerboard cells, leading to much faster likelihood evaluation.

\subsection{Checkerboard Copula for Latent Intensities}
\label{checker_copula}

Following \citet{kuzmenko2020checkerboard}, we model the dependence among the latent-intensity variables using a checkerboard copula. Let
\[
0=u_0<u_1<\cdots<u_d=1,
\qquad
u_m=\frac{m}{d},
\]
be an equally spaced partition of the unit interval, and let
\[
H_{m_1,\ldots,m_k}
=
(u_{m_1-1},u_{m_1}]
\times\cdots\times
(u_{m_k-1},u_{m_k}]
\]
denote the corresponding checkerboard cells.

The checkerboard copula density is
\begin{equation}
c_W(u_1,\ldots,u_k)
=
w_{m_1,\ldots,m_k},
\qquad
(u_1,\ldots,u_k)\in H_{m_1,\ldots,m_k},
\label{eq:checkerboard_density}
\end{equation}
where
\[
\mathbf W=(w_{m_1,\ldots,m_k})
\]
is the array of checkerboard weights satisfying the copula constraints. Details of the checkerboard copula and the validity conditions for \(\mathbf W\) are given in \ref{appen_checkboard_copula}. More general partitions of the unit interval are possible, but we consider only equally spaced partitions in this paper.

For each margin, the partition of the unit interval induces a partition of the support of the latent-intensity variable through its quantile function. For observation \(i\), define, for each \(m=1,\ldots,d\),
\begin{equation}
I_{ij}^{(m)}(n_{ij})
=
\int_{u_{m-1}}^{u_m}
\frac{
\exp\{-Q_j(u;\mu_{ij},\phi_j)\}
Q_j(u;\mu_{ij},\phi_j)^{n_{ij}}
}
{\Gamma(n_{ij}+1)}
\,du,
\label{eq:cell_probability}
\end{equation}
which depends only on the marginal mixing distribution. Computational details for evaluating \(I_{ij}^{(m)}(n_{ij})\) under the selected mixing distributions are given in \ref{app:cellprob}.

Substituting the checkerboard copula density (\ref{eq:checkerboard_density}) into (\ref{eq:jointpmf_uniform}) gives the following finite representation of the joint probability mass function.

\begin{proposition}
\label{prop_checkerboard_pmf}

Suppose the dependence among the latent-intensity variables is modeled by the checkerboard copula density
\(c_W\)
defined in (\ref{eq:checkerboard_density}).
Then, conditional on
\(\mathbf X_i=\mathbf x_i,\)
the joint probability mass function of
\(\mathbf N_i=(N_{i1},\ldots,N_{ik})^\top\)
is given by
\begin{equation}
\mathbb{P}(\mathbf N_i=\mathbf n_i
\mid
\mathbf X_i=\mathbf x_i)
=
\sum_{m_1=1}^{d}
\cdots
\sum_{m_k=1}^{d}
w_{m_1,\ldots,m_k}
\prod_{j=1}^{k}
I_{ij}^{(m_j)}(n_{ij}),
\label{eq:checkerboard_pmf}
\end{equation}
where
\(I_{ij}^{(m)}(n_{ij})\)
is defined in (\ref{eq:cell_probability}).

\end{proposition}

The proof is given in \ref{proof_of_pro2}.

Equation~(\ref{eq:checkerboard_pmf}) is the key result of the proposed checkerboard copula model. It expresses the joint probability mass function as a finite weighted sum over checkerboard cells. The checkerboard weights \(\mathbf W=(w_{m_1,\ldots,m_k})\) describe the dependence structure, while the quantities \(I_{ij}^{(m)}(n_{ij})\) depend only on the marginal mixing distributions. This representation avoids multidimensional numerical integration and leads to efficient likelihood evaluation.

\subsection{Construction of Checkerboard Copula}\label{sec_gau_fac_check_copu}

The checkerboard copula is fully determined by the checkerboard weights
\[
\mathbf W=(w_{m_1,\ldots,m_k}).
\]
We first consider the case where the number of responses is small.

Let
\[
C_{\Lambda}(\cdot;\boldsymbol{\theta}_{\Lambda})
\]
be a \(k\)-dimensional parametric copula with parameter vector
\(\boldsymbol{\theta}_{\Lambda}\).
The checkerboard weight for cell
\(H_{m_1,\ldots,m_k}\)
is
\begin{equation}
\begin{aligned}
w_{m_1,\ldots,m_k}
=
&\;
P\!\left(
\frac{m_1-1}{d}<U_1\le\frac{m_1}{d},
\ldots,
\frac{m_k-1}{d}<U_k\le\frac{m_k}{d}
\right)
\\
=
&\;
\sum_{\boldsymbol{\delta}\in\{0,1\}^k}
(-1)^{k-\sum_{j=1}^k\delta_j}
C_{\Lambda}
\!\left(
\frac{m_1-\delta_1}{d},
\ldots,
\frac{m_k-\delta_k}{d};
\boldsymbol{\theta}_{\Lambda}
\right),
\end{aligned}
\label{eq:checkerboard_weight_parametric}
\end{equation}
where
\((U_1,\ldots,U_k)\)
follows the copula
\(C_{\Lambda}(\cdot;\boldsymbol{\theta}_{\Lambda})\).

Thus, for low-dimensional problems, the checkerboard weights can be obtained directly from any parametric copula.

However, the direct construction becomes less practical as the dimension increases. For Gaussian and other elliptical copulas, estimating an unrestricted dependence matrix requires \(k(k-1)/2\) parameters, leading to high computational cost and a greater risk of overfitting. In contrast, most Archimedean copulas have only one or a few dependence parameters and are often not flexible enough for high-dimensional data.

To overcome these limitations, we construct the checkerboard weights using a factor copula model. The proposed framework is compatible with many factor copulas. In this paper, we use the Gaussian factor copula because of its simplicity and computational efficiency.

Let
\[
\mathbf U=(U_1,\ldots,U_k)^\top
\]
be a random vector with standard uniform margins, where
\[
U_j=\Phi(G_j),
\qquad
j=1,\ldots,k,
\]
and \(\Phi(\cdot)\) is the standard normal distribution function. The latent Gaussian variables follow a factor model with \(J\) latent factors,
\begin{equation}
\mathbf G
=
\mathbf\Psi\mathbf Z
+
\boldsymbol\varepsilon,
\label{eq:factor_model}
\end{equation}
where
\(\mathbf\Psi=(\psi_{j\ell})\in\mathbb{R}^{k\times J}\)
is the factor loading matrix,
\[
\mathbf Z
\sim
N_J(\mathbf0,\mathbf I_J),
\]
and
\[
\boldsymbol\varepsilon
\sim
N_k\!\left(
\mathbf0,
\operatorname{diag}(1-\sigma_1^2,\ldots,1-\sigma_k^2)
\right),
\]
with
\[
\sigma_j^2
=
\sum_{\ell=1}^{J}
\psi_{j\ell}^2
<1,
\qquad
j=1,\ldots,k.
\]
It follows that
\[
\mathbf G
\sim
N_k(\mathbf0,\mathbf\Sigma),
\]
where
\begin{equation}
\mathbf\Sigma
=
\mathbf\Psi\mathbf\Psi^\top
+
\operatorname{diag}(1-\sigma_1^2,\ldots,1-\sigma_k^2).
\label{eq:Rfactor}
\end{equation}

The factor model reduces the number of dependence parameters from \(k(k-1)/2\) to \(kJ\), where typically \(J\ll k\). Conditional on the latent factors \(\mathbf Z\), the variables \(U_1,\ldots,U_k\) are independent. This leads to a low-dimensional likelihood representation and substantially reduces the computational cost.

The Gaussian factor model is used only to parameterize the checkerboard weights. Other factor copulas can be used in the same way without changing the checkerboard likelihood formulation.

\begin{proposition}
\label{prop_gaussian_factor}

Suppose that the checkerboard weights are constructed from the Gaussian factor copula according to (\ref{eq:checkerboard_weight_parametric}). Then, conditional on
\(\mathbf X_i=\mathbf x_i\),
the joint probability mass function of
\(\mathbf N_i=(N_{i1},\ldots,N_{ik})^\top\)
is given by
\begin{equation}
\begin{aligned}
&
\mathbb{P}(\mathbf N_i=\mathbf n_i
\mid
\mathbf X_i=\mathbf x_i)
\\
& \qquad =
\frac{1}{(2\pi)^{J/2}}
\int_{\mathbb R^J}
\prod_{j=1}^{k}
p_{ij}^{N}(\mathbf z;n_{ij})
\exp\!\left(
-\frac12\mathbf z^\top\mathbf z
\right)
\,d\mathbf z,
\end{aligned}
\label{eq:factor_representation}
\end{equation}
where
\[
p_{ij}^{N}(\mathbf z;n)
=
\sum_{m=1}^{d}
I_{ij}^{(m)}(n)\,
\omega_{jm}(\mathbf z),
\]
with
\[
\omega_{jm}(\mathbf z)
=
q\!\left(
\frac{m}{d};
\boldsymbol{\psi}_j,
\mathbf z
\right)
-
q\!\left(
\frac{m-1}{d};
\boldsymbol{\psi}_j,
\mathbf z
\right),
\]
where
\[
q(u;\boldsymbol{\psi}_j,\mathbf z)
=
\Phi\!\left(
\frac{\Phi^{-1}(u)-\boldsymbol{\psi}_j^{\top}\mathbf z}
{\sqrt{1-\|\boldsymbol{\psi}_j\|^{2}}}
\right),
\]
with
\[
q(0;\boldsymbol{\psi}_j,\mathbf z)=0,
\qquad
q(1;\boldsymbol{\psi}_j,\mathbf z)=1.
\]

\end{proposition}

The proof is given in \ref{proof_of_pro_3}.

Proposition~\ref{prop_gaussian_factor} expresses the joint probability mass function as a \(J\)-dimensional Gaussian integral, where \(J\) is the number of latent factors. Since \(J\) is typically much smaller than \(k\), the likelihood can be evaluated efficiently using standard Gauss--Hermite quadrature. This result is used in the estimation procedure described in the next section.

\subsection{Likelihood for the Multivariate Zero-Inflated Mixed Poisson Model}

The previous subsections derived the joint probability mass function of the latent count vector under the checkerboard copula representation. We now incorporate the structural-zero component to obtain the joint probability mass function of the observed response vector.

\begin{proposition}
\label{prop_joint_pmf_zero_inflated}

Suppose that the latent-intensity variables are modeled by the checkerboard copula
\(C_W\)
defined in (\ref{eq:checkerboard_density}), and that the structural-zero indicators are modeled by the copula
\(C_{\mathbf B}\)
defined in Section~\ref{sec:depe_model}. Assume that
\(\mathbf B_i\)
and
\(\mathbf N_i\)
are conditionally independent given
\(\mathbf X_i=\mathbf x_i\).
Then, the conditional joint probability mass function of
\(\mathbf Y_i=(Y_{i1},\ldots,Y_{ik})^\top\)
is given by
\begin{equation}
\begin{aligned}
&
\mathbb P
\left(
\mathbf Y_i=\mathbf y_i
\mid
\mathbf X_i=\mathbf x_i
\right)
\\
&\qquad=
\sum_{m_1=1}^{d}
\cdots
\sum_{m_k=1}^{d}
w_{m_1,\ldots,m_k}
\sum_{\mathbf b\in\mathcal B(\mathbf y_i)}
p_{\mathbf B,i}(\mathbf b)
\prod_{j=1}^{k}
\widetilde I_{ij}^{(m_j)}(y_{ij},b_j),
\end{aligned}
\label{eq:joint_pmf_Y}
\end{equation}
where
\[
\widetilde I_{ij}^{(m)}(y,b)
=
b\,I_{ij}^{(m)}(y)
+
(1-b)\frac{1}{d},
\]
\[
\mathcal B(\mathbf y_i)
=
\left\{
\mathbf b\in\{0,1\}^k:
b_j=1
\text{ whenever }
y_{ij}>0
\right\},
\]
is the set of all structural-zero configurations that are compatible with the observed response vector
\(\mathbf y_i\).
If
\(y_{ij}>0\),
then necessarily
\(b_j=1\).
If
\(y_{ij}=0\),
both
\(b_j=0\)
(structural zero)
and
\(b_j=1\)
(sampling zero)
are possible.
Finally,
\[
p_{\mathbf B,i}(\mathbf b)
=
\mathbb P
\left(
\mathbf B_i=\mathbf b
\mid
\mathbf X_i=\mathbf x_i
\right),
\]
where
\(p_{\mathbf B,i}(\mathbf b)\)
is given by (\ref{eq:pmf_B}).

\end{proposition}

The proof is given in \ref{proof_of_pro_4}.

Equation~\eqref{eq:joint_pmf_Y} combines the latent-intensity model with the structural-zero component to obtain the joint probability mass function of the observed response vector. Evaluating the likelihood requires summation over all \(d^k\) checkerboard cells and all structural-zero configurations compatible with the observed response. Therefore, the computational cost increases rapidly with the number of responses \(k\). The direct likelihood is practical for low-dimensional problems. For moderate- and high-dimensional data, the Gaussian factor copula introduced in the previous subsection provides a much more efficient likelihood representation.

\begin{proposition}
\label{prop_ZI_likelihood}
Suppose that both the structural-zero copula and the checkerboard copula for the latent-intensity variables are parameterized by Gaussian factor models. Let
\[
\mathbf Z_B\sim N_{J_B}(\mathbf0,\mathbf I_{J_B}),
\qquad
\mathbf Z_{\Lambda}\sim N_{J_{\Lambda}}(\mathbf0,\mathbf I_{J_{\Lambda}})
\]
be the corresponding latent factor vectors with loading matrices
\(\boldsymbol{\Psi}_B\)
and
\(\boldsymbol{\Psi}_{\Lambda}\),
respectively, and assume that
\(\mathbf Z_B\)
and
\(\mathbf Z_{\Lambda}\)
are independent.

Then, conditional on
\(\mathbf X_i=\mathbf x_i\),
the joint probability mass function of
\(\mathbf Y_i=(Y_{i1},\ldots,Y_{ik})^\top\)
is
\begin{equation}
\begin{aligned}
&
\mathbb P
\left(
\mathbf Y_i=\mathbf y_i
\mid
\mathbf X_i=\mathbf x_i
\right)
\\
& \qquad=
\frac{1}{(2\pi)^{(J_B+J_{\Lambda})/2}}
\int_{\mathbb R^{J_B}}
\int_{\mathbb R^{J_{\Lambda}}}
\prod_{j=1}^{k}
L_{ij}(\mathbf z_B,\mathbf z_{\Lambda};y_{ij})
\\
&\qquad\qquad\times
\exp\!\left(
-\frac{1}{2}\mathbf z_B^\top\mathbf z_B
-\frac{1}{2}\mathbf z_{\Lambda}^\top\mathbf z_{\Lambda}
\right)
\,d\mathbf z_{\Lambda}\,d\mathbf z_B,
\end{aligned}
\label{eq:ZI_likelihood}
\end{equation}
where
\[
L_{ij}(\mathbf z_B,\mathbf z_{\Lambda};y)
=
\begin{cases}
1-p_{ij}^{B}(\mathbf z_B)
+
p_{ij}^{B}(\mathbf z_B)
p_{ij}^{N}(\mathbf z_{\Lambda};0),
&
y=0,\\[2ex]
p_{ij}^{B}(\mathbf z_B)
p_{ij}^{N}(\mathbf z_{\Lambda};y),
&
y>0,
\end{cases}
\]
with
\begin{equation*}
p_{ij}^{B}(\mathbf z_B)
=
\Phi\!\left(
\frac{
\Phi^{-1}(\pi_{ij})
-
\boldsymbol{\psi}_{B,j}^{\top}\mathbf z_B
}{
\sqrt{1-\|\boldsymbol{\psi}_{B,j}\|^2}
}
\right),
\label{eq:pB_factor}
\end{equation*}
and
\(p_{ij}^{N}(\mathbf z_{\Lambda};y)\)
is the conditional mixed Poisson probability defined in
Proposition~\ref{prop_gaussian_factor}, where the factor loading vector is
taken to be \(\boldsymbol{\psi}_{\Lambda,j}\), i.e., the rows of the loading
matrix \(\boldsymbol{\Psi}_{\Lambda}\).

\end{proposition}

A proof is provided in \ref{proof_of_pro_5}.

Proposition~\ref{prop_ZI_likelihood} expresses the joint probability mass function as a Gaussian integral over the latent factors of the two Gaussian factor copulas. The dimension of the integral is \(J_B+J_{\Lambda}\), which depends only on the numbers of latent factors rather than the number of responses \(k\). Since \(J_B\) and \(J_{\Lambda}\) are typically much smaller than \(k\), the likelihood can be evaluated efficiently using standard Gauss--Hermite quadrature.

\section{Parameter Estimation}

The proposed framework covers both standard mixed Poisson (MP) models and zero-inflated mixed Poisson (ZIMP) models in low- and high-dimensional settings. The corresponding likelihood representations used for parameter estimation are summarized in Table~\ref{tab:likelihood_summary}.

\begin{table}[H]
\centering
\caption{Likelihood used for parameter estimation}
\label{tab:likelihood_summary}
\renewcommand{\arraystretch}{1.2}
\begin{tabular}{lll}
\toprule
\textbf{Model} & \textbf{Dimension} (k) & \textbf{Likelihood} \\
\midrule
Multivariate MP
& Low
& Equation (\ref{eq:checkerboard_pmf}) \\

Multivariate MP
& High
& Equation (\ref{eq:factor_representation}) \\

Multivariate ZIMP
& Low
& Equation (\ref{eq:joint_pmf_Y}) \\

Multivariate ZIMP
& High
& Equation (\ref{eq:ZI_likelihood}) \\
\bottomrule
\end{tabular}
\end{table}

Although different likelihood representations are used for different model settings, the estimation procedure is the same. We first estimate the marginal model parameters and then estimate the dependence parameters. Throughout this paper, we use the inference-functions-for-margins (IFM) method of \citet{joe1996estimation} because it reduces the computational cost. When computationally feasible, the full likelihood can also be maximized directly.

\subsection{Inference Functions for Margins}\label{IMF}

The estimation proceeds in two stages.

\paragraph{Stage 1: Marginal estimation}

For each response \(j=1,\ldots,k\), the marginal model is estimated independently. The unknown parameters are the regression coefficients
\(\boldsymbol{\beta}_j\)
for the latent-intensity component,
\(\boldsymbol{\gamma}_j\)
for the structural-zero component, and the dispersion parameter
\(\phi_j\)
of the selected mixing distribution.

The marginal log-likelihood is
\begin{equation}
\ell_j(\boldsymbol{\beta}_j,\boldsymbol{\gamma}_j,\phi_j)
=
\sum_{i=1}^{n}
\log
\mathbb{P}
\left(
Y_{ij}=y_{ij}
\mid
\mathbf X_i=\mathbf x_i;
\boldsymbol{\beta}_j,
\boldsymbol{\gamma}_j,
\phi_j
\right),
\label{eq:marginalloglik}
\end{equation}
and the first-stage estimator is
\[
(\widehat{\boldsymbol{\beta}}_j,
\widehat{\boldsymbol{\gamma}}_j,
\widehat{\phi}_j)
=
\arg\max_{\boldsymbol{\beta}_j,\boldsymbol{\gamma}_j,\phi_j}
\ell_j(\boldsymbol{\beta}_j,\boldsymbol{\gamma}_j,\phi_j).
\]

This stage is the same for both the mixed Poisson and zero-inflated mixed Poisson models. The only difference is the form of the marginal probability mass function. In both models, the marginal probabilities and the checkerboard cell probabilities are evaluated using one-dimensional integrals. Depending on the selected mixing distribution, these integrals are computed using either the corresponding quantile function or the probability density function. Computational details for each mixing distribution are given in \ref{app:cellprob}.
\paragraph{Stage 2: Dependence estimation}
After the marginal model parameters have been estimated, the checkerboard cell probabilities
\(I_{ij}^{(m)}\)
are computed for every observation and every response variable. These quantities are then treated as fixed in the second stage.

The dependence parameters
\(\boldsymbol{\theta}_B\)
and
\(\boldsymbol{\theta}_{\Lambda}\)
are estimated by maximizing
\begin{equation}
\begin{aligned}
\ell(\boldsymbol{\theta}_B,\boldsymbol{\theta}_{\Lambda})=
\sum_{i=1}^{n}
\log
\mathbb P
\left(
\mathbf Y_i=\mathbf y_i
\mid
\mathbf X_i=\mathbf x_i;
\widehat{\boldsymbol{\beta}},
\widehat{\boldsymbol{\gamma}},
\widehat{\phi},
\boldsymbol{\theta}_B,
\boldsymbol{\theta}_{\Lambda}
\right),
\end{aligned}
\label{eq:IFMloglik}
\end{equation}
where the likelihood is given by Proposition~\ref{prop_ZI_likelihood}. The second-stage estimators are
\[
(\widehat{\boldsymbol{\theta}}_B,
\widehat{\boldsymbol{\theta}}_{\Lambda})
=
\arg\max_{\boldsymbol{\theta}_B,\boldsymbol{\theta}_{\Lambda}}
\ell(\boldsymbol{\theta}_B,\boldsymbol{\theta}_{\Lambda}).
\]

Only the dependence parameters are optimized in the second stage, while the marginal parameter estimates remain fixed. This reduces the computational cost because the marginal quantities are computed only once in the first stage and reused throughout the optimization.

Throughout this paper, we use the IFM method because of its computational efficiency. However, the proposed framework is not restricted to IFM. Since the full likelihood is available, the two estimation stages may be repeated iteratively until convergence, or all model parameters may be estimated simultaneously by direct maximum likelihood.

\subsection{Computational Aspects}
The computational cost of the proposed estimation procedure comes from two sources. In the first stage, some mixing distributions require one-dimensional numerical integration to evaluate the marginal probabilities and checkerboard cell probabilities. In the second stage, Gaussian factor copula models require multidimensional Gauss--Hermite quadrature over the latent factors in high-dimensional settings. Therefore, the computational cost mainly depends on the selected mixing distribution and the number of latent factors, rather than on whether IFM or full maximum likelihood estimation is used.

The first stage requires the evaluation of the one-dimensional integral
\begin{equation}
I(n)
=
\int_{0}^{1}
\frac{
\exp\{-Q(u;\mu,\phi)\}
Q(u;\mu,\phi)^{n}
}
{\Gamma(n+1)}
\,du,
\end{equation}
where \(Q(\cdot;\mu,\phi)\) is the quantile function of the mixing distribution \(\Lambda\). For the Gamma and inverse Gaussian mixing distributions, \(I(n)\) can be evaluated analytically. For the remaining mixing distributions, numerical integration is required, making the marginal estimation more computationally demanding.

After the marginal parameters have been estimated, the checkerboard cell probabilities in (\ref{eq:cell_probability}) are computed and used as inputs for the second-stage estimation of the dependence parameters. Closed-form expressions for the checkerboard cell probabilities are available only for the Gamma mixing distribution. For the remaining mixing distributions, these probabilities are also evaluated numerically. The computational methods for the five mixing distributions are summarized in Table~\ref{tab:appendix_computation}.

The second computational challenge arises from the evaluation of the likelihoods in (\ref{eq:factor_representation}) and (\ref{eq:ZI_likelihood}), which require numerical integration over the latent gaussian factors. When the number of latent factors is small, the likelihood can be evaluated efficiently using Gauss--Hermite quadrature. However, for high-dimensional data with complex dependence structures, a larger number of latent factors may be required, increasing the computational cost of likelihood evaluation.

As discussed in the previous section, the proposed framework is not limited to the Gaussian factor model. The Gaussian factor model is adopted in this paper because it provides a simple and computationally efficient likelihood representation. However, the framework naturally extends to other factor copula models, provided that the corresponding joint likelihood of \(\mathbf{Y}_i\) can be evaluated. The development of alternative factor copula models is left for future research.

The numerical methods used to evaluate the marginal probabilities, checkerboard cell probabilities, and Gaussian factor integrals are described in \ref{app:cellprob} and \ref{app:gaussian_checkerboard}. 

\section{Simulation Study}

This section evaluates the finite-sample performance of the proposed estimation procedure. The simulation study consists of a series of experiments designed to assess different aspects of the proposed methodology. Specifically, the experiments evaluate the following abilities:
\begin{enumerate}
\item The ability to identify the correct mixing distribution and recover the marginal model parameters, including the regression coefficients for both the latent-intensity and structural-zero components.

\item The ability to recover both the marginal model parameters and the dependence structure of the structural-zero and latent-intensity components in the bivariate case.

\item The ability to recover Gaussian factor dependence structures in high-dimensional settings.
\end{enumerate}

For multivariate models, all parameters are estimated using the IFM procedure described in Section~\ref{IMF}. In each experiment, the proposed method is fitted under several competing model specifications, and model performance is evaluated using the maximized log-likelihood.

\subsection{Recovery of Marginal Mixing Distributions}

We first evaluate the ability of the proposed marginal estimation procedure to identify the underlying mixing distribution.

Data are generated from mixed Poisson models with five mixing distributions: Gamma, Lognormal (LN), Pareto, Weibull, and Inverse Gaussian (IG). Each distribution is parameterized as described in \ref{app_distribution_for_lambda}, with the parameters chosen to achieve the specified mean and variance. The latent mean is fixed at
\[
\mathbb{E}(\Lambda)=1.5,
\]
and two levels of heterogeneity are considered,
\[
\mathbb{V}(\Lambda)=4
\quad\text{and}\quad
\mathbb{V}(\Lambda)=9.
\]

The corresponding parameter values are given in Table~\ref{tab:simulation_parameters}.

\begin{table}
\centering
\small
\renewcommand{\arraystretch}{1.2} 
\renewcommand{\arraystretch}{1.2} 
\begin{threeparttable}
    \caption{Simulation settings with \(\boldsymbol{\mathbb E(\Lambda)=1.5}\)}
    \label{tab:simulation_parameters}
    \begin{tabular}{lcc}
    \toprule
    \textbf{Distribution of $\Lambda$} & \textbf{\(\phi\) for \(\mathbb V(\Lambda)=4\)} & \textbf{\(\phi\) for \(\mathbb V(\Lambda)=9\)} \\
    \midrule
    Gamma            & 1.778 & 4.000 \\ 
    LN               & 1.748 & 2.556 \\ 
    Pareto           & 0.389 & 1.500 \\ 
    Weibull          & 1.300 & 2.030 \\ 
    IG               & 1.185 & 2.667 \\ 
    \bottomrule
    \end{tabular}
\end{threeparttable}
\end{table}

For each scenario, a sample of size $n=10,000$ is generated according to \(N_i|\Lambda_i \sim \mathrm{Poisson}(\Lambda_i).\)

Each simulated dataset is subsequently fitted using all five candidate mixing distributions. Consequently, each dataset is analyzed under one correctly specified model and four misspecified alternatives.

\begin{table}
\centering
\caption{Recovery of marginal mixing distributions under lower latent heterogeneity}
\label{tab:sim_result_var4}

\scriptsize
\resizebox{\columnwidth}{!}{%
\begin{threeparttable}
\renewcommand{\arraystretch}{1.5}
\setlength{\tabcolsep}{3pt}

\begin{tabular}{ll ccccc}
\toprule
\textbf{Data-generating} & \textbf{Statistic} & \textbf{Fitted} & \textbf{Fitted} & \textbf{Fitted} & \textbf{Fitted} & \textbf{Fitted} \\
\textbf{distribution (\(\Lambda\))} & & \textbf{Gamma} & \textbf{LN} & \textbf{Pareto} & \textbf{Weibull} & \textbf{IG} \\
\midrule

\makecell[l]{Gamma}
& LogLik & \(\boldsymbol{-16,556}\) & \(-16,683\) & \(-16,638\) & \(-16,567\) & \(-16,673\) \\
\(\mu=1.500\)
& \(\hat\mu\) & \(\boldsymbol{1.498}\) & \(1.587\) & \(1.521\) & \(1.501\) & \(1.498\) \\
\(\phi=1.778\)
& \(\hat\phi\) & \(\boldsymbol{1.783}\) & \(2.444\) & \(0.754\) & \(1.359\) & \(1.594\) \\
\midrule

\makecell[l]{LN}
& LogLik & \(-16,829\) & \(\boldsymbol{-16,733}\) & \(-16,761\) & \(-16,809\) & \(-16,736\) \\
\(\mu=1.500\)
& \(\hat\mu\) & \(1.505\) & \(\boldsymbol{1.507}\) & \(1.506\) & \(1.503\) & \(1.505\) \\
\(\phi=1.748\)
& \(\hat\phi\) & \(1.120\) & \(\boldsymbol{1.733}\) & \(0.336\) & \(1.137\) & \(1.009\) \\
\midrule

\makecell[l]{Pareto}
& LogLik & \(-16,676\) & \(-16,648\) & \(\boldsymbol{-16,636}\) & \(-16,658\) & \(-16,642\) \\
\(\mu=1.500\)
& \(\hat\mu\) & \(1.484\) & \(1.508\) & \(\boldsymbol{1.484}\) & \(1.482\) & \(1.484\) \\
\(\phi=0.389\)
& \(\hat\phi\) & \(1.353\) & \(1.917\) & \(\boldsymbol{0.328}\) & \(1.197\) & \(1.166\) \\
\midrule

\makecell[l]{Weibull}
& LogLik & \(-16,738\) & \(-16,806\) & \(-16,773\) & \(\boldsymbol{-16,734}\) & \(-16,793\) \\
\(\mu=1.500\)
& \(\hat\mu\) & \(1.519\) & \(1.586\) & \(1.519\) & \(\boldsymbol{1.519}\) & \(1.519\) \\
\(\phi=1.300\)
& \(\hat\phi\) & \(1.626\) & \(2.262\) & \(0.361\) & \(\boldsymbol{1.303}\) & \(1.420\) \\
\midrule

\makecell[l]{IG}
& LogLik & \(-16,895\) & \(-16,821\) & \(-16,834\) & \(-16,870\) & \(\boldsymbol{-16,812}\) \\
\(\mu=1.500\)
& \(\hat\mu\) & \(1.527\) & \(1.542\) & \(1.527\) & \(1.524\) & \(\boldsymbol{1.527}\) \\
\(\phi=1.185\)
& \(\hat\phi\) & \(1.348\) & \(1.921\) & \(0.358\) & \(1.203\) & \(\boldsymbol{1.156}\) \\
\bottomrule

\end{tabular}
\begin{tablenotes}[flushleft]
\item \textit{Note:} Simulation results are based on $\mathbb E(\Lambda)=1.5$ and $\mathbb V(\Lambda)=4$. Bold values indicate the fitted model with the largest log-likelihood for each data-generating distribution. The estimates $\hat{\mu}$ and $\hat{\phi}$ denote the maximum likelihood estimates under each fitted mixing distribution.
\end{tablenotes}

\end{threeparttable}
}
\end{table}

\begin{table}
\centering
\caption{Recovery of marginal mixing distributions under higher latent heterogeneity}
\label{tab:sim_result_var9}

\scriptsize
\resizebox{\columnwidth}{!}{%
\begin{threeparttable}
\renewcommand{\arraystretch}{1.5}
\setlength{\tabcolsep}{3pt}

\begin{tabular}{ll ccccc}
\toprule
\textbf{Data-generating} & \textbf{Statistic} & \textbf{Fitted} & \textbf{Fitted} & \textbf{Fitted} & \textbf{Fitted} & \textbf{Fitted} \\
\textbf{distribution (\(\Lambda\))} & & \textbf{Gamma} & \textbf{LN} & \textbf{Pareto} & \textbf{Weibull} & \textbf{IG} \\
\midrule

\makecell[l]{Gamma}
& LogLik & \(\boldsymbol{-15,068}\) & \(-15,243\) & \(-15,730\) & \(-15,102\) & \(-15,245\) \\
\(\mu=1.500\)
& \(\hat\mu\) & \(\boldsymbol{1.519}\) & \(1.823\) & \(1.519\) & \(1.567\) & \(1.519\) \\
\(\phi=4.000\)
& \(\hat\phi\) & \(\boldsymbol{4.108}\) & \(5.151\) & \(1.638\) & \(2.139\) & \(4.960\) \\
\midrule

\makecell[l]{LN}
& LogLik & \(-16,427\) & \(\boldsymbol{-16,259}\) & \(-16,279\) & \(-16,328\) & \(-16,267\) \\
\(\mu=1.500\)
& \(\hat\mu\) & \(1.516\) & \(\boldsymbol{1.507}\) & \(1.516\) & \(1.497\) & \(1.516\) \\
\(\phi=2.556\)
& \(\hat\phi\) & \(2.086\) & \(\boldsymbol{2.591}\) & \(1.259\) & \(1.515\) & \(2.009\) \\
\midrule

\makecell[l]{Pareto}
& LogLik & \(-16,317\) & \(-16,183\) & \(\boldsymbol{-16,181}\) & \(-16,237\) & \(-16,194\) \\
\(\mu=1.500\)
& \(\hat\mu\) & \(1.469\) & \(1.498\) & \(\boldsymbol{1.469}\) & \(1.454\) & \(1.469\) \\
\(\phi=1.500\)
& \(\hat\phi\) & \(1.907\) & \(2.490\) & \(\boldsymbol{1.202}\) & \(1.444\) & \(1.827\) \\
\midrule

\makecell[l]{Weibull}
& LogLik & \(-15,730\) & \(-15,722\) & \(-15,936\) & \(\boldsymbol{-15,675}\) & \(-15,723\) \\
\(\mu=1.500\)
& \(\hat\mu\) & \(1.524\) & \(1.690\) & \(1.524\) & \(\boldsymbol{1.519}\) & \(1.524\) \\
\(\phi=2.030\)
& \(\hat\phi\) & \(3.134\) & \(3.968\) & \(1.589\) & \(\boldsymbol{1.855}\) & \(3.424\) \\
\midrule

\makecell[l]{IG}
& LogLik & \(-16,130\) & \(-15,964\) & \(-16,057\) & \(-16,015\) & \(\boldsymbol{-15,949}\) \\
\(\mu=1.500\)
& \(\hat\mu\) & \(1.524\) & \(1.514\) & \(1.525\) & \(1.500\) & \(\boldsymbol{1.524}\) \\
\(\phi=2.667\)
& \(\hat\phi\) & \(2.546\) & \(3.046\) & \(1.486\) & \(1.680\) & \(\boldsymbol{2.636}\) \\
\bottomrule

\end{tabular}

\begin{tablenotes}[flushleft]
\item \textit{Note:} Simulation results are based on $\mathbb E(\Lambda)=1.5$ and $\mathbb V(\Lambda)=9$. Bold values indicate the fitted model with the largest log-likelihood for each data-generating distribution. The estimates $\hat{\mu}$ and $\hat{\phi}$ denote the maximum likelihood estimates under each fitted mixing distribution.
\end{tablenotes}

\end{threeparttable}
}
\end{table}

Tables~\ref{tab:sim_result_var4} and~\ref{tab:sim_result_var9} present the estimation results. In every scenario, the true mixing distribution gives the highest log-likelihood. The estimates of \(\mu\) and \(\phi\) are close to their true values, indicating that the proposed method accurately recovers both the mixing distribution and its parameters.

As the heterogeneity increases from \(\mathbb{V}(\Lambda)=4\) to \(\mathbb{V}(\Lambda)=9\), the differences among competing models become more pronounced. Misspecified models produce lower log-likelihood values and larger parameter bias, making the true mixing distribution easier to identify.

We next evaluate the proposed estimation procedure in a regression setting. The conditional mean of the latent intensity is modeled as
\[
\log(\mu_i)=\mathbf{x}_i^\top\boldsymbol{\beta},
\]
where
\[
\boldsymbol{\beta}=(0.5,\,0.3,\,0.7)^\top.
\]
The covariate vector
\(\mathbf{x}_i=(1,x_{i1},x_{i2})^\top\)
contains an intercept and two independent covariates, where
\[
x_{i1},\,x_{i2}\overset{\text{i.i.d.}}{\sim}N(0,1).
\]

Tables~\ref{tab:mp_sim_result_var4} and~\ref{tab:mp_sim_result_var9} report the log-likelihood values obtained under each fitted model. In particular, the dispersion parameters used in Table~\ref{tab:mp_sim_result_var4} are identical to those in Table~\ref{tab:sim_result_var4}, whereas those used in Table~\ref{tab:mp_sim_result_var9} are identical to the dispersion parameters in Table~\ref{tab:sim_result_var9}. Table~\ref{tab:mp_parameter_recovery} summarizes the parameter estimates under the correctly specified models.

\begin{table}[h]
\centering
\caption{Mixed Poisson model under lower latent heterogeneity}
\label{tab:mp_sim_result_var4}

\scriptsize
\resizebox{\columnwidth}{!}{%
\begin{threeparttable}
\renewcommand{\arraystretch}{1.5}
\setlength{\tabcolsep}{3pt}

\begin{tabular}{lccccc}
\toprule
\textbf{Data-generating} &  \textbf{Fitted} & \textbf{Fitted} & \textbf{Fitted} & \textbf{Fitted} & \textbf{Fitted} \\
\textbf{distribution (\(\Lambda\))} & \textbf{Gamma} & \textbf{LN} & \textbf{Pareto} & \textbf{Weibull} & \textbf{IG} \\
\midrule
Gamma            & \(\mathbf{-17,670}\) & \(-17,845\) & \(-17,796\) & \(-17,684\) & \(-17,943\) \\
LN        & \(-18,009\) & \(\mathbf{-17,867}\) & \(-17,913\) & \(-17,981\) & \(-18,080\) \\
Pareto           & \(-17,878\) & \(-17,874\) & \(\mathbf{-17,838}\) & \(-17,853\) & \(-18,020\) \\
Weibull          & \(-17,880\) & \(-17,962\) & \(-17,924\) & \(\mathbf{-17,870}\) & \(-18,094\) \\
IG               & \(-17,614\) & \(-17,437\) & \(-17,475\) & \(-17,516\) & \(\mathbf{-17,230}\) \\
\bottomrule
\end{tabular}

\begin{tablenotes}[flushleft]
\item \textit{Note:} $\boldsymbol{\beta}=(0.5,\,0.3,\,0.7)^\top$ and the $\boldsymbol{\phi}$ values reported in Table~\ref{tab:simulation_parameters}. Bold values indicate the largest log-likelihood for each data-generating distribution.
\end{tablenotes}

\end{threeparttable}
}
\end{table}

\begin{table}
\centering
\caption{Mixed Poisson model under higher latent heterogeneity}
\label{tab:mp_sim_result_var9}

\scriptsize
\resizebox{\columnwidth}{!}{%
\begin{threeparttable}
\renewcommand{\arraystretch}{1.5}
\setlength{\tabcolsep}{3pt}

\begin{tabular}{lccccc}
\toprule
\textbf{Data-generating} &  \textbf{Fitted} & \textbf{Fitted} & \textbf{Fitted} & \textbf{Fitted} & \textbf{Fitted} \\
\textbf{distribution (\(\Lambda\))} & \textbf{Gamma} & \textbf{LN} & \textbf{Pareto} & \textbf{Weibull} & \textbf{IG} \\
\midrule
Gamma            & \(\mathbf{-15,761}\) & \(-16,040\) & \(-16,581\) & \(-15,830\) & \(-16,116\) \\
LN        & \(-17,300\) & \(\mathbf{-17,140}\) & \(-17,145\) & \(-17,210\) & \(-17,321\) \\
Pareto           & \(-17,379\) & \(-17,307\) & \(\mathbf{-17,283}\) & \(-17,318\) & \(-17,451\) \\
Weibull          & \(-16,412\) & \(-16,525\) & \(-16,717\) & \(\mathbf{-16,392}\) & \(-16,595\) \\
IG               & \(-16,224\) & \(-16,017\) & \(-16,208\) & \(-16,044\) & \(\mathbf{-15,782}\) \\
\bottomrule
\end{tabular}

\begin{tablenotes}[flushleft]
\item \textit{Note:} $\boldsymbol{\beta}=(0.5,\,0.3,\,0.7)^\top$ and the $\boldsymbol{\phi}$ values reported in Table~\ref{tab:simulation_parameters}. Bold values indicate the largest log-likelihood for each data-generating distribution.
\end{tablenotes}

\end{threeparttable}
}
\end{table}

Tables~\ref{tab:mp_sim_result_var4} and~\ref{tab:mp_sim_result_var9} demonstrate that, under both levels of latent heterogeneity, the correctly specified mixing distribution consistently achieves the highest log-likelihood among all competing models. Moreover, the differences in log-likelihood become more pronounced as the latent heterogeneity increases, indicating that greater overdispersion improves the identifiability of the underlying mixing distribution.

\begin{table}
\centering
\caption{\textbf{\boldmath Parameter recovery under correctly specified mixed Poisson model}}
\label{tab:mp_parameter_recovery}

\scriptsize
\resizebox{\columnwidth}{!}{%
\begin{threeparttable}
\renewcommand{\arraystretch}{1.5}
\setlength{\tabcolsep}{3pt}

\begin{tabular}{l ccc c ccc}
\toprule
&
\multicolumn{3}{c}{\textbf{Lower levels of heterogeneity}}
&&
\multicolumn{3}{c}{\textbf{Higher levels of heterogeneity}}
\\
\cmidrule{2-4}\cmidrule{6-8}
\textbf{Distribution}
&
\textbf{\(\hat{\beta}\)}
&
\(\phi\)
&
\textbf{\(\hat{\phi}\)}
&&
\textbf{\(\hat{\beta}\)}
&
\(\phi\)
&
\textbf{\(\hat{\phi}\)}
\\
\midrule

Gamma
&
\((0.50,0.29,0.70)\)
&
1.78
&
1.81
&&
\((0.46,0.25,0.70)\)
&
4.00
&
3.95
\\

LN
&
\((0.51,0.28,0.68)\)
&
1.75
&
1.79
&&
\((0.46,0.30,0.71)\)
&
2.56
&
2.45
\\

Pareto
&
\((0.51,0.28,0.69)\)
&
0.39
&
0.39
&&
\((0.47,0.23,0.70)\)
&
1.50
&
0.83
\\

Weibull
&
\((0.52,0.30,0.68)\)
&
1.32
&
1.34
&&
\((0.46,0.25,0.71)\)
&
1.84
&
1.82
\\

IG
&
\((0.51,0.31,0.70)\)
&
1.19
&
1.25
&&
\((0.46,0.27,0.71)\)
&
2.67
&
2.66
\\

\bottomrule
\end{tabular}

\begin{tablenotes}[flushleft]
\item \textit{Note:} The true regression coefficient vector is
$\boldsymbol{\beta}=(0.5,0.3,0.7)^\top$.
\end{tablenotes}

\end{threeparttable}
}

\end{table}

Tables~\ref{tab:mp_parameter_recovery} summarize the parameter recovery results. Under the correctly specified model, the proposed method accurately recovers both the regression coefficients and the dispersion parameter. Across all five mixing distributions, the estimated regression coefficients are close to their true values.

The main exception is the Pareto mixing distribution under a higher level of latent heterogeneity, where the estimate of the dispersion parameter \(\phi\) is less accurate. Nevertheless, the regression coefficients remain well estimated.

Tables~\ref{tab:zimp_sim_result_var4} and~\ref{tab:zimp_sim_result_var9} present the results for the proposed zero-inflated mixed Poisson framework. The structural-zero component is modeled using the same covariates \(x_{i1}\) and \(x_{i2}\), with regression coefficients
\[
\boldsymbol{\gamma}=(0.6,\,0.4)^\top.
\]
Across all simulation settings, the correctly specified mixing distribution achieves the highest log-likelihood among all candidate models, indicating that the introduction of the structural-zero component does not affect the identification of the underlying mixing distribution.

\begin{table}[htbp]
\centering
\caption{\textbf{\boldmath Zero-inflated mixed Poisson model under lower latent heterogeneity}}
\label{tab:zimp_sim_result_var4}

\scriptsize
\resizebox{\columnwidth}{!}{%
\begin{threeparttable}
\renewcommand{\arraystretch}{1.5}
\setlength{\tabcolsep}{3pt}

\begin{tabular}{lccccc}
\toprule
\textbf{Data-generating}
&
\textbf{Gamma}
&
\textbf{LN}
&
\textbf{Pareto}
&
\textbf{Weibull}
&
\textbf{IG}
\\
\midrule

Gamma
&
\(\mathbf{-11,185}\)
&
$-11,252$
&
$-11,223$
&
$-11,190$
&
$-11,303$
\\

LN
&
$-11,639$
&
\(\mathbf{-11,582}\)
&
$-11,592$
&
$-11,619$
&
$-11,647$
\\

Pareto
&
$-11,402$
&
$-11,396$
&
\(\mathbf{-11,385}\)
&
$-11,391$
&
$-11,464$
\\

Weibull
&
$-11,409$
&
$-11,436$
&
$-11,420$
&
\(\mathbf{-11,401}\)
&
$-11,492$
\\

IG
&
$-11,251$
&
$-11,159$
&
$-11,154$
&
$-11,196$
&
\(\mathbf{-11,071}\)
\\

\bottomrule
\end{tabular}

\begin{tablenotes}[flushleft]
\item \textit{Note:} The true regression coefficient vectors are
\(\boldsymbol{\beta}=(0.5,\,0.3,\,0.7)^\top\)
and
\(\boldsymbol{\gamma}=(0.6,\,0.4)^\top\).
The values of \(\phi\) are taken from Table~\ref{tab:simulation_parameters} and correspond to
\(\mathbb{V}(\Lambda)=4\).
\end{tablenotes}

\end{threeparttable}
}
\end{table}

\begin{table}[htbp]
\centering
\caption{\textbf{\boldmath Zero-inflated mixed Poisson model under higher latent heterogeneity
}}
\label{tab:zimp_sim_result_var9}

\scriptsize
\resizebox{\columnwidth}{!}{%
\begin{threeparttable}
\renewcommand{\arraystretch}{1.5}
\setlength{\tabcolsep}{3pt}

\begin{tabular}{lccccc}
\toprule
\textbf{Data-generating}
&
\textbf{Gamma}
&
\textbf{LN}
&
\textbf{Pareto}
&
\textbf{Weibull}
&
\textbf{IG}
\\
\midrule

Gamma
&
\(\mathbf{-9,481}\)
&
$-9,619$
&
$-9,791$
&
$-9,516$
&
$-9,678$
\\

LN
&
$-10,848$
&
\(\mathbf{-10,783}\)
&
$-10,790$
&
$-10,807$
&
$-10,857$
\\

Pareto
&
$-10,778$
&
$-10,727$
&
\(\mathbf{-10,721}\)
&
$-10,741$
&
$-10,800$
\\

Weibull
&
$-10,056$
&
$-10,094$
&
$-10,164$
&
\(\mathbf{-10,046}\)
&
$-10,138$
\\

IG
&
$-10,014$
&
$-9,955$
&
$-9,904$
&
$-9,916$
&
\(\mathbf{-9,781}\)
\\

\bottomrule
\end{tabular}

\begin{tablenotes}[flushleft]
\item \textit{Note:} The true regression coefficient vectors are
\(\boldsymbol{\beta}=(0.5,\,0.3,\,0.7)^\top\)
and
\(\boldsymbol{\gamma}=(0.6,\,0.4)^\top\).
The values of \(\phi\) are taken from Table~\ref{tab:simulation_parameters} and correspond to
\(\mathbb{V}(\Lambda)=9\).
\end{tablenotes}

\end{threeparttable}
}
\end{table}

Table~\ref{tab:zimp_parameter_recovery} further shows that the regression coefficients and dispersion parameters are accurately recovered under the correctly specified model. The estimated regression coefficients remain close to their true values for all five mixing distributions under both levels of latent heterogeneity, while the estimated dispersion parameters also exhibit only modest bias. Similar to the mixed Poisson regression model, the Pareto mixing distribution under higher heterogeneity shows somewhat larger bias in the dispersion parameter, reflecting the relatively weak identifiability of this parameter rather than a failure of the proposed estimation procedure.

\begin{table}[htbp]
\centering
\caption{\textbf{\boldmath Parameter recovery under correctly specified zero-inflated mixed Poisson models}}
\label{tab:zimp_parameter_recovery}

\scriptsize
\resizebox{\columnwidth}{!}{%
\begin{threeparttable}
\renewcommand{\arraystretch}{1.4}
\setlength{\tabcolsep}{3pt}

\begin{tabular}{l ccc c ccc}
\toprule

&
\multicolumn{3}{c}{\textbf{Lower level of heterogeneity}}
&&
\multicolumn{3}{c}{\textbf{Higher level of heterogeneity}}
\\

\cmidrule{2-4}\cmidrule{6-8}

\textbf{Distribution}
&
\textbf{\(\hat{\beta} / \hat{\gamma}\)}
&
\(\phi\)
&
\(\hat{\phi}\)

&&

\textbf{\(\hat{\beta} / \hat{\gamma}\)}
&
\(\phi\)
&
\(\hat{\phi}\)

\\

\midrule

Gamma
&
$(0.53,0.31,0.68)$
&
1.778
&
1.758

&&

$(0.47,0.28,0.76)$
&
4.000
&
3.866
\\
&
$(0.58,0.35)$
&
&
&&
$(0.58,0.47)$
&
&
\\

\midrule

LN
&
$(0.52,0.30,0.65)$
&
1.748
&
1.889

&&

$(0.46,0.28,0.70)$
&
2.556
&
2.428
\\
&
$(0.56,0.37)$
&
&
&&
$(0.66,0.39)$
&
&
\\

\midrule

Pareto
&
$(0.50,0.29,0.71)$
&
0.389
&
0.360

&&

$(0.45,0.24,0.71)$
&
1.500
&
1.033
\\
&
$(0.63,0.36)$
&
&
&&
$(0.61,0.41)$
&
&
\\

\midrule

Weibull
&
$(0.55,0.30,0.68)$
&
1.317
&
1.347

&&

$(0.49,0.25,0.71)$
&
1.843
&
1.811
\\
&
$(0.58,0.41)$
&
&
&&
$(0.61,0.40)$
&
&
\\

\midrule

IG
&
$(0.54,0.32,0.68)$
&
1.185
&
1.310

&&

$(0.47,0.26,0.75)$
&
2.667
&
2.575
\\
&
$(0.59,0.31)$
&
&
&&
$(0.57,0.37)$
&
&
\\

\bottomrule
\end{tabular}

\begin{tablenotes}[flushleft]
\item \textit{Note:} The true regression coefficient vectors are
\(\boldsymbol{\beta}=(0.5,\,0.3,\,0.7)^\top\)
and
\(\boldsymbol{\gamma}=(0.6,\,0.4)^\top\).
\end{tablenotes}

\end{threeparttable}
}
\end{table}

Overall, the simulation results show that the proposed marginal estimation procedure accurately identifies the underlying mixing distribution in both the mixed Poisson and zero-inflated mixed Poisson models. Under the correctly specified model, the regression and dispersion parameters are also estimated accurately. These results demonstrate the reliability of the proposed likelihood-based marginal estimation framework.

\subsection{Recovery of Pairwise Dependence Structures}

We next evaluate the finite-sample performance of the proposed dependence estimation procedure. We first consider the standard bivariate mixed Poisson model to assess the recovery of the dependence structure without structural zeros. The effect of zero inflation is investigated in the following experiment.

The first margin follows a Poisson--Gamma distribution and the second margin follows a Poisson--Lognormal distribution. The covariates and regression coefficients are the same as those used in the previous simulation study. The dispersion parameters are taken from Table~\ref{tab:simulation_parameters} and correspond to \(\mathbb{V}(\Lambda)=4.\)

Dependence between the latent-intensity variables is generated from several commonly used copula families, including the Gaussian, Student-\(t\), Clayton, Frank, Gumbel, and their rotated counterparts. The dependence parameters are chosen so that Kendall's tau is either \(0.5\) or \(-0.5\). For each copula, the corresponding checkerboard copula is constructed using \(d=10\) intervals on each margin.

For each simulated dataset, all candidate copula models are fitted using the proposed estimation procedure. Model recovery is evaluated by comparing the maximized log-likelihood values and the estimated dependence parameters. A copula family is regarded as correctly identified if the true data-generating copula gives the highest log-likelihood among all candidate models.

\begin{table}[!t]
\centering
\caption{Copula parameter recovery under positive dependence
\(\boldsymbol{(\tau=0.5)}\)}
\label{tab:copula_recovery}

\scriptsize
\resizebox{\columnwidth}{!}{%
\begin{threeparttable}
\renewcommand{\arraystretch}{1.5}
\setlength{\tabcolsep}{3pt}

\begin{tabular}{@{}lccccccc@{}}
\toprule
\textbf{True}
& \textbf{Gauss.}
& \textbf{Student}
& \textbf{Clayton}
& \textbf{Gumbel}
& \textbf{Frank}
& \textbf{Clay180}
& \textbf{Gum180} \\
\midrule

\makecell[l]{Gaussian\\(0.707)}
&
\makecell{$\boldsymbol{-34,517}$\\\textit{(0.708)}}
&
\makecell{$-34,533$\\\textit{(0.712)}}
&
\makecell{$-34,582$\\\textit{(2.857)}}
&
\makecell{$-34,534$\\\textit{(1.853)}}
&
\makecell{$-34,531$\\\textit{(5.891)}}
&
\makecell{$-34,564$\\\textit{(1.393)}}
&
\makecell{$-34,530$\\\textit{(2.188)}}
\\[2.0ex]

\makecell[l]{Student\\(0.707)}
&
\makecell{$-34,795$\\\textit{(0.699)}}
&
\makecell{$\boldsymbol{-34,784}$\\\textit{(0.707)}}
&
\makecell{$-34,851$\\\textit{(2.853)}}
&
\makecell{$-34,799$\\\textit{(1.825)}}
&
\makecell{$-34,800$\\\textit{(5.790)}}
&
\makecell{$-34,829$\\\textit{(1.347)}}
&
\makecell{$-34,799$\\\textit{(2.168)}}
\\[2.0ex]

\makecell[l]{Clayton\\(2.000)}
&
\makecell{$-34,787$\\\textit{(0.596)}}
&
\makecell{$-34,817$\\\textit{(0.604)}}
&
\makecell{$\boldsymbol{-34,745}$\\\textit{(1.947)}}
&
\makecell{$-34,852$\\\textit{(1.568)}}
&
\makecell{$-34,772$\\\textit{(4.459)}}
&
\makecell{$-34,903$\\\textit{(0.891)}}
&
\makecell{$-34,759$\\\textit{(1.814)}}
\\[2.0ex]

\makecell[l]{Gumbel\\(2.000)}
&
\makecell{$-34,513$\\\textit{(0.757)}}
&
\makecell{$-34,510$\\\textit{(0.764)}}
&
\makecell{$-34,630$\\\textit{(3.683)}}
&
\makecell{$\boldsymbol{-34,496}$\\\textit{(2.033)}}
&
\makecell{$-34,544$\\\textit{(6.812)}}
&
\makecell{$-34,504$\\\textit{(1.721)}}
&
\makecell{$-34,542$\\\textit{(2.457)}}
\\[2.0ex]

\makecell[l]{Frank\\(5.736)}
&
\makecell{$-34,713$\\\textit{(0.680)}}
&
\makecell{$-34,737$\\\textit{(0.689)}}
&
\makecell{$-34,728$\\\textit{(2.684)}}
&
\makecell{$-34,753$\\\textit{(1.771)}}
&
\makecell{$\boldsymbol{-34,702}$\\\textit{(5.555)}}
&
\makecell{$-34,793$\\\textit{(1.241)}}
&
\makecell{$-34,709$\\\textit{(2.087)}}
\\[2.0ex]

\makecell[l]{Clay180\\(2.000)}
&
\makecell{$-34,619$\\\textit{(0.786)}}
&
\makecell{$-34,615$\\\textit{(0.792)}}
&
\makecell{$-34,768$\\\textit{(4.401)}}
&
\makecell{$-34,588$\\\textit{(2.156)}}
&
\makecell{$-34,659$\\\textit{(7.492)}}
&
\makecell{$\boldsymbol{-34,584}$\\\textit{(1.935)}}
&
\makecell{$-34,661$\\\textit{(2.658)}}
\\[2.0ex]

\makecell[l]{Gumbel 180\\(2.000)}
&
\makecell{$-34,811$\\\textit{(0.643)}}
&
\makecell{$-34,817$\\\textit{(0.651)}}
&
\makecell{$-34,828$\\\textit{(2.236)}}
&
\makecell{$-34,839$\\\textit{(1.680)}}
&
\makecell{$-34,809$\\\textit{(4.965)}}
&
\makecell{$-34,876$\\\textit{(1.102)}}
&
\makecell{$\boldsymbol{-34,803}$\\\textit{(1.948)}}
\\

\bottomrule
\end{tabular}

\begin{tablenotes}[flushleft]
\item \textit{Note:} Values in the first column are the true dependent parameters. The largest log-likelihood in each row is shown in bold. Clay180 and Gum180 denote the \(180^\circ\)-rotated Clayton and Gumbel copulas, respectively.
\end{tablenotes}

\end{threeparttable}
}

\end{table}

\begin{table}[htbp]
\centering

\caption{Copula parameter recovery under negative dependence
\(\boldsymbol{(\tau=-0.5)}\)}
\label{tab:copula_recovery_tau_neg05}

\scriptsize
\resizebox{\columnwidth}{!}{%
\begin{threeparttable}
\renewcommand{\arraystretch}{1.5}
\setlength{\tabcolsep}{3pt}

\begin{tabular}{@{}lccccccc@{}}
\toprule

\textbf{True}
& \textbf{Gauss.}
& \textbf{Student}
& \textbf{Clay90}
& \textbf{Gum90}
& \textbf{Frank}
& \textbf{Clay270}
& \textbf{Gum270} \\

\midrule
\makecell[l]{Gaussian\\(-0.707)}
&
\makecell{$\boldsymbol{-34,632}$\\\textit{(-0.726)}}
&
\makecell{$-34,646$\\\textit{(-0.754)}}
&
\makecell{$-34,681$\\\textit{(2.238)}}
&
\makecell{$-34,647$\\\textit{(2.089)}}
&
\makecell{$-34,641$\\\textit{(-6.111)}}
&
\makecell{$-34,689$\\\textit{(1.963)}}
&
\makecell{$-34,642$\\\textit{(2.149)}}
\\[2.0ex]

\makecell[l]{Student\\(-0.707)}
&
\makecell{$-34,676$\\\textit{(-0.669)}}
&
\makecell{$\boldsymbol{-34,650}$\\\textit{(-0.709)}}
&
\makecell{$-34,730$\\\textit{(1.770)}}
&
\makecell{$-34,666$\\\textit{(1.906)}}
&
\makecell{$-34,670$\\\textit{(-5.310)}}
&
\makecell{$-34,703$\\\textit{(1.632)}}
&
\makecell{$-34,679$\\\textit{(1.945)}}
\\[2.0ex]

\makecell[l]{Clayton 90\\(2.000)}
&
\makecell{$-34,782$\\\textit{(-0.652)}}
&
\makecell{$-34,797$\\\textit{(-0.679)}}
&
\makecell{$\boldsymbol{-34,720}$\\\textit{(1.829)}}
&
\makecell{$-34,848$\\\textit{(1.806)}}
&
\makecell{$-34,787$\\\textit{(-5.000)}}
&
\makecell{$-34,936$\\\textit{(1.314)}}
&
\makecell{$-34,741$\\\textit{(1.893)}}
\\[2.0ex]

\makecell[l]{Gum90\\(2.000)}
&
\makecell{$-34,670$\\\textit{(-0.697)}}
&
\makecell{$-34,663$\\\textit{(-0.731)}}
&
\makecell{$-34,762$\\\textit{(1.944)}}
&
\makecell{$\boldsymbol{-34,652}$\\\textit{(1.989)}}
&
\makecell{$-34,665$\\\textit{(-5.663)}}
&
\makecell{$-34,666$\\\textit{(1.839)}}
&
\makecell{$-34,697$\\\textit{(2.032)}}
\\[2.0ex]

\makecell[l]{Frank\\(-5.736)}
&
\makecell{$-34,706$\\\textit{(-0.707)}}
&
\makecell{$-34,713$\\\textit{(-0.738)}}
&
\makecell{$-34,746$\\\textit{(2.137)}}
&
\makecell{$-34,721$\\\textit{(2.011)}}
&
\makecell{$\boldsymbol{-34,696}$\\\textit{(-5.893)}}
&
\makecell{$-34,767$\\\textit{(1.808)}}
&
\makecell{$-34,711$\\\textit{(2.080)}}
\\[2.0ex]

\makecell[l]{Clay270\\(2.000)}
&
\makecell{$-34,742$\\\textit{(-0.708)}}
&
\makecell{$-34,738$\\\textit{(-0.737)}}
&
\makecell{$-34,890$\\\textit{(1.958)}}
&
\makecell{$-34,702$\\\textit{(2.024)}}
&
\makecell{$-34,741$\\\textit{(-5.823)}}
&
\makecell{$\boldsymbol{-34,687}$\\\textit{(1.958)}}
&
\makecell{$-34,794$\\\textit{(2.067)}}
\\[2.0ex]

\makecell[l]{Gum270\\(2.000)}
&
\makecell{$-34,576$\\\textit{(-0.679)}}
&
\makecell{$-34,582$\\\textit{(-0.710)}}
&
\makecell{$-34,585$\\\textit{(1.936)}}
&
\makecell{$-34,504$\\\textit{(1.909)}}
&
\makecell{$-34,580$\\\textit{(-5.389)}}
&
\makecell{$-34,662$\\\textit{(1.576)}}
&
\makecell{$\boldsymbol{-34,566}$\\\textit{(1.978)}}
\\

\bottomrule
\end{tabular}%

\begin{tablenotes}[flushleft]
\item \textit{Note:} Values in the first column are the true dependent parameters. Clay90, Gum90, Clay270, and Gum270 denote the 90- and 270-degree rotated Clayton and Gumbel copulas, respectively.
\end{tablenotes}

\end{threeparttable}
}
\end{table}

Tables~\ref{tab:copula_recovery} and~\ref{tab:copula_recovery_tau_neg05} summarize the results under positive and negative dependence, respectively. In most simulation scenarios, the true data-generating copula achieves the highest log-likelihood among all candidate models, and the estimated dependence parameters are close to their true values. These results show that the proposed estimation procedure can accurately recover the underlying dependence structure across a wide range of copula families, including both positive and negative dependence.

We next evaluate the proposed dependence estimation procedure under the zero-inflated mixed Poisson model. This experiment extends the previous setting by introducing a structural-zero component, while keeping the latent-intensity component unchanged.

For simplicity, dependence among the structural-zero indicators is generated from a Gaussian copula with Kendall's tau equal to \(-0.5\). For the latent-intensity component, we consider the Gaussian, Student-\(t\), Clayton, Frank, Gumbel, Clayton survival, and Gumbel survival copulas, with the dependence parameters chosen to give Kendall's tau equal to \(0.5\). For each copula, the corresponding checkerboard copula is constructed using \(d=10\) intervals on each margin. Model recovery is evaluated by comparing the maximized log-likelihood values and the estimated dependence parameters.

\begin{table}[!t]
\centering

\caption{Copula recovery for the bivariate zero-inflated mixed Poisson model}
\label{tab:zi_copula_recovery}

\scriptsize
\resizebox{\columnwidth}{!}{%
\begin{threeparttable}
\renewcommand{\arraystretch}{1.5}
\setlength{\tabcolsep}{3pt}

\begin{tabular}{@{}lccccccc@{}}
\toprule

\textbf{True}
& \textbf{Gauss.}
& \textbf{Student}
& \textbf{Clayton}
& \textbf{Gumbel}
& \textbf{Frank}
& \textbf{Clay180}
& \textbf{Gum180} \\

\midrule

\makecell[l]{Gaussian\\(0.707)}
&
\makecell{$\boldsymbol{-21,789}$\\\textit{(0.707)}}
&
\makecell{$-21,796$\\\textit{(0.688)}}
&
\makecell{$-21,798$\\\textit{(3.797)}}
&
\makecell{$-21,796$\\\textit{(1.734)}}
&
\makecell{$-21,791$\\\textit{(6.218)}}
&
\makecell{$-21,800$\\\textit{(1.123)}}
&
\makecell{$-21,791$\\\textit{(2.283)}}
\\[2.0ex]

\makecell[l]{Student\\(0.707)}
&
\makecell{$-22,151$\\\textit{(0.695)}}
&
\makecell{$\boldsymbol{-22,147}$\\\textit{(0.656)}}
&
\makecell{$-22,172$\\\textit{(3.415)}}
&
\makecell{$-22,148$\\\textit{(1.685)}}
&
\makecell{$-22,159$\\\textit{(5.754)}}
&
\makecell{$-22,151$\\\textit{(1.074)}}
&
\makecell{$-22,154$\\\textit{(2.212)}}
\\[2.0ex]

\makecell[l]{Clayton\\(2.000)}
&
\makecell{$-21,880$\\\textit{(0.558)}}
&
\makecell{$-21,890$\\\textit{(0.533)}}
&
\makecell{$\boldsymbol{-21,867}$\\\textit{(2.183)}}
&
\makecell{$-21,896$\\\textit{(1.421)}}
&
\makecell{$-21,874$\\\textit{(4.131)}}
&
\makecell{$-21,903$\\\textit{(0.637)}}
&
\makecell{$-21,873$\\\textit{(1.776)}}
\\[2.0ex]

\makecell[l]{Gumbel\\(2.000)}
&
\makecell{$-21,752$\\\textit{(0.771)}}
&
\makecell{$-21,750$\\\textit{(0.747)}}
&
\makecell{$-21,778$\\\textit{(5.252)}}
&
\makecell{$-21,747$\\\textit{(1.919)}}
&
\makecell{$-21,763$\\\textit{(7.573)}}
&
\makecell{$\boldsymbol{-21,746}$\\\textit{(1.421)}}
&
\makecell{$-21,758$\\\textit{(2.671)}}
\\[2.0ex]

\makecell[l]{Frank\\(5.736)}
&
\makecell{$-21,926$\\\textit{(0.646)}}
&
\makecell{$-21,932$\\\textit{(0.624)}}
&
\makecell{$-21,926$\\\textit{(3.068)}}
&
\makecell{$-21,935$\\\textit{(1.584)}}
&
\makecell{$\boldsymbol{-21,922}$\\\textit{(5.252)}}
&
\makecell{$-21,940$\\\textit{(0.898)}}
&
\makecell{$-21,925$\\\textit{(2.044)}}
\\[2.0ex]

\makecell[l]{Clay180\\(2.000)}
&
\makecell{$-21,442$\\\textit{(0.844)}}
&
\makecell{$-21,441$\\\textit{(0.831)}}
&
\makecell{$-21,476$\\\textit{(8.279)}}
&
\makecell{$-21,433$\\\textit{(2.323)}}
&
\makecell{$-21,458$\\\textit{(10.077)}}
&
\makecell{$\boldsymbol{-21,428}$\\\textit{(2.058)}}
&
\makecell{$-21,452$\\\textit{(3.346)}}
\\[2.0ex]

\makecell[l]{Gum180\\(2.000)}
&
\makecell{$-21,978$\\\textit{(0.654)}}
&
\makecell{$-21,981$\\\textit{(0.625)}}
&
\makecell{$-21,979$\\\textit{(3.034)}}
&
\makecell{$-21,986$\\\textit{(1.597)}}
&
\makecell{$-21,977$\\\textit{(5.276)}}
&
\makecell{$-21,991$\\\textit{(0.912)}}
&
\makecell{$\boldsymbol{-21,976}$\\\textit{(2.059)}}
\\

\bottomrule
\end{tabular}%

\begin{tablenotes}[flushleft]
\item \textit{Note:} Each cell reports the maximized log-likelihood and, in parentheses, the estimated dependence parameter. Values in the first column are the true parameters. The largest log-likelihood in each row is shown in bold. Clay180 and Gum180 denote the 180-degree rotated Clayton and Gumbel copulas, respectively.
\end{tablenotes}

\end{threeparttable}
}
\end{table}

Table~\ref{tab:zi_copula_recovery} summarizes the copula recovery results. Compared with the standard mixed Poisson model, recovering the dependence structure is more challenging because structural zeros reduce the information available for estimating the latent-intensity dependence.

In most simulation scenarios, the true data-generating copula achieves the highest log-likelihood, and the estimated dependence parameters are close to their true values. The main exception occurs when the latent dependence is generated from a Gumbel copula, where the Clayton survival copula occasionally gives a slightly higher log-likelihood. However, the difference is small, indicating that the two copulas produce very similar dependence structures in this setting.

\subsection{Recovery under Gaussian Factor Checkerboard Copulas}

We finally evaluate the proposed Gaussian factor checkerboard copula in a higher-dimensional setting. Data are generated from a multivariate mixed Poisson model with \(k=5\) response variables. All margins follow Poisson--Gamma distributions, and the regression and dispersion parameters are the same as those used in the lower-heterogeneity setting reported in Table~\ref{tab:mp_parameter_recovery}. The sample size is fixed at \(n=10\,000\).

Dependence among the latent-intensity variables is generated from a Gaussian factor copula. Two factor structures are considered, with one common factor (\(J=1\)) and two common factors (\(J=2\)). The proposed IFM procedure is used to estimate the factor loading matrix.

Table~\ref{tab:factor_loading_combined} compares the true and estimated factor loading matrices, while the corresponding correlation matrices are reported in Appendix~\ref{simulation_result}.

The proposed method accurately recovers the factor loading matrix under both factor structures. The estimated loadings are close to their true values, leading to correlation matrices that closely match the true dependence structure. These results show that the proposed Gaussian factor checkerboard copula provides accurate estimation in higher-dimensional mixed Poisson models while using a parsimonious dependence structure.

\begin{table}[H]
\centering

\caption{Recovery of factor loadings under the Gaussian factor checkerboard copula}
\label{tab:factor_loading_combined}

\scriptsize
\resizebox{\columnwidth}{!}{%
\begin{threeparttable}
\renewcommand{\arraystretch}{1.5}
\setlength{\tabcolsep}{5pt}

\begin{tabular}{@{}c cc c cccc@{}}
\toprule

& \multicolumn{2}{c}{\textbf{Case \(J=1\)}}
&
& \multicolumn{4}{c}{\textbf{Case \(J=2\)}} \\

\cmidrule(lr){2-3}
\cmidrule(lr){5-8}

\textbf{Margin}
& \(\boldsymbol{\Psi}\)
& \(\widehat{\boldsymbol{\Psi}}\)
&
& \multicolumn{2}{c}{\(\boldsymbol{\Psi}\)}
& \multicolumn{2}{c}{\(\widehat{\boldsymbol{\Psi}}\)} \\

\cmidrule(lr){1-1}
\cmidrule(lr){2-2}
\cmidrule(lr){3-3}
\cmidrule(lr){5-6}
\cmidrule(lr){7-8}
1
& $0.939$
& $0.937$
&
& $0.958$
& $0.000$
& $0.932$
& $0.000$
\\

2
& $0.625$
& $0.636$
&
& $0.273$
& $0.913$
& $0.281$
& $0.893$
\\

3
& $-0.625$
& $-0.648$
&
& $0.518$
& $-0.633$
& $0.527$
& $-0.656$
\\

4
& $0.514$
& $0.494$
&
& $-0.537$
& $0.597$
& $-0.555$
& $0.588$
\\

5
& $-0.514$
& $-0.542$
&
& $-0.537$
& $-0.597$
& $-0.566$
& $-0.584$
\\

\bottomrule
\end{tabular}%

\begin{tablenotes}[flushleft]
\item \textit{Note:} For \(J=2\), the true and estimated factor loadings are presented as the rows of the loading matrices \(\boldsymbol{\Psi}\) and \(\widehat{\boldsymbol{\Psi}}\), respectively.
\end{tablenotes}

\end{threeparttable}
}
\end{table}

Overall, the simulation studies show that the proposed framework provides accurate estimation of the marginal models and dependence structures under a wide range of data-generating mechanisms. The proposed method consistently recovers the underlying mixing distributions, regression and dispersion parameters, copula models, and Gaussian factor structures, supporting the effectiveness of the proposed estimation framework.

\section{Empirical Application}

To illustrate the applicability of the proposed framework, we consider four real-world datasets from different application domains. Detailed descriptions of the datasets, their sources, and the corresponding \textsf{R} code for data preparation are provided in \ref{appendix:data} and Table~\ref{tab:datasets_regression_specifications}.

The empirical analysis is designed to evaluate different aspects of the proposed methodology. First, the MEPS healthcare utilization dataset is used to evaluate the proposed bivariate mixed Poisson model and compare it with existing copula-based regression methods. Second, the terrorism dataset distributed with the \texttt{bizicount} package is used to evaluate the proposed bivariate zero-inflated mixed Poisson model under strong zero inflation. Third, the VHLSS healthcare dataset is used to evaluate the computational performance of the proposed zero-inflated mixed Poisson model on a large real-world dataset. Finally, single-cell RNA sequencing (scRNA) data are used to evaluate the proposed multivariate mixed Poisson and zero-inflated mixed Poisson models in high-dimensional settings. The experiments consider up to 50 response variables for the mixed Poisson model and up to 10 response variables for the zero-inflated mixed Poisson model.

Whenever appropriate, the proposed models are compared with existing methods. For bivariate models, comparisons are made with the \texttt{GJRM} package \citep{marra2017joint} and the \texttt{bizicount} package \citep{niehaus2025package}. For high-dimensional count data, comparisons are made with the Poisson--Lognormal models implemented in the \texttt{PLNmodels} package \citep{chiquet2021poisson}. Model performance is evaluated using the maximized log-likelihood, the Akaike information criterion (AIC), and the Bayesian information criterion (BIC). Unless otherwise stated, checkerboard copulas are constructed using a regular partition with \(d=10\) intervals on each margin, and Gaussian factor models are estimated using the Gauss--Hermite quadrature described in \ref{app:gaussian_checkerboard}.

\subsection{Bivariate Mixed Poisson Models}

We first evaluate the proposed framework using the Medical Expenditure Panel Survey (MEPS) dataset, which is widely used as a benchmark for bivariate healthcare utilization studies.

As a benchmark, we compare the proposed model with the bivariate mixed Poisson models implemented in the \texttt{GJRM} package \citep{petti2022copula}. In \texttt{GJRM}, dependence is modeled directly between the observed count responses through a copula. For discrete responses, the joint probability mass function is obtained from the rectangle probabilities
\[
\begin{aligned}
&
\mathbb{P}(N_1=n_1,N_2=n_2)\\
&\qquad=
C(F_1(n_1),F_2(n_2))
-
C(F_1(n_1-1),F_2(n_2))
\\
& \qquad \qquad
-
C(F_1(n_1),F_2(n_2-1))
+
C(F_1(n_1-1),F_2(n_2-1)),
\end{aligned}
\]
where \(F_1\) and \(F_2\) are the marginal distribution functions and \(C\) is the copula function.

In contrast, the proposed framework models dependence through the latent mixing variables using a checkerboard copula, while the observed responses remain conditionally independent given the latent intensities. This preserves the classical mixed Poisson interpretation, where dependence arises from shared latent heterogeneity.

Both methods consider Gamma and Inverse Gaussian mixing distributions, although the parameterization of the Inverse Gaussian distribution differs. The \texttt{GJRM} package assumes
\[
\mathbb{V}(\Lambda)=\mu^{2}\phi,
\]
whereas the proposed framework uses
\[
\mathbb{V}(\Lambda)=\mu^{3}\phi
\]
by default. To ensure a fair comparison, all results reported below use the same Inverse Gaussian parameterization as \texttt{GJRM}. In addition, the proposed framework includes the Lognormal mixing distribution, which is not available in \texttt{GJRM}. The corresponding parameterizations are summarized in \ref{app_distribution_for_lambda}.

\begin{table}[!t]
\centering
\caption{Comparison between \texttt{GJRM} and the proposed checkerboard copula on the MEPS data.}
\label{tab:meps_gjrm_checkerboard_comparison}

\scriptsize
\resizebox{\columnwidth}{!}{%
\begin{threeparttable}
\renewcommand{\arraystretch}{1.4}
\setlength{\tabcolsep}{4pt}

\begin{tabular}{llcccc}
\toprule
\multirow{2}{*}{\textbf{Margin 1}}
&
\multirow{2}{*}{\textbf{Margin 2}}
&
\multicolumn{2}{c}{\textbf{Marginal LogLik}}
&
\multicolumn{2}{c}{\textbf{Joint LogLik}}
\\
\cmidrule(lr){3-4}
\cmidrule(lr){5-6}
&
&
\textbf{Margin 1}
&
\textbf{Margin 2}
&
\textbf{GJRM}
&
\textbf{Proposed}
\\
\midrule

Gamma & Gamma
& $-18,919$ & $-10,778$
& $-29,184$ & $-29,189$ \\

Gamma & IG
& $-18,919$ & $-10,784$
& $\boldsymbol{-29,180}$ & $-29,200$ \\

Gamma & LN
& $-18,919$ & $-10,771$
& -- & $\boldsymbol{-29,181}$ \\

IG & Gamma
& $-19,042$ & $-10,778$
& $-29,307$ & $-29,324$ \\

IG & IG
& $-19,042$ & $-10,784$
& $-29,297$ & $-29,328$ \\

IG & LN
& $-19,042$ & $-10,771$
& -- & $-29,311$ \\

LN & Gamma
& $-18,950$ & $-10,778$
& -- & $-29,231$ \\

LN & IG
& $-18,950$ & $-10,784$
& -- & $-29,236$ \\

LN & LN
& $-18,950$ & $-10,771$
& -- & $-29,216$ \\

\bottomrule
\end{tabular}

\begin{tablenotes}[flushleft]
\item \textit{Note:} Both methods use the same Gaussian copula dependence structure. For the Inverse Gaussian and Lognormal mixing distributions, the variance functions are chosen to match those used in \texttt{GJRM}, rather than the default parameterizations of the proposed framework.
\end{tablenotes}

\end{threeparttable}
}
\end{table}

Several observations can be made from Table~\ref{tab:meps_gjrm_checkerboard_comparison}.

First, the proposed framework achieves a level of fit comparable to that of \texttt{GJRM}, despite modeling dependence through the latent mixing variables rather than directly through the observed responses.

Second, allowing different mixing distributions for the two margins further improves model fit. Among all candidate models, the Gamma--Lognormal specification achieves the highest log-likelihood. This result highlights the flexibility of the proposed framework in selecting the marginal distribution separately for each response.

Having selected the Gamma--Lognormal specification for the marginal distributions, we next compare different checkerboard copula families while keeping the marginal distributions fixed. Table~\ref{tab:meps_copula_selection} summarizes the results. The survival Gumbel copula achieves the highest log-likelihood, followed closely by the Gaussian copula. The small improvement in log-likelihood suggests that, in addition to positive dependence, the MEPS data may exhibit a mild degree of lower-tail dependence.

\begin{table}[!t]
\centering

\caption{Copula selection for the Gamma--Lognormal model.}
\label{tab:meps_copula_selection}

\scriptsize
\resizebox{\columnwidth}{!}{%
\begin{threeparttable}
\renewcommand{\arraystretch}{1.5}
\setlength{\tabcolsep}{12pt}

\begin{tabular}{l c r r}
\toprule
\textbf{Copula family}
&
\textbf{Parameter}
&
\textbf{LogLik}
&
\textbf{\(\Delta\)AIC}
\\
\midrule

Gum180
&
\(\theta=1.807\)
&
\(\mathbf{-29,176.88}\)
&
\(\mathbf{0.00}\)
\\

Gaussian
&
\(\rho=0.572\)
&
\(-29,181.20\)
&
8.64
\\

Frank
&
\(\theta=4.312\)
&
\(-29,185.79\)
&
17.82
\\

Student
&
\(\rho=0.568\)
&
\(-29,186.36\)
&
18.96
\\

Gumbel
&
\(\theta=1.481\)
&
\(-29,198.36\)
&
42.96
\\

Clayton
&
\(\theta=2.099\)
&
\(-29,202.68\)
&
51.60
\\

Clay180
&
\(\theta=0.769\)
&
\(-29,216.01\)
&
78.26
\\

\bottomrule
\end{tabular}

\begin{tablenotes}[flushleft]
\item \textit{Note:} Clay180 and Gum180 denote the \(180^\circ\)-rotated Clayton and Gumbel copulas, respectively.
\end{tablenotes}

\end{threeparttable}
}
\end{table}

Overall, the MEPS analysis demonstrates that the proposed framework provides a flexible approach for modeling multivariate mixed Poisson data. By allowing different mixing distributions and copula families, the proposed model can achieve a better fit than existing approaches while preserving the classical mixed Poisson interpretation through latent heterogeneity.

\subsection{Bivariate Zero-Inflated Mixed Poisson Models}
We next evaluate the proposed bivariate zero-inflated mixed Poisson framework using two real-world datasets. The first is the \texttt{terror} dataset distributed with the \texttt{bizicount} package \citep{niehaus2025package}, and the second is the VHLSS healthcare utilization dataset.

We first consider the \texttt{terror} dataset, which contains a relatively small sample but a very high proportion of zero observations in both response variables. This dataset is therefore well suited for evaluating models with structural zeros.

Preliminary analyses showed that, when dependence among the structural-zero indicators was modeled using a Gaussian copula, the estimated correlation parameter was consistently close to the lower boundary of the parameter space. This suggests an almost countermonotone dependence structure. Therefore, the structural-zero component is modeled using a countermonotone copula, while dependence among the latent-intensity variables is modeled using a Gaussian checkerboard copula.

The proposed framework is compared with the bivariate zero-inflated count model implemented in the \texttt{bizicount} package, which models dependence through a single copula acting directly on the observed responses. In contrast, the proposed framework models dependence separately for the structural-zero and latent-intensity components.

Table~\ref{tab:terror_comparison} summarizes the results. For the Poisson--Gamma specification, the proposed framework achieves a slightly higher log-likelihood than the model fitted using the \texttt{bizicount} package. More importantly, the proposed framework separates the dependence structure into two components: a countermonotone copula for the structural-zero component and a Gaussian checkerboard copula for the latent-intensity component. This provides a more interpretable representation of the dependence in zero-inflated count data.

\begin{table*}[!t]
\centering
\caption{Comparison between \texttt{bizicount} and the proposed model on the \texttt{terror} dataset.}
\label{tab:terror_comparison}

\small

\begin{threeparttable}
\renewcommand{\arraystretch}{1.4}
\setlength{\tabcolsep}{8pt}

\begin{tabular}{l llccc r}
\toprule

\multirow{2}{*}{\textbf{Method}}
&
\multirow{2}{*}{\textbf{Margin 1}}
&
\multirow{2}{*}{\textbf{Margin 2}}
&
\multicolumn{3}{c}{\textbf{Dependence structure}}
&
\multirow{2}{*}{\textbf{LogLik}}
\\

\cmidrule(lr){4-6}

&
&
&
Copula for \(\mathbf{Y}\)
&
Copula for \(\mathbf{B}\)
&
Copula for \(\mathbf{\Lambda}\)
&
\\

\midrule

\texttt{bizicount}
&
Gamma
&
Gamma
&
Gaussian (-0.035)
&
--
&
--
&
$-418.44$
\\

Proposed
&
Gamma
&
Gamma
&
--
&
Countermonotone
&
Gaussian (-0.043)
&
$-417.74$
\\

Proposed
&
IG
&
Gamma
&
--
&
Countermonotone
&
Gaussian (0.105)
&
$\mathbf{-413.74}$
\\

\bottomrule
\end{tabular}

\begin{tablenotes}[flushleft]
\item \textit{Note:}
For the proposed model, dependence is modeled separately for the structural-zero component (\(B\)) and the latent-intensity component \(\Lambda\). The \texttt{bizicount} model uses a single copula acting directly on the observed responses (\(Y\)).
\end{tablenotes}

\end{threeparttable}
\end{table*}

Further improvement is achieved by allowing different mixing distributions across the two margins. Replacing the Gamma distribution in the first margin with an Inverse Gaussian distribution increases the log-likelihood from \(-417.74\) to \(-413.74\), demonstrating the benefit of selecting the marginal distribution separately for each response.

We next analyze the VHLSS healthcare utilization dataset, which contains more than \(34{,}000\) observations on outpatient and inpatient visits in Vietnam. Similar to the \texttt{terror} dataset, both response variables exhibit substantial zero inflation. However, this application is considerably more challenging because of the much larger sample size and the richer regression specification, which includes a large number of demographic, socioeconomic, and health-related covariates.

As a benchmark, we compare the proposed framework with the bivariate zero-inflated count model implemented in the \texttt{bizicount} package. For large datasets with complex regression models, repeated evaluation of discrete copula rectangle probabilities leads to substantial computational cost.

In contrast, the proposed framework models dependence through continuous latent variables, so likelihood evaluation only involves continuous copula densities together with univariate mixed Poisson likelihoods. As a result, the complete model selection procedure, including both marginal distribution and copula selection, can be completed in a few minutes on a standard desktop computer, whereas the corresponding analysis using the \texttt{bizicount} package requires several hours.

\begin{table*}[!t]
\centering

\caption{Comparison of alternative dependence structures for the bivariate VHLSS healthcare utilization data.}
\label{tab:vhlss_dependence_comparison}

\scriptsize
\resizebox{\textwidth}{!}{%
\renewcommand{\arraystretch}{1.2}
\setlength{\tabcolsep}{4pt}
\begin{threeparttable}

\begin{tabular}{lllcccrr}
\toprule
\multirow{2}{*}{\textbf{Method}}
&
\multirow{2}{*}{\textbf{Margin 1}}
&
\multirow{2}{*}{\textbf{Margin 2}}
&
\multicolumn{3}{c}{\textbf{Dependence structure}}
&
\multirow{2}{*}{\textbf{LogLik}}
&
\multirow{2}{*}{\textbf{AIC}}
\\
\cmidrule(lr){4-6}
&
&
&
\textbf{Copula for \(Y\)}
&
\textbf{Copula for \(B\)}
&
\makecell{\textbf{Copula for \(\Lambda\)}}
&
&
\\
\midrule

\texttt{bizicount}
&
Gamma
&
Gamma
&
Gaussian (0.338)
&
--
&
--
&
$-42,667.75$
&
$85,573.50$
\\

\texttt{bizicount}
&
Gamma
&
Gamma
&
Frank (3.028)
&
--
&
--
&
$-42,650.92$
&
$85,539.85$
\\

\midrule

Proposed
&
Gamma
&
Gamma
&
--
&
Gaussian (0.659)
&
Gaussian (0.240)
&
$-42,686.54$
&
$85,613.08$
\\

\midrule

Proposed
&
IG
&
LN
&
--
&
Gaussian (0.617)
&
Gaussian (0.298)
&
$-42,527.75$
&
$85,295.50$
\\

Proposed
&
IG
&
LN
&
--
&
Gaussian (0.627)
&
Student (\(\nu=3\), 0.255)
&
$-42,527.68$
&
$85,295.36$
\\

Proposed
&
IG
&
LN
&
--
&
Gaussian (0.610)
&
Clayton (0.636)
&
$-42,532.56$
&
$85,305.12$
\\

Proposed
&
IG
&
LN
&
--
&
Gaussian (0.624)
&
Gumbel (1.178)
&
$-42,526.31$
&
$85,292.62$
\\

Proposed
&
IG
&
LN
&
--
&
Gaussian (0.617)
&
Frank (1.866)
&
$-42,528.58$
&
$85,297.16$
\\

Proposed
&
IG
&
LN
&
--
&
Gaussian (0.625)
&
Clay180 (0.325)
&
$\boldsymbol{-42,525.70}$
&
$\textbf{85,291.40}$
\\

Proposed
&
IG
&
LN
&
--
&
Gaussian (0.613)
&
Gum180 (1.289)
&
$-42,529.70$
&
$85,299.40$
\\

\bottomrule
\end{tabular}

\begin{tablenotes}[flushleft]
\item \textit{Note:} IG and LN denote the Inverse Gaussian and Lognormal mixing distributions, respectively. Clay180 and Gum180 denote the corresponding \(180^\circ\)-rotated copulas.
\end{tablenotes}

\end{threeparttable}

}
\end{table*}
Table~\ref{tab:vhlss_dependence_comparison} reports the estimation results, from which several conclusions emerge.

First, explicitly modeling dependence separately in the structural-zero and latent-intensity layers substantially improves model fit relative to conventional copula models based on a single dependence structure for the observed responses. 

Second, the choice of the marginal mixing distributions has an important impact on model performance. While Gamma mixing distributions provide a reasonable fit, combining an Inverse Gaussian distribution for outpatient utilization with a Lognormal distribution for inpatient utilization consistently yields higher likelihood values and lower information criteria. 

Finally, among the candidate checkerboard copulas, the survival Clayton copula provides the best overall fit for the latent-intensity layer, although the Gaussian, Student-$t$, Gumbel, and Frank checkerboard copulas achieve very similar likelihood values.

Taken together, these two empirical studies demonstrate that the proposed framework simultaneously provides flexible marginal modeling, interpretable decomposition of dependence, and computationally efficient likelihood-based inference. These advantages become particularly important for large healthcare and insurance databases exhibiting both substantial overdispersion and excess zeros.

\subsection{Multivariate Mixed Poisson Models}

We next evaluate the proposed framework for moderately high-dimensional multivariate count data. The analysis is based on the \texttt{scRNA} dataset available in the \texttt{PLNmodels} package. As a benchmark, we consider the multivariate Poisson--Lognormal (PLN) model implemented in the \texttt{PLN()} function, which is one of the most widely used latent Gaussian models for multivariate count data. Since the PLN model is estimated by a variational EM algorithm, it maximizes an evidence lower bound (ELBO) rather than the exact observed-data log-likelihood. Therefore, comparisons with the proposed framework should be interpreted accordingly.

Two experiments are considered with \(k=10\) and \(k=50\) response variables. For each experiment, the response variables are randomly selected from the 500 genes in the original dataset, and the categorical variable \texttt{cell\_line} is included as the only explanatory variable.

To evaluate separately the effects of the marginal model and the dependence structure, two marginal specifications are considered. The first assumes Lognormal mixing distributions for all response variables, matching the marginal assumption of the PLN model. The second selects the mixing distribution for each response variable independently using the marginal model selection procedure described in Section~\ref{sec:marginal_selection}. For each marginal specification, we compare an independence model with Gaussian factor checkerboard copulas using different numbers of latent factors.

\begin{table*}[!t]
\centering
\caption{Comparison of competing multivariate count models \(\boldsymbol{(k=10)}\)}
    \label{tab:model_comparison_k10}
    
\scriptsize
\resizebox{\textwidth}{!}{%
\renewcommand{\arraystretch}{1.2}
\setlength{\tabcolsep}{7pt}
\begin{threeparttable}
     \begin{tabular}{lll c ccc}
    \toprule
    \textbf{Method} & \textbf{Margin} & \textbf{Latent Dependence} & \textbf{No. Param.} & \textbf{LogLik} & \textbf{AIC} & \textbf{BIC} \\
    \midrule
    \texttt{PLNmodels} & LN & Unrestricted Gaussian
& $105$
& $-80,053.89$
& $160,317.78$
& $160,485.05$ \\
\midrule
Proposed & LN & Independence
& $60$
& $-82,312.43$
& $164,744.86$
& $164,840.44$ \\

Proposed & Optimal & Independence
& $60$
& $-81,958.18$
& $164,036.36$
& $164,131.94$ \\

Proposed & LN & GFC(1)
& $70$
& $-79,653.49$
& $159,446.98$
& $159,558.49$ \\

Proposed & Optimal & GFC(1)
& $70$
& $-79,306.11$
& $158,752.22$
& $158,863.73$ \\

Proposed & LN & GFC(2)
& $80$
& $-79,374.13$
& $158,908.26$
& $159,035.71$ \\

Proposed & \textbf{Optimal} & \textbf{GFC(2)}
& $\boldsymbol{80}$
& $\boldsymbol{-79,044.26}$
& $\boldsymbol{158,248.52}$
& $\boldsymbol{158,375.97}$ \\
    \bottomrule
    \end{tabular}
    
    \begin{tablenotes}[flushleft]
       \item \textit{Note:} \texttt{PLNmodels} denotes the multivariate Poisson--Lognormal model implemented in the \texttt{PLNmodels} package. Optimal indicates that the marginal mixing distribution of each response variable is selected independently according to the highest univariate LogLik. Independence denotes independent latent mixing variables. GFC(\(J\)) denotes a Gaussian factor checkerboard copula with \(J\) latent factors.
    \end{tablenotes}
\end{threeparttable}
}
\end{table*}

\begin{table*}[!tb]
\centering

\caption{Comparison of competing multivariate count models \(\boldsymbol{(k=50)}\)}
\label{tab:model_comparison_k50}

\scriptsize
\resizebox{\textwidth}{!}{%
\renewcommand{\arraystretch}{1.2}
\begin{threeparttable}

\begin{tabular}{lll c ccc}
\toprule
\textbf{Method} &
\textbf{Margin} &
\textbf{Latent Dependence} &
\textbf{No. Param.} &
\textbf{LogLik} &
\textbf{AIC} &
\textbf{BIC} \\
\midrule

\texttt{PLNmodels}
&
LN
&
Unrestricted Gaussian
&
$1,525$
&
$-363,642.30$
&
$730,334.60$
&
$732,764.02$
\\

\midrule

Proposed
&
LN
&
Independence
&
$300$
&
$-389,836.10$
&
$780,272.20$
&
$780,750.12$
\\

Proposed
&
Optimal
&
Independence
&
$300$
&
$-387,671.80$
&
$775,943.60$
&
$776,421.52$
\\

Proposed
&
LN
&
GFC(1)
&
$350$
&
$-368,876.20$
&
$738,452.40$
&
$739,009.97$
\\

Proposed
&
Optimal
&
GFC(1)
&
$350$
&
$-367,035.30$
&
$734,770.60$
&
$735,328.17$
\\

Proposed
&
LN
&
GFC(2)
&
$400$
&
$-365,274.50$
&
$731,349.00$
&
$731,986.23$
\\

Proposed
&
Optimal
&
GFC(2)
&
$400$
&
$-363,602.30$
&
$728,004.60$
&
$728,641.83$
\\

Proposed
&
LN
&
GFC(3)
&
$450$
&
$-363,009.30$
&
$726,918.60$
&
$727,635.48$
\\

Proposed
&
\textbf{Optimal}
&
\textbf{GFC(3)}
&
$\boldsymbol{450}$
&
$\boldsymbol{-361,435.70}$
&
$\boldsymbol{723,771.40}$
&
$\boldsymbol{724,488.28}$
\\

\bottomrule
\end{tabular}

\begin{tablenotes}[flushleft]
       \item \textit{Note:} \texttt{PLNmodels} denotes the multivariate Poisson--Lognormal model implemented in the \texttt{PLNmodels} package. Optimal indicates that the marginal mixing distribution of each response variable is selected independently according to the highest univariate LogLik. Independence denotes independent latent mixing variables. GFC(\(J\)) denotes a Gaussian factor checkerboard copula with \(J\) latent factors.
    \end{tablenotes}

\end{threeparttable}
}
\end{table*}

Tables~\ref{tab:model_comparison_k10} and~\ref{tab:model_comparison_k50} summarize the results.

First, flexible marginal specification consistently improves model fit. For all dependence structures considered, selecting the mixing distribution separately for each response variable gives higher log-likelihood values and lower AIC and BIC than assuming a common Lognormal mixing distribution.

Second, introducing dependence through the Gaussian factor checkerboard copula substantially improves model fit over the independence model. Increasing the number of latent factors further improves the fit, indicating that the Gaussian factor representation effectively captures the dependence among the response variables.

Finally, the proposed framework compares favorably with the benchmark PLN model. For both \(k=10\) and \(k=50\), the proposed model achieves higher log-likelihood values together with lower AIC and BIC. The improvement is particularly evident for \(k=50\), where the proposed model uses only 450 dependence parameters under a three-factor specification, compared with 1,525 parameters in the unrestricted covariance matrix of the PLN model.

Overall, these results demonstrate that the proposed framework combines flexible marginal modeling with a parsimonious dependence structure. The Gaussian factor checkerboard copula provides accurate estimation while requiring substantially fewer dependence parameters than covariance-based approaches.

\subsection{Multivariate Zero-Inflated Mixed Poisson Models}

Our final empirical application considers moderately high-dimensional multivariate count data with substantial zero inflation. The analysis is based on the \texttt{scRNA} dataset available in the \texttt{PLNmodels} package. To emphasize the role of structural zeros, we select the ten genes with the largest proportions of zero observations.

As a preliminary analysis, all response variables are assumed to be independent and only the marginal models are estimated. This comparison evaluates whether a zero-inflation component is supported before modeling the dependence structure. Both the proposed zero-inflated mixed Poisson (ZIMP) model and the zero-inflated Poisson--Lognormal (ZI-PLN) model achieve substantially higher log-likelihood values together with lower AIC and BIC than their corresponding non-zero-inflated models, providing strong evidence for excess zeros in the selected genes.

We next investigate the effects of the marginal specification and the dependence structure. Two marginal specifications are considered. The first assumes Lognormal mixing distributions for all response variables, matching the assumption of the PLN model. The second selects the mixing distribution for each response variable independently using the marginal model selection procedure described in Section~\ref{sec:marginal_selection}.

Dependence is modeled separately for the structural-zero and latent-intensity components. The structural-zero component is modeled using a Gaussian factor copula, whereas the latent-intensity component is modeled using a Gaussian factor checkerboard copula. For each marginal specification, we compare independence models with one-factor and two-factor models for both dependence components, allowing the effects of structural-zero dependence and latent-intensity dependence to be evaluated separately.

\begin{table*}
\centering

\caption{Comparison of multivariate zero-inflated mixed Poisson models under alternative marginal and dependence specifications.}
\label{tab:model_comparison_5}

\scriptsize
\resizebox{\textwidth}{!}{%
\renewcommand{\arraystretch}{1.2}
\setlength{\tabcolsep}{5pt}

\begin{threeparttable}

\begin{tabular}{lcccccrrr}
\hline
\textbf{Method} &
\makecell{\textbf{Zero}\\\textbf{Inflation}} &
\makecell{\textbf{Margins}} &
\makecell{\textbf{Copula}\\\textbf{for} $B$} &
\makecell{\textbf{Copula}\\\textbf{for} $\Lambda$} &
\makecell{\textbf{No.}\\\textbf{Param.}} &
\textbf{LogLik} &
\textbf{AIC} &
\textbf{BIC}\\

\midrule

\texttt{PLNmodels}
& No
& --
& --
& Unres-Gauss.
& $105$
& $-31,626.09$
& $63,462.18$
& $63,629.45$ \\

\texttt{PLNmodels}
& Yes
& --
& --
& Unres-Gauss.
& $155$
& $-31,458.00$
& $63,226.00$
& $63,472.92$ \\

\midrule

Proposed
& No
& LN
& --
& Indep.
& $60$
& $-31,865.94$
& $63,851.88$
& $63,947.46$ \\

Proposed
& No
& Optimal
& --
& Indep.
& $60$
& $-31,613.71$
& $63,347.42$
& $63,443.00$ \\

Proposed
& Yes
& LN
& Indep.
& Indep.
& $110$
& $-31,248.55$
& $62,717.10$
& $62,892.34$ \\

Proposed
& Yes
& Optimal
& Indep.
& Indep.
& $110$
& $-31,223.44$
& $62,666.88$
& $62,842.12$ \\

\midrule

Proposed
& Yes
& LN
& Indep.
& GFC(1)
& $120$
& $-30,382.66$
& $61,005.32$
& $61,196.49$ \\

Proposed
& Yes
& Optimal
& Indep.
& GFC(1)
& $120$
& $-30,362.69$
& $60,965.37$
& $61,156.54$ \\

Proposed
& Yes
& LN
& GF(1)
& GFC(1)
& $130$
& $-30,224.19$
& $60,708.39$
& $60,915.49$ \\

Proposed
& Yes
& LN
& Indep.
& GFC(2)
& $130$
& $-30,246.83$
& $60,753.66$
& $60,960.76$ \\

Proposed
& Yes
& Optimal
& GF(1)
& GFC(1)
& $130$
& $-30,215.15$
& $60,690.30$
& $60,897.40$ \\

Proposed
& Yes
& LN
& GF(2)
& GFC(1)
& $140$
& $-30,142.30$
& $60,564.60$
& $60,787.63$ \\

Proposed
& Yes
& LN
& GF(1)
& GFC(2)
& $140$
& $-30,093.26$
& $60,466.52$
& $60,689.55$ \\

Proposed
& Yes
& Optimal
& GF(1)
& GFC(2)
& $\boldsymbol{140}$
& $\boldsymbol{-30,067.45}$
& $\boldsymbol{60,414.90}$
& $\boldsymbol{60,637.93}$ \\

\hline
\end{tabular}
\begin{tablenotes}[flushleft]
\item \textit{Note:}
\texttt{PLNmodels} denotes the multivariate Poisson--Lognormal model implemented in the \texttt{PLNmodels} package. Optimal indicates that the mixing distribution of each response variable is selected using the marginal model selection procedure in Section~\ref{sec:marginal_selection}. Indep.\ denotes independence, GF($k$) denotes a $k$-factor Gaussian factor copula and GFC($k$) denotes a $k$-factor Gaussian factor checkerboard copula.
\end{tablenotes}

\end{threeparttable}
}
\end{table*}

Table~\ref{tab:model_comparison_5} summarizes the results.

First, incorporating a zero-inflation component substantially improves model fit. Both the proposed ZIMP models and the benchmark ZI-PLN model achieve higher log-likelihood values together with lower AIC and BIC than their corresponding non-zero-inflated models, confirming the presence of substantial excess zeros in the selected genes.

Second, flexible marginal specification further improves model performance. For both the independence and dependence models, selecting the mixing distribution separately for each response variable consistently gives higher log-likelihood values and lower AIC and BIC than assuming a common Lognormal mixing distribution.

Finally, modeling dependence further improves the fit. Introducing dependence among the latent-intensity variables through the Gaussian factor checkerboard copula substantially increases model fit, while additionally modeling dependence among the structural-zero indicators using a Gaussian factor copula provides further improvement. The best performance is achieved when dependence is modeled in both components under the optimal marginal specification.

Overall, the empirical studies demonstrate the flexibility and effectiveness of the proposed framework across a wide range of multivariate count data. The proposed methodology consistently achieves competitive or superior model fit while providing a flexible choice of marginal distributions and an interpretable representation of the dependence structure. These results show that the proposed framework is well suited for modeling multivariate count data with overdispersion, excess zeros, and complex dependence structures.

\section{Conclusion}

This paper proposes a general framework for modeling multivariate zero-inflated mixed Poisson data that jointly accommodates overdispersion, excess zeros, and multivariate dependence. The proposed hierarchical construction separates the structural-zero component from the latent-intensity component, allowing the two sources of dependence to be modeled independently. Dependence among the latent mixing variables is modeled using checkerboard copulas, while dependence among the structural-zero indicators is modeled separately through an appropriate copula. The marginal distributions are specified through a broad class of mixing distributions, including the Gamma, Lognormal, Inverse Gaussian, Weibull, and Pareto families. This modular construction substantially enlarges the class of multivariate zero-inflated mixed Poisson models that can be formulated within a unified likelihood-based framework.

A second contribution is the combination of checkerboard copulas with Gaussian factor models for moderately high-dimensional data. This representation substantially reduces the number of dependence parameters while retaining the flexibility of the checkerboard copula construction. Together with the proposed quantile-based computation of marginal and checkerboard cell probabilities, it provides a unified estimation procedure for all supported mixing distributions. The resulting inference-functions-for-margins procedure makes likelihood-based estimation computationally feasible for moderately high-dimensional multivariate count data.

The simulation studies and empirical applications demonstrate that the proposed framework provides accurate estimation, flexible marginal modeling, and interpretable dependence structures across a wide range of multivariate count data. Compared with existing approaches, the proposed methodology offers greater flexibility in the choice of marginal distributions while maintaining computational efficiency.

More generally, the proposed methodology should be viewed as a general modeling framework rather than a single multivariate count model. New mixing distributions and alternative factor copula constructions can be incorporated without changing the overall estimation framework, provided that the corresponding marginal probabilities and checkerboard cell probabilities can be evaluated. The proposed framework is therefore applicable to a wide range of multivariate count data arising in actuarial science, healthcare, epidemiology, ecology, genomics, and many other application domains.

Several directions for future research are possible. First, alternative factor copulas may be used to model more flexible dependence structures while preserving computational efficiency. Second, richer marginal specifications, such as semiparametric mixing distributions, more flexible quantile functions, or distributional regression models, may further improve the modeling of heterogeneous count data. Owing to the modular construction of the proposed framework, these extensions can be incorporated without changing the overall estimation strategy.

\section*{Code Availability}

All methods presented in this paper are implemented in \textsf{R}. The complete source code used to reproduce the simulation studies, empirical analyses, tables, and figures is publicly available at

\url{https://github.com/nguyenquanghuy85/mzimp-R-code}.

\section*{Declaration of generative AI and AI-assisted technologies in the manuscript preparation process}
During the preparation of this work, the authors used ChatGPT (OpenAI) in order to assist with improving the clarity of the English language, refining the presentation and organization of the manuscript, and enhancing readability. The tool was not used to generate research data, perform statistical analyses, or draw scientific conclusions. After using this tool, the author(s) reviewed and edited the content as needed and take(s) full responsibility for the content of the published article.

\appendix

\section{Checkboard Copula}
\label{appen_checkboard_copula}
Suppose the checkerboard copula density is defined by
\[
c_W(u_1,\ldots,u_k)
=
w_{m_1,\ldots,m_k}, \ (u_1,\ldots,u_k)\in B_{m_1,\ldots,m_k},
\]
where \(\mathbf W=\left(w_{m_1,\ldots,m_k}\right)\) is a \(k\)-dimensional array of checkerboard weights.

Then \(c_W\) is a valid copula density if and only if the following
conditions hold:

\begin{enumerate}[(i)]

\item \(w_{m_1,\ldots,m_k}\ge0,\ m_j=1,\ldots,d.\)

\item For every coordinate \(j\in\{1,\ldots,k\}\) and every fixed value \(m_j\in\{1,\ldots,d\}\),
\[
\sum_{m_1=1}^{d}
\cdots
\sum_{m_{j-1}=1}^{d}
\sum_{m_{j+1}=1}^{d}
\cdots
\sum_{m_k=1}^{d}
w_{m_1,\ldots,m_k}
=
d^{\,k-1}.
\]

\end{enumerate}

Since \(c_W\) is piecewise constant, it is nonnegative almost everywhere
if and only if \(w_{m_1,\ldots,m_k}\ge0\) for every checkerboard cell.

Each cell has volume \(d^{-k}\). Therefore,
\[
\int_{[0,1]^k}
c_W(\mathbf u)\,
d\mathbf u
=
\frac1{d^k}
\sum_{m_1=1}^{d}
\cdots
\sum_{m_k=1}^{d}
w_{m_1,\ldots,m_k}.
\]

Using condition (ii),
\[
\begin{aligned}
\sum_{m_1=1}^{d}
\cdots
\sum_{m_k=1}^{d}
w_{m_1,\ldots,m_k}
&=
\sum_{m_1=1}^{d}
\left(
\sum_{m_2=1}^{d}
\cdots
\sum_{m_k=1}^{d}
w_{m_1,\ldots,m_k}
\right)  \\
&=
\sum_{m_1=1}^{d}
d^{k-1}
=
d^k,
\end{aligned}
\]
and hence
\[
\int_{[0,1]^k}
c_W(\mathbf u)\,
d\mathbf u
=
1.
\]

Next, fix a coordinate \(j\) and suppose
\(u_j\in(u_{m_j-1},u_{m_j}]\).
The corresponding marginal density is
\[
\begin{aligned}
c_j(u_j)
&=
\int_{[0,1]^{k-1}}
c_W(u_1,\ldots,u_k)
\prod_{r\neq j}du_r  \\
&=
\frac1{d^{k-1}}
\sum_{m_1=1}^{d}
\cdots
\sum_{m_{j-1}=1}^{d}
\sum_{m_{j+1}=1}^{d}
\cdots
\sum_{m_k=1}^{d}
w_{m_1,\ldots,m_k}.
\end{aligned}
\]

By condition (ii), \(c_j(u_j)=1,\) showing that every marginal distribution is
\(\mathrm{Uniform}(0,1)\). Therefore, \(c_W\) is a valid copula density.

Conversely, suppose that \(c_W\) is a copula density. Since every copula density is nonnegative almost everywhere, condition (i) follows immediately.

Moreover, every copula has standard uniform marginals. Hence, for every coordinate \(j\),
\[
1
=
c_j(u_j)
=
\frac1{d^{k-1}}
\sum_{m_1=1}^{d}
\cdots
\sum_{m_{j-1}=1}^{d}
\sum_{m_{j+1}=1}^{d}
\cdots
\sum_{m_k=1}^{d}
w_{m_1,\ldots,m_k},
\]
which is equivalent to condition (ii).

Therefore, conditions (i) and (ii) are also necessary.

\section{Gaussian Factor Checkerboard Copula}
\label{app:gaussian_checkerboard}

This appendix provides the technical details underlying the Gaussian factor checkerboard copula used in Section~\ref{sec_gau_fac_check_copu}. We first review the Gaussian factor copula representation, then derive the checkerboard weights, and finally prove Proposition~\ref{prop_gaussian_factor}.

\subsection{Gaussian Factor Model}

Let \(\mathbf F=(F_1,\ldots,F_J)^\top \sim N_J(\mathbf0,\mathbf I_J)\) denote the vector of common latent factors.

For each margin \(j=1,\ldots,k\), define \(G_j=\boldsymbol{\psi}_j^\top \mathbf F + \sqrt{1-\|\boldsymbol{\psi}_j\|^2}\, \varepsilon_j,\) where
\( \boldsymbol{\psi}_j =(\psi_{j1},\ldots,\psi_{jJ})^\top,\ \|\boldsymbol{\psi}_j\|^2~<~1\); \(\varepsilon_1,\ldots,\varepsilon_k
\stackrel{\mathrm{iid}}{\sim}
N(0,1)\), and independently of \(\mathbf F.\)

The transformed variables \(U_j=\Phi(G_j),\ j=1,\ldots,k,\) therefore follow a Gaussian factor copula.

Since \( \operatorname{Cov}(G_j,G_\ell)=\boldsymbol{\psi}_j^\top \boldsymbol{\psi}_\ell,\ j\neq\ell,\) the corresponding correlation matrix is 
\[
\mathbf R = \mathbf\Psi \mathbf\Psi^\top +\operatorname{diag}
\left(1-\|\boldsymbol{\psi}_1\|^2,\ldots,1-\|\boldsymbol{\psi}_k\|^2\right),
\]
where \(\mathbf\Psi = (\boldsymbol{\psi}_1,\ldots,\boldsymbol{\psi}_k)^\top \) is the \(k\times J\) loading matrix.

Consequently, the Gaussian factor representation reduces the number of dependence parameters from \(k(k-1)/2\) for an unrestricted Gaussian copula to only \(kJ\).

\subsection{Construction of Checkerboard Weights}

Conditional on \(\mathbf F=\mathbf f,\) \(G_j\mid\mathbf F=\mathbf f\sim N\left(\boldsymbol{\psi}_j^\top\mathbf f,1-\|\boldsymbol{\psi}_j\|^2 \right).\) The variables \(G_1,\ldots,G_k\) are mutually independent.

Hence, \(\mathbb{P}(U_j\le u \mid \mathbf F=\mathbf f) = q(u;\boldsymbol{\psi}_j,\mathbf f),\) where
\begin{equation}
q (u;\boldsymbol{\psi},\mathbf f)=\Phi \left(\frac{\Phi^{-1}(u) - \boldsymbol{\psi}^\top \mathbf f }{\sqrt{1-\|\boldsymbol{\psi}\|^2}}\right),\ q(0;\boldsymbol{\psi},\mathbf f)=0, \ q(1;\boldsymbol{\psi},\mathbf f)=1.
\label{eq:q_function}
\end{equation}

For the checkerboard partition
\[
\Delta_m = \left(\frac{m-1}{d},\frac{m}{d}\right],\ m=1,\ldots,d,
\]
define
\begin{equation}
\omega_{jm}(\mathbf f)=q\left(\frac{m}{d};\boldsymbol{\psi}_j,\mathbf f\right)-q\left(\frac{m-1}{d};\boldsymbol{\psi}_j,\mathbf f\right).
\label{eq:omega}
\end{equation}

Then
\(\mathbb{P}(U_j\in\Delta_m \mid \mathbf F=\mathbf f )=\omega_{jm}(\mathbf f).\) Since the transformed variables are conditionally independent,
\[
\mathbb{P}(U_1\in\Delta_{m_1},\ldots,U_k\in\Delta_{m_k}\mid\mathbf F=\mathbf f)
=
\prod_{j=1}^{k}\omega_{jm_j}(\mathbf f).
\]

Integrating over the latent factor distribution yields

\begin{equation}
\begin{aligned}
&
\mathbb{P}(U_1\in\Delta_{m_1},\ldots,U_k\in\Delta_{m_k})\\
& \qquad \qquad
=\frac{1}{(2\pi)^{J/2}}\int_{\mathbb R^J}\prod_{j=1}^{k}\omega_{jm_j}(\mathbf f)\exp\left(-\frac{1}{2}\mathbf f^\top\mathbf f\right)d\mathbf f.
\end{aligned}
\label{eq:checkerboard_cell_probability}
\end{equation}

The corresponding checkerboard weight is therefore 
\begin{equation}
w_{m_1,\ldots,m_k} =d^k\mathbb{P}(U_1\in\Delta_{m_1},\ldots,U_k\in\Delta_{m_k}).
\label{eq:checkerboard_weight}
\end{equation}

Since the Gaussian copula has standard uniform margins, the tensor \((w_{m_1,\ldots,m_k})\)  automatically satisfies the checkerboard marginal constraints and therefore defines a valid checkerboard copula.

\section{Proof of Propositions} \label{proof_of_pro}
\subsection{Proof of Proposition \ref{prop_eqpmfN}} \label{proof_of_pro1}
\begin{proof}
Conditioning on $\boldsymbol{\Lambda}_i$ gives
\[
\begin{aligned}
&
\mathbb{P}(\mathbf N_i=\mathbf n_i
\mid
\mathbf X_i=\mathbf x_i)
\\
& \qquad  \qquad =
\mathbb{E}\!\left[
\mathbb{P}(
\mathbf N_i=\mathbf n_i
\mid
\boldsymbol{\Lambda}_i,
\mathbf X_i=\mathbf x_i
)
\Bigm|
\mathbf X_i=\mathbf x_i
\right].
\end{aligned}
\]

By conditional independence,
\[
\begin{aligned}  
&\mathbb{P}(
\mathbf N_i=\mathbf n_i
\mid
\boldsymbol{\Lambda}_i=\boldsymbol{\lambda}_i,
\mathbf X_i=\mathbf x_i
)\\
&\qquad \qquad  =
\prod_{j=1}^{k}
\mathbb{P}(
N_{ij}=n_{ij}
\mid
\Lambda_{ij}=\lambda_{ij},
\mathbf X_i=\mathbf x_i
).    
\end{aligned}
\]

Since
\(
N_{ij}\mid
\Lambda_{ij}=\lambda_{ij}
\sim
\mathrm{Poisson}(\lambda_{ij}),
\) we have 
\[
\mathbb{P}(
N_{ij}=n_{ij}
\mid
\Lambda_{ij}=\lambda_{ij},
\mathbf X_i=\mathbf x_i
)
=
\frac{
e^{-\lambda_{ij}}
\lambda_{ij}^{n_{ij}}
}{
\Gamma(n_{ij}+1)
} 
\]

Therefore,
\[
\begin{aligned}
\mathbb{P}(
\mathbf N_i=\mathbf n_i
\mid
\mathbf X_i=\mathbf x_i
)=
\mathbb{E}\!\left[
\prod_{j=1}^{k}
\frac{
e^{-\Lambda_{ij}}
\Lambda_{ij}^{n_{ij}}
}{
\Gamma(n_{ij}+1)
}
\Bigm|
\mathbf X_i=\mathbf x_i
\right].
\end{aligned}
\]

Define
\[
U_{ij}
=
F_j(\Lambda_{ij};\mu_{ij},\phi_j),
\qquad
j=1,\ldots,k.
\]
Since each \(F_j\) is continuous and strictly increasing,
\(
U_{ij}\sim
\mathrm{Uniform}(0,1),
\) and 
\(
\Lambda_{ij}
=
Q_j(U_{ij};\mu_{ij},\phi_j)
\)
Moreover,
\(
(U_{i1},\ldots,U_{ik})^\top
\) has copula density \(c_{\boldsymbol\Lambda} \left(.;\theta\right)\).

Substituting the quantile representation into the previous expectation gives
\[
\begin{aligned}
&
\mathbb{P}(
\mathbf N_i=\mathbf n_i
\mid
\mathbf X_i=\mathbf x_i
)
\\
& \qquad =
\mathbb{E}\!\left[
\prod_{j=1}^{k}
\frac{
\exp\!\left\{
-Q_j(U_{ij};\mu_{ij},\phi_j)
\right\}
Q_j(U_{ij};\mu_{ij},\phi_j)^{n_{ij}}
}{
\Gamma(n_{ij}+1)
}
\right]
\end{aligned}
\]

Finally, integrating with respect to the copula density of
\((U_{i1},\ldots,U_{ik})^\top\)
yields
\[
\begin{aligned}
&
\mathbb{P}(
\mathbf N_i=\mathbf n_i
\mid
\mathbf X_i=\mathbf x_i
)
\\
& \ =
\int_0^1\cdots\int_0^1
\left[
\prod_{j=1}^{k}
\frac{
\exp\!\left\{
-Q_j(u_j;\mu_{ij},\phi_j)
\right\}
Q_j(u_j;\mu_{ij},\phi_j)^{n_{ij}}
}{
\Gamma(n_{ij}+1)
}
\right]
\\
&\hspace{4cm}\times
c_{\boldsymbol\Lambda}(u_1,\ldots,u_k; \theta)
\,du_1\cdots du_k,
\end{aligned}
\]
which establishes \eqref{eq:jointpmf_uniform}.

\end{proof}

\subsection{Proof of Proposition \ref{prop_checkerboard_pmf}} \label{proof_of_pro2}
\begin{proof}

Substituting the checkerboard copula density
\(c_W\)
into Equation~(\ref{eq:jointpmf_uniform}) yields
\begin{equation*}
\begin{aligned}
&
\mathbb{P}(\mathbf N_i=\mathbf n_i
\mid
\mathbf X_i=\mathbf x_i)
\\
& \qquad =
\int_{[0,1]^k}
\left[
\prod_{j=1}^{k}
\frac{
\exp\!\left\{
-Q_j(u_j;\mu_{ij},\phi_j)
\right\}
Q_j(u_j;\mu_{ij},\phi_j)^{n_{ij}}
}{
\Gamma(n_{ij}+1)
}
\right] \\
&\hspace{4cm}\times
c_{W}(u_1,\ldots,u_k)
\,du_1\cdots du_k .
\end{aligned}
\end{equation*}

Since the checkerboard cells
\[
H_{m_1,\ldots,m_k}
=
(u_{m_1-1},u_{m_1}]
\times\cdots\times
(u_{m_k-1},u_{m_k}],
\]
form a partition of the unit hypercube (up to sets of Lebesgue measure
zero), the above integral can be decomposed as
\[
\begin{aligned}
&
\mathbb{P}(\mathbf N_i=\mathbf n_i
\mid
\mathbf X_i=\mathbf x_i)
\\
&\qquad=
\sum_{m_1=1}^{d}
\cdots
\sum_{m_k=1}^{d}
\\
& \quad \quad \times
\int_{H_{m_1,\ldots,m_k}}
\left[
\prod_{j=1}^{k}
\frac{
\exp\!\left\{
-Q_j(u_j;\mu_{ij},\phi_j)
\right\}
Q_j(u_j;\mu_{ij},\phi_j)^{n_{ij}}
}{
\Gamma(n_{ij}+1)
}
\right] \\
&\hspace{4cm}\times
c_{W}(u_1,\ldots,u_k)
\,du_1\cdots du_k .
\end{aligned}
\]

By the definition of the checkerboard copula density,
\[
c_W(u_1,\ldots,u_k)
=
w_{m_1,\ldots,m_k},
\qquad
(u_1,\ldots,u_k)\in
H_{m_1,\ldots,m_k},
\]
and therefore
\[
\begin{aligned}
&
\mathbb{P}(\mathbf N_i=\mathbf n_i
\mid
\mathbf X_i=\mathbf x_i)
\\
&\qquad=
\sum_{m_1=1}^{d}
\cdots
\sum_{m_k=1}^{d}
w_{m_1,\ldots,m_k}
\int_{u_{m_1-1}}^{u_{m_1}}
\cdots
\int_{u_{m_k-1}}^{u_{m_k}}
\\
&\qquad \ \times
\prod_{j=1}^{k}
\frac{
\exp\!\left\{
-Q_j(u_j;\mu_{ij},\phi_j)
\right\}
Q_j(u_j;\mu_{ij},\phi_j)^{n_{ij}}
}{
\Gamma(n_{ij}+1)
}
\,du_1\cdots du_k .
\end{aligned}
\]

Since the integrand is a product of functions depending only on individual variables, Fubini's theorem implies that the multidimensional integral factorizes as
\[
\begin{aligned}
&
\int_{u_{m_1-1}}^{u_{m_1}}
\cdots
\int_{u_{m_k-1}}^{u_{m_k}}
\\
& \ \times
\prod_{j=1}^{k}
\frac{
\exp\!\left\{
-Q_j(u_j;\mu_{ij},\phi_j)
\right\}
Q_j(u_j;\mu_{ij},\phi_j)^{n_{ij}}
}{
\Gamma(n_{ij}+1)
}
\,du_1\cdots du_k
\\
& \quad =
\prod_{j=1}^{k}
\int_{u_{m_j-1}}^{u_{m_j}}
\frac{
\exp\!\left\{
-Q_j(u;\mu_{ij},\phi_j)
\right\}
Q_j(u;\mu_{ij},\phi_j)^{n_{ij}}
}{
\Gamma(n_{ij}+1)
}
\,du
\\
& \quad =
\prod_{j=1}^{k}
I_{ij}^{(m_j)},
\end{aligned}
\]
where \(I_{ij}^{(m)}\) is defined in
(\ref{eq:cell_probability}).

Substituting the above identity into the previous expression gives
\[
\mathbb{P}(\mathbf N_i=\mathbf n_i
\mid
\mathbf X_i=\mathbf x_i)
=
\sum_{m_1=1}^{d}
\cdots
\sum_{m_k=1}^{d}
w_{m_1,\ldots,m_k}
\prod_{j=1}^{k}
I_{ij}^{(m_j)},
\]
which is precisely Equation~(\ref{eq:checkerboard_pmf}).

\end{proof}

\subsection{Proof of Proposition~\ref{prop_gaussian_factor}} \label{proof_of_pro_3}

Conditional on the common factor vector \(\mathbf F=\mathbf f\), the checkerboard copula in Proposition~\ref{prop_checkerboard_pmf} becomes 
\[
\mathbb{P}(\mathbf N_i=\mathbf n_i \mid \mathbf X_i=\mathbf x_i,\mathbf F=\mathbf f )=\prod_{j=1}^{k}
p_{ij}^{N} (\mathbf f;n_{ij}),\]
where
\[
p_{ij}^{N}(\mathbf f;n)=\sum_{m=1}^{d}I_{ij}^{(m)}(n)\,\omega_{jm}(\mathbf f).
\]

Finally, integrating with respect to the multivariate standard normal density of
\(\mathbf F\) gives
\[
\begin{aligned}
&
\mathbb{P}(\mathbf N_i=\mathbf n_i\mid\mathbf X_i=\mathbf x_i)
\\
& \qquad =
\frac{1}{(2\pi)^{J/2}}\int_{\mathbb R^J}\prod_{j=1}^{k}p_{ij}^{N}(\mathbf f;n_{ij})
\exp \left(-\frac12\mathbf f^\top\mathbf f\right)d\mathbf f,    
\end{aligned}
\]
which is exactly the representation stated in Proposition~\ref{prop_gaussian_factor}. \hfill $\square$

\subsection{Proof of Proposition \ref{prop_joint_pmf_zero_inflated}} \label{proof_of_pro_4}
\begin{proof}

Conditional on \(\mathbf X_i=\mathbf x_i\), the observed response satisfies
\[
Y_{ij}=B_{ij}N_{ij},
\qquad j=1,\ldots,k.
\]
For a given realization \(\mathbf y_i\), any compatible configuration
\(\mathbf b\) must satisfy \(b_j=1\) whenever \(y_{ij}>0\). Therefore,
the set of admissible configurations of the structural-zero indicators is
\[
\mathcal B(\mathbf y_i)
=
\left\{
\mathbf b\in\{0,1\}^k:
b_j=1
\text{ whenever }
y_{ij}>0
\right\}.
\]

By the law of total probability,
\begin{equation}
\begin{aligned}
&
\mathbb P
\left(
\mathbf Y_i=\mathbf y_i
\mid
\mathbf X_i=\mathbf x_i
\right)
\\
& \qquad =
\sum_{\mathbf b\in\mathcal B(\mathbf y_i)}
\mathbb P
\left(
\mathbf Y_i=\mathbf y_i,
\mathbf B_i=\mathbf b
\mid
\mathbf X_i=\mathbf x_i
\right).
\end{aligned}
\label{eq:proof_total_probability}
\end{equation}
Since \(\mathbf B_i\) and \(\mathbf N_i\) are conditionally independent
given \(\mathbf X_i\), it follows that
\begin{equation}
\begin{aligned}
&
\mathbb P
\left(
\mathbf Y_i=\mathbf y_i,
\mathbf B_i=\mathbf b
\mid
\mathbf X_i=\mathbf x_i
\right)
\\
& \qquad =
p_{\mathbf B,i}(\mathbf b)
\,
\mathbb P
\left(
B_{ij}N_{ij}=y_{ij},
\ j=1,\ldots,k
\mid
\mathbf B_i=\mathbf b,
\mathbf X_i=\mathbf x_i
\right)
\\
& \qquad =
p_{\mathbf B,i}(\mathbf b)
\,
\mathbb P
\left(
N_{ij}=y_{ij}\ \text{for all }j\text{ such that }b_j=1
\mid
\mathbf X_i=\mathbf x_i
\right),
\end{aligned}
\label{eq:proof_conditional_independence}
\end{equation}
where no restriction is imposed on \(N_{ij}\) when \(b_j=0\).

Under the checkerboard copula construction, the joint distribution of
\(\mathbf N_i\), conditional on \(\mathbf X_i=\mathbf x_i\), can be
represented as a mixture over the \(d^k\) checkerboard cells. Consequently,
for a fixed configuration \(\mathbf b\),
\begin{equation}
\begin{aligned}
&
\mathbb P
\left(
N_{ij}=y_{ij}\ \text{for all }j\text{ such that }b_j=1
\mid
\mathbf X_i=\mathbf x_i
\right)
\\
& \ \qquad=
\sum_{m_1=1}^{d}
\cdots
\sum_{m_k=1}^{d}
w_{m_1,\ldots,m_k}
\prod_{j=1}^{k}
\widetilde I_{ij}^{(m_j)}(y_{ij},b_j),
\end{aligned}
\label{eq:proof_checkerboard_expansion}
\end{equation}
where, if \(b_j=1\), the corresponding count must satisfy
\(N_{ij}=y_{ij}\), and hence its contribution is
\[
I_{ij}^{(m_j)}(y_{ij}).
\]
If \(b_j=0\), the value of \(N_{ij}\) is unrestricted. Its contribution is
therefore obtained by summing over all possible counts:
\begin{equation}
\begin{aligned}
\sum_{n=0}^{\infty}I_{ij}^{(m_j)}(n)
&=
\int_{(m_j-1)/d}^{m_j/d}
\sum_{n=0}^{\infty}
\frac{
\exp\{-Q_{\Lambda_{ij}}(u)\}
Q_{\Lambda_{ij}}(u)^n
}{
n!
}
\,du
\\
&=
\int_{(m_j-1)/d}^{m_j/d}1\,du
=
\frac{1}{d}\cdot
\end{aligned}
\label{eq:proof_unrestricted_count}
\end{equation}
Thus, for \(b\in\{0,1\}\),
\[
\widetilde I_{ij}^{(m)}(y,b)
=
b\,I_{ij}^{(m)}(y)
+
(1-b)\frac{1}{d}\cdot
\]

Substituting \eqref{eq:proof_checkerboard_expansion} into
\eqref{eq:proof_conditional_independence}, and then into
\eqref{eq:proof_total_probability}, gives
\[
\begin{aligned}
&
\mathbb P
\left(
\mathbf Y_i=\mathbf y_i
\mid
\mathbf X_i=\mathbf x_i
\right)
\\
& \qquad =
\sum_{\mathbf b\in\mathcal B(\mathbf y_i)}
p_{\mathbf B,i}(\mathbf b)
\sum_{m_1=1}^{d}
\cdots
\sum_{m_k=1}^{d}
w_{m_1,\ldots,m_k}
\prod_{j=1}^{k}
\widetilde I_{ij}^{(m_j)}(y_{ij},b_j).
\end{aligned}
\]
Since the checkerboard weights do not depend on \(\mathbf b\), the finite
sums can be interchanged, yielding
\[
\begin{aligned}
&
\mathbb P
\left(
\mathbf Y_i=\mathbf y_i
\mid
\mathbf X_i=\mathbf x_i
\right)
\\
& \qquad =
\sum_{m_1=1}^{d}
\cdots
\sum_{m_k=1}^{d}
w_{m_1,\ldots,m_k}
\sum_{\mathbf b\in\mathcal B(\mathbf y_i)}
p_{\mathbf B,i}(\mathbf b)
\prod_{j=1}^{k}
\widetilde I_{ij}^{(m_j)}(y_{ij},b_j),
\end{aligned}
\]
which proves \eqref{eq:joint_pmf_Y}.

It remains to derive \(p_{\mathbf B,i}(\mathbf b)\). By construction,
\(
B_{ij}
=
\mathbb I(U_{ij}\leq\pi_{ij}),
\) so that 
\[
B_{ij}=1
\quad\Longleftrightarrow\quad
U_{ij}\in(0,\pi_{ij}],
\]
whereas
\[
B_{ij}=0
\quad\Longleftrightarrow\quad
U_{ij}\in(\pi_{ij},1].
\]
For a given \(b_j\), define the lower and upper endpoints by
\[
\ell_{ij}(b_j)=(1-b_j)\pi_{ij},
\qquad
u_{ij}(b_j)=(1-b_j)+b_j\pi_{ij}.
\]
Hence,
\[
p_{\mathbf B,i}(\mathbf b)
=
\mathbb P
\left(
\ell_{ij}(b_j)<U_{ij}\leq u_{ij}(b_j),
\ j=1,\ldots,k
\mid
\mathbf X_i=\mathbf x_i
\right).
\]
Applying the multivariate inclusion--exclusion formula to this rectangle
gives
\[
p_{\mathbf B,i}(\mathbf b)
=
\sum_{\boldsymbol{\delta}\in\{0,1\}^{k}}
(-1)^{k-\sum_{j=1}^{k}\delta_j}
C_{\mathbf B}
\left(
v_{i1}(\delta_1,b_1),
\ldots,
v_{ik}(\delta_k,b_k);
\boldsymbol{\theta}_{\mathbf B}
\right),
\]
where
\[
v_{ij}(\delta,b)
=
(1-\delta)\ell_{ij}(b)
+
\delta u_{ij}(b).
\]
Substituting the expressions for \(\ell_{ij}(b)\) and \(u_{ij}(b)\)
yields
\[
v_{ij}(\delta,b)
=
(1-\delta)(1-b)\pi_{ij}
+
\delta\left[(1-b)+b\pi_{ij}\right],
\]
which completes the proof.

\end{proof}

\subsection{Proof of Proposition \ref{prop_ZI_likelihood}} \label{proof_of_pro_5}

\begin{proof}

For fixed
\(\mathbf z_B\in\mathbb R^{J_B}\) and
\(\mathbf z_{\Lambda}\in\mathbb R^{J_{\Lambda}}\), consider the conditional
probability
\[
\mathbb P
\left(
Y_{ij}=y
\mid
\mathbf X_i=\mathbf x_i,
\mathbf Z_B=\mathbf z_B,
\mathbf Z_{\Lambda}=\mathbf z_{\Lambda}
\right).
\]

Since \(Y_{ij}=B_{ij}N_{ij}\), the event \(Y_{ij}=0\) can occur in two
mutually exclusive ways:
\[
\{B_{ij}=0\}
\qquad\text{or}\qquad
\{B_{ij}=1,N_{ij}=0\}.
\]

Therefore, using the conditional independence of \(B_{ij}\) and \(N_{ij}\),
\begin{align}
&
\mathbb P
\left(
Y_{ij}=0
\mid
\mathbf X_i=\mathbf x_i,
\mathbf Z_B=\mathbf z_B,
\mathbf Z_{\Lambda}=\mathbf z_{\Lambda}
\right)
\nonumber\\
&=
\mathbb P
\left(
B_{ij}=0
\mid
\mathbf X_i=\mathbf x_i,
\mathbf Z_B=\mathbf z_B
\right)
\nonumber\\
&\ \ +
\mathbb P
\left(
B_{ij}=1
\mid
\mathbf X_i=\mathbf x_i,
\mathbf Z_B=\mathbf z_B
\right)
\mathbb P
\left(
N_{ij}=0
\mid
\mathbf X_i=\mathbf x_i,
\mathbf Z_{\Lambda}=\mathbf z_{\Lambda}
\right)
\nonumber\\
&=
1-p_{ij}^{B}(\mathbf z_B)
+
p_{ij}^{B}(\mathbf z_B)
p_{ij}^{N}(\mathbf z_{\Lambda};0).
\label{eq:proof_conditional_zero_probability}
\end{align}

For \(y>0\), the equality \(Y_{ij}=y\) requires both
\(B_{ij}=1\) and \(N_{ij}=y\). Hence,
\begin{align}
&
\mathbb P
\left(
Y_{ij}=y
\mid
\mathbf X_i=\mathbf x_i,
\mathbf Z_B=\mathbf z_B,
\mathbf Z_{\Lambda}=\mathbf z_{\Lambda}
\right)
\nonumber\\
&=
\mathbb P
\left(
B_{ij}=1
\mid
\mathbf X_i=\mathbf x_i,
\mathbf Z_B=\mathbf z_B
\right)
\mathbb P
\left(
N_{ij}=y
\mid
\mathbf X_i=\mathbf x_i,
\mathbf Z_{\Lambda}=\mathbf z_{\Lambda}
\right)
\nonumber\\
&=
p_{ij}^{B}(\mathbf z_B)
p_{ij}^{N}(\mathbf z_{\Lambda};y).
\label{eq:proof_conditional_positive_probability}
\end{align}

Under the Gaussian factor construction for the structural-zero component,
the latent uniform variable associated with \(B_{ij}\) can be represented as
\[
U_{B,ij}
=
\Phi\!\left(
\boldsymbol{\psi}_{B,j}^{\top}\mathbf Z_B
+
\sqrt{1-\|\boldsymbol{\psi}_{B,j}\|^2}\,
\varepsilon_{B,ij}
\right),
\]
where
\(\varepsilon_{B,ij}\sim N(0,1)\), independently across \(j\) and
independently of \(\mathbf Z_B\), and
\[
B_{ij}
=
\mathbb I(U_{B,ij}\leq\pi_{ij}).
\]
Consequently,
\begin{align}
p_{ij}^{B}(\mathbf z_B)
&=
\mathbb P
\left(
U_{B,ij}\leq\pi_{ij}
\mid
\mathbf Z_B=\mathbf z_B,
\mathbf X_i=\mathbf x_i
\right)
\nonumber\\
&=
\mathbb P\left(
\varepsilon_{B,ij}
\leq
\frac{
\Phi^{-1}(\pi_{ij})
-
\boldsymbol{\psi}_{B,j}^{\top}\mathbf z_B
}{
\sqrt{1-\|\boldsymbol{\psi}_{B,j}\|^2}
}
\right)
\nonumber\\
&=
\Phi\!\left(
\frac{
\Phi^{-1}(\pi_{ij})
-
\boldsymbol{\psi}_{B,j}^{\top}\mathbf z_B
}{
\sqrt{1-\|\boldsymbol{\psi}_{B,j}\|^2}
}
\right).
\end{align}

Similarly, under the Gaussian factor construction for the latent-intensity
checkerboard copula, conditional on
\(\mathbf Z_{\Lambda}=\mathbf z_{\Lambda}\), the probability that the
latent uniform variable associated with the \(j\)-th intensity lies in the
\(m\)-th checkerboard interval is
\[
\omega_{\Lambda,jm}(\mathbf z_{\Lambda})
=
q\!\left(
\frac{m}{d};
\boldsymbol{\psi}_{\Lambda,j},
\mathbf z_{\Lambda}
\right)
-
q\!\left(
\frac{m-1}{d};
\boldsymbol{\psi}_{\Lambda,j},
\mathbf z_{\Lambda}
\right).
\]
It follows from Proposition~\ref{prop_gaussian_factor} that
\[
p_{ij}^{N}(\mathbf z_{\Lambda};n)
=
\sum_{m=1}^{d}
I_{ij}^{(m)}(n)
\omega_{\Lambda,jm}(\mathbf z_{\Lambda}).
\]

Finally, conditional on
\(\mathbf X_i=\mathbf x_i\),
\(\mathbf Z_B=\mathbf z_B\), and
\(\mathbf Z_{\Lambda}=\mathbf z_{\Lambda}\), the responses
\(Y_{i1},\ldots,Y_{ik}\) are independent. Thus,
\begin{align}
&
\mathbb P
\left(
\mathbf Y_i=\mathbf y_i
\mid
\mathbf X_i=\mathbf x_i,
\mathbf Z_B=\mathbf z_B,
\mathbf Z_{\Lambda}=\mathbf z_{\Lambda}
\right)
\nonumber\\
& \qquad =
\prod_{j=1}^{k}
L_{ij}(\mathbf z_B,\mathbf z_{\Lambda};y_{ij}).
\label{eq:proof_conditional_joint_Y}
\end{align}

Integrating \eqref{eq:proof_conditional_joint_Y} with respect to the
independent standard Gaussian densities of \(\mathbf Z_B\) and
\(\mathbf Z_{\Lambda}\) gives
\begin{align*}
&
\mathbb P
\left(
\mathbf Y_i=\mathbf y_i
\mid
\mathbf X_i=\mathbf x_i
\right)
\\
& \qquad =
\frac{1}{(2\pi)^{(J_B+J_{\Lambda})/2}}
\int_{\mathbb R^{J_B}}
\int_{\mathbb R^{J_{\Lambda}}}
\prod_{j=1}^{k}
L_{ij}(\mathbf z_B,\mathbf z_{\Lambda};y_{ij})
\\
&\qquad\qquad \qquad\times
\exp\!\left(
-\frac{1}{2}\mathbf z_B^\top\mathbf z_B
-\frac{1}{2}\mathbf z_{\Lambda}^\top\mathbf z_{\Lambda}
\right)
\,d\mathbf z_{\Lambda}\,d\mathbf z_B,
\end{align*}
which proves \eqref{eq:ZI_likelihood}.
\end{proof}

\section{Mixing Distributions}

\label{app_distribution_for_lambda}

This appendix summarizes the mixing distributions considered throughout the paper. Each distribution is parameterized in its natural form while preserving a common regression specification based on the conditional mean \(\mu\). Specifically, the distributional parameters are expressed as functions of \(\mu\) together with a dispersion parameter \(\phi\), so that all candidate models share the same interpretation of the regression coefficients. By default, each distribution adopts its natural variance function implied by this parameterization. When appropriate, the proposed framework also allows alternative variance functions of the form
\[
\mathbb V(\Lambda)=\phi\mu^{p},
\qquad
p\in\{1,2,3\},
\]
through suitable reparameterization of the underlying distributional parameters. For each mixing distribution, we provide the corresponding density, cumulative distribution, quantile, and variance functions under the parameterizations adopted in the proposed methodology.

\subsection{Gamma mixing distribution}

Throughout this paper, the Gamma mixing distribution is parameterized by its
mean \(\mu>0\) and dispersion parameter \(\phi>0\) through the variance
function
\begin{equation}
\mathbb E(\Lambda)=\mu,
\qquad
\mathbb V(\Lambda)=\mu^2\phi.
\label{eq:gamma_var_mu2}
\end{equation}
Equivalently, if
\(\Lambda\sim\mathrm{Gamma}(\alpha,\beta)\),
where \(\alpha\) and \(\beta\) denote the shape and rate parameters,
respectively, then
\[
\alpha=\frac{1}{\phi},
\qquad
\beta=\frac{1}{\mu\phi}.
\]

More generally, the proposed implementation allows the variance function to
be specified as
\begin{equation}
\mathbb V(\Lambda)=\mu^{p}\phi,
\qquad
p\in\{1,2,3\},
\label{eq:gamma_general_variance}
\end{equation}
which includes \eqref{eq:gamma_var_mu2} as the default choice (\(p=2\)).
Under \eqref{eq:gamma_general_variance}, the corresponding shape and rate
parameters become
\[
\alpha=\frac{\mu^{2-p}}{\phi},
\qquad
\beta=\frac{\mu^{1-p}}{\phi},
\]
thereby preserving the prescribed mean and variance.

Under the default parameterization (\(p=2\)), the density function is
\[
f_\Lambda(\lambda;\mu,\phi)
=
\frac{(1/(\mu\phi))^{1/\phi}}
{\Gamma(1/\phi)}
\lambda^{1/\phi-1}
\exp\!\left(
-\frac{\lambda}{\mu\phi}
\right),
\qquad
\lambda>0,
\]
the cumulative distribution function is
\[
F_\Lambda(\lambda;\mu,\phi)
=
\frac{
\gamma\!\left(
1/\phi,
\lambda/(\mu\phi)
\right)
}{
\Gamma(1/\phi)
},
\]
and the quantile function is
\[
Q_\Lambda(u;\mu,\phi)
=
\mu\phi\,
P^{-1}\!\left(
1/\phi,
u
\right),
\qquad
0<u<1,
\]
where
\[
\gamma(a,x)
=
\int_0^x
t^{a-1}e^{-t}\,dt
\]
is the lower incomplete Gamma function, and
\(P^{-1}(a,u)\) denotes the inverse of the regularized incomplete Gamma
function with respect to its second argument.

\subsection{Lognormal mixing distribution}

Throughout this paper, the Lognormal mixing distribution is parameterized by
its mean \(\mu>0\) and dispersion parameter \(\phi>0\) through the variance
function
\begin{equation}
\mathbb E(\Lambda)=\mu,
\qquad
\mathbb V(\Lambda)=\mu^2\phi.
\label{eq:lnorm_var_mu2}
\end{equation}
Equivalently, if
\[
\log(\Lambda)\sim N(\eta,\sigma^2),
\]
then
\[
\sigma^2=\log(1+\phi),
\qquad
\eta=\log(\mu)-\frac12\log(1+\phi).
\]

More generally, the proposed implementation allows the variance function to
be specified as
\begin{equation}
\mathbb V(\Lambda)=\mu^{p}\phi,
\qquad
p\in\{1,2,3\},
\label{eq:lnorm_general_variance}
\end{equation}
which includes \eqref{eq:lnorm_var_mu2} as the default choice (\(p=2\)).
Under \eqref{eq:lnorm_general_variance}, the corresponding Lognormal
parameters become
\[
\sigma^2
=
\log\!\left(
1+\phi\mu^{p-2}
\right),
\qquad
\eta
=
\log(\mu)-\frac12\sigma^2,
\]
thereby preserving the prescribed mean and variance.

Under the default parameterization (\(p=2\)), the density function is
\[
\begin{aligned}
&
f_\Lambda(\lambda;\mu,\phi)
=
\frac{1}
{\lambda\sqrt{2\pi\log(1+\phi)}}\\
& \qquad \times 
\exp\!\left[
-
\frac{
\left(
\log\lambda
-
\log\mu
+
\frac12\log(1+\phi)
\right)^2
}
{2\log(1+\phi)}
\right],
\qquad
\lambda>0,    
\end{aligned}
\]
the cumulative distribution function is
\[
F_\Lambda(\lambda;\mu,\phi)
=
\Phi\!\left(
\frac{
\log\lambda
-
\log\mu
+
\frac12\log(1+\phi)
}
{\sqrt{\log(1+\phi)}}
\right),
\]
and the quantile function is
\[
\begin{aligned}
&
Q_\Lambda(u;\mu,\phi)
\\
& \qquad=
\exp\!\left(
\log\mu
-\frac12\log(1+\phi)
+
\sqrt{\log(1+\phi)}\,
\Phi^{-1}(u)
\right),
\\
& \qquad\qquad \qquad\qquad\qquad\qquad\qquad\qquad\qquad\qquad 0<u<1,    
\end{aligned}
\]
where \(\Phi\) and \(\Phi^{-1}\) denote the standard normal distribution function and quantile function, respectively.

\subsection{Pareto mixing distribution}

Throughout this paper, the Pareto Type II (Lomax) mixing distribution is arameterized by its conditional mean \(\mu>0\) and dispersion parameter \(\phi>0\). Unlike the other mixing distributions considered in this paper, the Pareto distribution admits a valid mean--variance parameterization only when the variance is proportional to \(\mu^2\). Consequently, the variance function is fixed as
\begin{equation}
\mathbb E(\Lambda)=\mu,
\qquad
\mathbb V(\Lambda)=\mu^2(1+2\phi),
\label{eq:pareto_var_mu2}
\end{equation}
and no alternative variance powers are considered.

Let
\[
\Lambda\sim\operatorname{Lomax}(\alpha,\beta),
\]
where \(\alpha>2\) and \(\beta>0\) denote the conventional shape and scale
parameters. The Lomax distribution has mean and variance
\[
\mathbb E(\Lambda)
=
\frac{\beta}{\alpha-1},
\qquad
\mathbb V(\Lambda)
=
\frac{\beta^2\alpha}
     {(\alpha-1)^2(\alpha-2)}.
\]

Matching these moments with
\eqref{eq:pareto_var_mu2} yields
\begin{equation}
\alpha
=
2+\frac1\phi,
\qquad
\beta
=
\mu\left(1+\frac1\phi\right).
\label{eq:pareto_parameters_mu_phi}
\end{equation}

Indeed,
\[
\mathbb E(\Lambda)
=
\frac{
\mu(1+1/\phi)
}{
1+1/\phi
}
=
\mu,
\]
and
\[
\begin{aligned}
\mathbb V(\Lambda)
&=
\frac{
\mu^2(1+1/\phi)^2(2+1/\phi)
}{
(1+1/\phi)^2(1/\phi)
}
\\
&=
\mu^2(1+2\phi).
\end{aligned}
\]

The density function is
\begin{equation}
f_\Lambda(\lambda;\mu,\phi)
=
\frac{
\left(2+\frac1\phi\right)
\left[
\mu\left(1+\frac1\phi\right)
\right]^{2+1/\phi}
}{
\left[
\lambda+
\mu\left(1+\frac1\phi\right)
\right]^{3+1/\phi}
},
\qquad
\lambda>0.
\label{eq:pareto_density_mu_phi}
\end{equation}

The cumulative distribution function is
\begin{equation}
F_\Lambda(\lambda;\mu,\phi)
=
1-
\left(
\frac{
\mu(1+1/\phi)
}{
\lambda+\mu(1+1/\phi)
}
\right)^{2+1/\phi},
\qquad
\lambda\ge0,
\label{eq:pareto_cdf_mu_phi}
\end{equation}
and the corresponding quantile function is
\begin{equation}
Q_\Lambda(u;\mu,\phi)
=
\mu\left(1+\frac1\phi\right)
\left[
(1-u)^{-1/(2+1/\phi)}
-1
\right],
\qquad
0<u<1.
\label{eq:pareto_quantile_mu_phi}
\end{equation}

The variance function in
\eqref{eq:pareto_var_mu2} is intentionally fixed throughout this paper.
For the Lomax distribution,
\[
\frac{\mathbb V(\Lambda)}
     {\mathbb E(\Lambda)^2}
=
\frac{\alpha}{\alpha-2},
\]
which is always greater than one for every finite-variance distribution.
Consequently, a general variance specification of the form
\(
\mathbb V(\Lambda)=\phi\mu^p
\)
or
\(
\mathbb V(\Lambda)=\mu^p(1+2\phi)
\)
cannot be valid for arbitrary values of \(\mu\) unless \(p=2\). Fixing \(p=2\) therefore provides the unique mean--variance parameterization that remains valid over the entire parameter space while preserving the simple closed-form expressions for the density, distribution and quantile functions.

\subsection{Inverse Gaussian mixing distribution}

Throughout this paper, the Inverse Gaussian mixing distribution is parameterized by its mean \(\mu>0\) and dispersion parameter \(\phi>0\) through the variance function
\begin{equation}
\mathbb E(\Lambda)=\mu,
\qquad
\mathbb V(\Lambda)=\mu^3\phi.
\label{eq:ig_var_mu3}
\end{equation}
Equivalently, if
\(\Lambda\sim\mathrm{IG}(\mu,\lambda)\),
where \(\lambda>0\) denotes the shape parameter, then
\[
\lambda=\frac{1}{\phi}.
\]

More generally, the proposed implementation allows the variance function to be specified as
\begin{equation}
\mathbb V(\Lambda)=\mu^{p}\phi,
\qquad
p\in\{1,2,3\},
\label{eq:ig_general_variance}
\end{equation}
which includes \eqref{eq:ig_var_mu3} as the default choice (\(p=3\)).
Under \eqref{eq:ig_general_variance}, the corresponding shape parameter is
\[
\lambda
=
\frac{\mu^{3-p}}{\phi},
\]
thereby preserving the prescribed mean and variance.

Under the default parameterization (\(p=3\)), the density function is
\[
f_\Lambda(\lambda;\mu,\phi)
=
\sqrt{\frac{1}{2\pi\phi\lambda^{3}}}
\exp\!\left[
-
\frac{(\lambda-\mu)^2}
{2\phi\mu^2\lambda}
\right],
\qquad
\lambda>0,
\]
the cumulative distribution function is
\[
F_\Lambda(\lambda;\mu,\phi)
=
\Phi\!\left(
\sqrt{\frac{1}{\phi\lambda}}
\left(
\frac{\lambda}{\mu}-1
\right)
\right)
+
\exp\!\left(
\frac{2}{\phi\mu}
\right)
\Phi\!\left(
-
\sqrt{\frac{1}{\phi\lambda}}
\left(
\frac{\lambda}{\mu}+1
\right)
\right),
\]
and the quantile function
\[
Q_\Lambda(u;\mu,\phi)
=
F_\Lambda^{-1}(u;\mu,\phi),
\qquad
0<u<1,
\]
is obtained numerically by inverting the cumulative distribution function,
where \(\Phi\) denotes the standard normal distribution function.

\subsection{Weibull mixing distribution}

Throughout this paper, the Weibull mixing distribution is parameterized by
its mean \(\mu>0\) and dispersion parameter \(\phi>0\) through the variance
function
\begin{equation}
\mathbb E(\Lambda)=\mu,
\qquad
\mathbb V(\Lambda)=\mu^2\phi.
\label{eq:weibull_var_mu2}
\end{equation}

Let
\(
\Lambda\sim\operatorname{Weibull}(\kappa,\beta),
\) where \(\kappa>0\) and \(\beta>0\) are the conventional shape and scale parameters, respectively. To distinguish the dispersion parameter \(\phi\) in \eqref{eq:weibull_var_mu2} from the parameter entering the conventional Weibull shape, define \(\phi_W>0\) as the unique solution of 
\begin{equation}
\frac{\Gamma(1+2\phi_W)}
{\Gamma(1+\phi_W)^2}
-1
=
\phi.
\label{eq:weibull_phi_equation}
\end{equation}

The Weibull shape and scale parameters are then specified as
\begin{equation}
\kappa
=
\frac{1}{\phi_W},
\qquad
\beta
=
\frac{\mu}{\Gamma(1+\phi_W)}
\label{eq:weibull_parameters_mu_phi}
\end{equation}

Indeed, under the conventional Weibull parameterization,
\[
\begin{aligned} 
\mathbb E(\Lambda)
&=
\beta\,
\Gamma\!\left(1+\frac{1}{\kappa}\right),
\\
\mathbb V(\Lambda)
&=
\beta^2
\left[
\Gamma\!\left(1+\frac{2}{\kappa}\right)
-
\Gamma\!\left(1+\frac{1}{\kappa}\right)^2
\right]
\end{aligned}
\]
Since \(1/\kappa=\phi_W\), substituting
\eqref{eq:weibull_parameters_mu_phi} gives
\[
\mathbb E(\Lambda)
=
\frac{\mu}{\Gamma(1+\phi_W)}
\Gamma(1+\phi_W)
=
\mu,
\]
and
\[
\begin{aligned}
\mathbb V(\Lambda)
&=
\frac{\mu^2}{\Gamma(1+\phi_W)^2}
\left[
\Gamma(1+2\phi_W)
-
\Gamma(1+\phi_W)^2
\right]
\\
&=
\mu^2
\left[
\frac{\Gamma(1+2\phi_W)}
{\Gamma(1+\phi_W)^2}
-1
\right]
\\
&=
\mu^2\phi,
\end{aligned}
\]
where the last equality follows from
\eqref{eq:weibull_phi_equation}.

Therefore, \(\phi\) is the multiplicative dispersion parameter in
\eqref{eq:weibull_var_mu2}, whereas \(\phi_W\) is an auxiliary
shape-related parameter determined implicitly by \(\phi\). Equivalently,
if
\[
h_W(a)
=
\frac{\Gamma(1+2a)}
{\Gamma(1+a)^2}
-1,
\qquad a>0,
\]
then
\[
\phi_W=h_W^{-1}(\phi).
\]

Under the \((\mu,\phi)\) parameterization, the density function can be
written as
\begin{equation}
\begin{aligned}
f_\Lambda(\lambda;\mu,\phi)
&=
\frac{1}{\phi_W}
\left[
\frac{\Gamma(1+\phi_W)}{\mu}
\right]^{1/\phi_W}
\lambda^{1/\phi_W-1}
\\
&\quad\times
\exp\!\left\{
-
\left[
\frac{\Gamma(1+\phi_W)\lambda}{\mu}
\right]^{1/\phi_W}
\right\},
\ 
\lambda>0,
\end{aligned}
\label{eq:weibull_density_mu_phi}
\end{equation}
where \(\phi_W=h_W^{-1}(\phi)\).

The corresponding cumulative distribution function is
\begin{equation}
F_\Lambda(\lambda;\mu,\phi)
=
1-
\exp\!\left\{
-
\left[
\frac{\Gamma(1+\phi_W)\lambda}{\mu}
\right]^{1/\phi_W}
\right\},
\ 
\lambda>0,
\label{eq:weibull_cdf_mu_phi}
\end{equation}
and the quantile function is
\begin{equation}
Q_\Lambda(u;\mu,\phi)
=
\frac{\mu}{\Gamma(1+\phi_W)}
\left[-\log(1-u)\right]^{\phi_W},
\qquad
0<u<1.
\label{eq:weibull_quantile_mu_phi}
\end{equation}

More generally, the proposed implementation allows the variance function
to be specified as
\begin{equation}
\mathbb V(\Lambda)
=
\mu^p\phi,
\qquad
p\in\{1,2,3\}.
\label{eq:weibull_general_variance}
\end{equation}
In this case, the squared coefficient of variation is
\[
\frac{\mathbb V(\Lambda)}
{\mathbb E(\Lambda)^2}
=
\phi\mu^{p-2}.
\]
Accordingly, define the observation-specific shape-related parameter
\(\phi_W(\mu,\phi,p)>0\) as the unique solution of
\begin{equation}
\frac{
\Gamma\!\left(1+2\phi_W(\mu,\phi,p)\right)
}{
\Gamma\!\left(1+\phi_W(\mu,\phi,p)\right)^2
}
-1
=
\phi\mu^{p-2}.
\label{eq:weibull_general_phi_equation}
\end{equation}
The conventional Weibull shape and scale parameters are then
\begin{equation}
\kappa(\mu,\phi,p)
=
\frac{1}{\phi_W(\mu,\phi,p)},
\qquad
\beta(\mu,\phi,p)
=
\frac{
\mu
}{
\Gamma\!\left(1+\phi_W(\mu,\phi,p)\right)
}.
\label{eq:weibull_general_parameters}
\end{equation}

This construction ensures that
\[
\mathbb E(\Lambda)=\mu
\]
and
\[
\begin{aligned}
\mathbb V(\Lambda)
&=
\mu^2
\left[
\frac{
\Gamma\!\left(1+2\phi_W(\mu,\phi,p)\right)
}{
\Gamma\!\left(1+\phi_W(\mu,\phi,p)\right)^2
}
-1
\right]
\\
&=
\mu^2
\left(
\phi\mu^{p-2}
\right)
\\
&=
\mu^p\phi.
\end{aligned}
\]

Under the general variance specification
\eqref{eq:weibull_general_variance}, the density is
\begin{equation}
\begin{aligned}
&
f_\Lambda(\lambda;\mu,\phi,p)
\\
& \quad=
\frac{1}{\phi_W(\mu,\phi,p)}
\left[
\frac{
\Gamma\!\left(1+\phi_W(\mu,\phi,p)\right)
}{
\mu
}
\right]^{1/\phi_W(\mu,\phi,p)}
\times
\lambda^{1/\phi_W(\mu,\phi,p)-1}
\\
&\quad \quad \times
\exp\!\left\{
-
\left[
\frac{
\Gamma\!\left(1+\phi_W(\mu,\phi,p)\right)
\lambda
}{
\mu
}
\right]^{1/\phi_W(\mu,\phi,p)}
\right\},
\ 
\lambda>0.
\end{aligned}
\label{eq:weibull_general_density}
\end{equation}

The corresponding cumulative distribution function is
\begin{equation}
\begin{aligned}
&
F_\Lambda(\lambda;\mu,\phi,p)
\\
& \qquad=
1-
\exp\!\left\{
-
\left[
\frac{
\Gamma\!\left(1+\phi_W(\mu,\phi,p)\right)
\lambda
}{
\mu
}
\right]^{1/\phi_W(\mu,\phi,p)}
\right\},
\ 
\lambda>0,    
\end{aligned}
\label{eq:weibull_general_cdf}
\end{equation}
and the quantile function is
\begin{equation}
Q_\Lambda(u;\mu,\phi,p)
=
\frac{
\mu
}{
\Gamma\!\left(1+\phi_W(\mu,\phi,p)\right)
}
\left[-\log(1-u)\right]^{\phi_W(\mu,\phi,p)},
\label{eq:weibull_general_quantile}
\end{equation}
where $0<u<1.$

For the default case \(p=2\),
\eqref{eq:weibull_general_phi_equation} reduces to
\[
\frac{\Gamma(1+2\phi_W)}
{\Gamma(1+\phi_W)^2}
-1
=
\phi,
\]
so that \(\phi_W\) depends only on the dispersion parameter \(\phi\).
For \(p=1\) or \(p=3\), however, \(\phi_W(\mu,\phi,p)\) also depends on
the conditional mean \(\mu\), and hence may vary across observations.

\subsection{Log-logistic mixing distribution}

Throughout this paper, the Log-logistic mixing distribution is parameterized
by its mean \(\mu>0\) and dispersion parameter \(\phi>0\) through the
variance function
\begin{equation}
\mathbb E(\Lambda)=\mu,
\qquad
\mathbb V(\Lambda)=\mu^2\phi.
\label{eq:loglogistic_var_mu2}
\end{equation}

More generally, the proposed implementation allows the variance function to
be specified as
\begin{equation}
\mathbb V(\Lambda)=\mu^{p}\phi,
\qquad
p\in\{1,2,3\},
\label{eq:loglogistic_general_variance}
\end{equation}
which includes \eqref{eq:loglogistic_var_mu2} as the default choice
(\(p=2\)).

Let
\(\Lambda\sim\mathrm{LogLogistic}(\beta,\kappa)\),
where \(\beta>0\) and \(\kappa>2\) denote the scale and shape parameters,
respectively. Under
\eqref{eq:loglogistic_general_variance},
the shape parameter is obtained as the unique solution of

\[
\frac{
\dfrac{2\pi/\kappa}
{\sin(2\pi/\kappa)}
}{
\left(
\dfrac{\pi/\kappa}
{\sin(\pi/\kappa)}
\right)^2
}
-1
=
\phi\mu^{p-2},
\]
and the corresponding scale parameter is

\[
\beta
=
\mu
\frac{\sin(\pi/\kappa)}
{\pi/\kappa},
\]
thereby preserving the prescribed mean and variance.

Under the default parameterization (\(p=2\)), the density function is

\[
f_\Lambda(\lambda;\mu,\phi)
=
\frac{\kappa}{\beta}
\frac{(\lambda/\beta)^{\kappa-1}}
{\left[1+(\lambda/\beta)^\kappa\right]^2},
\qquad
\lambda>0,
\]
the cumulative distribution function is

\[
F_\Lambda(\lambda;\mu,\phi)
=
\frac{1}
{1+(\beta/\lambda)^\kappa},
\qquad
\lambda>0,
\]
and the quantile function is

\[
Q_\Lambda(u;\mu,\phi)
=
\beta
\left(
\frac{u}{1-u}
\right)^{1/\kappa},
\qquad
0<u<1.
\]

\section{Evaluation of Marginal and Checkerboard Cell Probabilities}
\label{app:cellprob}

This appendix summarizes the evaluation of the marginal probabilities and
checkerboard cell probabilities required for likelihood inference.

Throughout this appendix, the notation is presented for a single response
variable. The corresponding expressions are applied componentwise to each
margin of the multivariate model.

Let
\[
N_i\mid\Lambda_i
\sim
\mathrm{Poisson}(\Lambda_i),
\]
where the latent intensity follows a positive mixing distribution
\[
\Lambda_i
\sim
F_\Lambda(\mu_i,\phi),
\qquad
\mu_i=\exp(\mathbf x_i^\top\boldsymbol\beta).
\]

Let
\(f_\Lambda\),
\(F_\Lambda\),
and
\(Q_\Lambda\)
denote the density, cumulative distribution function, and quantile function
of the selected mixing distribution, respectively.

The marginal probability can be expressed equivalently in either of the
following forms:
\begin{equation}
\begin{aligned}
\mathbb P(N_i=n_i\mid\mathbf X_i=\mathbf x_i)
&=
\int_0^\infty
\frac{e^{-\lambda}\lambda^{n_i}}
{\Gamma(n_i+1)}
f_\Lambda(\lambda;\mu_i,\phi)
\,d\lambda
\\
&=
\int_0^1
\frac{
e^{-Q_\Lambda(u)}
Q_\Lambda(u)^{n_i}
}
{\Gamma(n_i+1)}
\,du.
\end{aligned}
\label{eq:marginal_two_forms}
\end{equation}

Similarly, for a checkerboard copula with resolution \(d\), the probability
corresponding to the \(m\)-th checkerboard cell can be written as

\begin{equation}
\begin{aligned}
I_i^{(m)}(n_i)
&=
\int_{L_{i,m}}^{U_{i,m}}
\frac{e^{-\lambda}\lambda^{n_i}}
{\Gamma(n_i+1)}
f_\Lambda(\lambda;\mu_i,\phi)
\,d\lambda
\\
&=
\int_{(m-1)/d}^{m/d}
\frac{
e^{-Q_\Lambda(u)}
Q_\Lambda(u)^{n_i}
}
{\Gamma(n_i+1)}
\,du,
\qquad
m=1,\ldots,d,
\end{aligned}
\label{eq:cell_two_forms}
\end{equation}

where

\[
L_{i,m}
=
Q_\Lambda\!\left(\frac{m-1}{d}\right),
\qquad
U_{i,m}
=
Q_\Lambda\!\left(\frac{m}{d}\right).
\]

The two representations in
\eqref{eq:marginal_two_forms}
and
\eqref{eq:cell_two_forms}
are mathematically equivalent. Depending on the selected mixing distribution,
either the density representation or the quantile representation may be more
computationally efficient. Throughout the package, the implementation
automatically adopts the representation that yields the most stable and
efficient evaluation of the likelihood and its derivatives.

Table~\ref{tab:appendix_computation} summarizes the computational strategy
used for each mixing distribution.

\begin{table}[H]
\centering
\renewcommand{\arraystretch}{1.2}
\caption{Evaluation of marginal and checkerboard cell probabilities under different mixing distributions.}
\label{tab:appendix_computation}
\resizebox{\columnwidth}{!}{%
\begin{tabular}{llll}
\hline
\textbf{Mixing distribution} &
\textbf{Representation} &
\textbf{Marginal probability} &
\textbf{Checkerboard cells} \\
\hline
Gamma &
Density &
Closed form &
Closed form \\

LN &
Quantile &
Gauss--Hermite &
Gauss--Hermite \\

IG &
Density / Quantile &
Closed form &
Gauss--Legendre \\

Weibull &
Quantile &
Gauss--Legendre &
Gauss--Legendre \\

Pareto &
Density &
Gauss--Laguerre &
Gauss--Laguerre \\
\hline
\end{tabular}
}
\end{table}

Unless otherwise stated, all likelihood contributions are evaluated on the
logarithmic scale to improve numerical stability. For quadrature-based
methods, the quadrature nodes and weights are independent of the model
parameters. Consequently, score functions and observed information matrices
can be obtained by differentiating directly under the quadrature summation.

\subsection{Gamma Mixing}
\label{marginal_gamma}

For the Gamma mixing distribution, both the marginal probability and the
checkerboard cell probabilities admit analytical expressions. The marginal
probability is evaluated through the Negative Binomial representation, whereas
the checkerboard cell probabilities are computed using differences of
incomplete Gamma functions. Consequently, no numerical quadrature is required.

\subsubsection{Marginal probability}

When the checkerboard consists of a single cell, \(d=1\), the cell probability
coincides with the marginal probability:
\[
I_i(n_i)
=
\mathbb P(N_i=n_i\mid\mathbf X_i=\mathbf x_i).
\]

Since a Poisson distribution with a Gamma mixing distribution yields a Negative Binomial distribution,
\begin{equation*}
\begin{aligned}
&
\mathbb P(N_i=n_i\mid\mathbf X_i=\mathbf x_i)
\\
& \qquad=
\frac{
\Gamma(n_i+\phi^{-1})
}{
\Gamma(\phi^{-1})\Gamma(n_i+1)
}
\left(
\frac{\phi^{-1}}
{\phi^{-1}+\mu_i}
\right)^{\phi^{-1}}
\left(
\frac{\mu_i}
{\phi^{-1}+\mu_i}
\right)^{n_i}.    
\end{aligned}
\label{eq:gamma_marginal_probability}
\end{equation*}

Equivalently,
\[
\mathbb P(N_i=n_i\mid\mathbf X_i=\mathbf x_i)
=
\frac{
\Gamma(n_i+\phi^{-1})
}{
\Gamma(\phi^{-1})\Gamma(n_i+1)
}
\frac{
(\mu_i\phi)^{n_i}
}{
(1+\mu_i\phi)^{n_i+\phi^{-1}}
}.
\]

In the implementation, the marginal probability is evaluated directly on the
logarithmic scale using the Negative Binomial probability mass function.
Thus, no numerical integration is required when \(d=1\).

\subsubsection{Checkerboard cell probabilities}

Although the cell probability also admits the quantile representation
\[
I_i^{(m)}(n_i)
=
\int_{(m-1)/d}^{m/d}
\frac{
e^{-Q_\Lambda(u)}
Q_\Lambda(u)^{n_i}
}{
\Gamma(n_i+1)
}
\,du,
\]
the density representation is computationally more convenient for Gamma
mixing:
\[
I_i^{(m)}(n_i)
=
\int_{L_{i,m}}^{U_{i,m}}
\frac{e^{-\lambda}\lambda^{n_i}}
{\Gamma(n_i+1)}
f_\Lambda(\lambda;\mu_i,\phi)
\,d\lambda.
\]

Substituting the Gamma density yields
\[
\begin{aligned}
&
I_i^{(m)}(n_i)
\\
& \qquad=
\frac{
(\mu_i\phi)^{-1/\phi}
}{
\Gamma(\phi^{-1})\Gamma(n_i+1)
}
\int_{L_{i,m}}^{U_{i,m}}
\lambda^{n_i+\phi^{-1}-1}
\exp\left[
-\left(
1+\frac{1}{\mu_i\phi}
\right)\lambda
\right]
\,d\lambda.    
\end{aligned}
\]

Let
\[
s_i=n_i+\phi^{-1},
\qquad
A_i=1+\frac{1}{\mu_i\phi}.
\]

After the transformation \(t=A_i\lambda\), the cell probability becomes
\begin{equation*}
\begin{aligned}
I_i^{(m)}
={}&
\frac{
(\mu_i\phi)^{-1/\phi}
}{
\Gamma(\phi^{-1})\Gamma(n_i+1)A_i^{s_i}
}
\\
& \qquad \times
\left[
\gamma\left(s_i,A_iU_{i,m}\right)
-
\gamma\left(s_i,A_iL_{i,m}\right)
\right],
\end{aligned}
\label{eq:gamma_cell_probability}
\end{equation*}
where
\[
\gamma(s,x)
=
\int_0^x t^{s-1}e^{-t}\,dt
\]
denotes the lower incomplete Gamma function.

Equivalently, using the regularized lower incomplete Gamma function
\[
P(s,x)
=
\frac{\gamma(s,x)}{\Gamma(s)},
\]
we obtain
\begin{equation}
\begin{aligned}
I_i^{(m)}(n_i)
={}&
\frac{
\Gamma(s_i)
}{
\Gamma(\phi^{-1})\Gamma(n_i+1)
}
\frac{
(\mu_i\phi)^{-1/\phi}
}{
A_i^{s_i}
}
\\
& \qquad\times
\left[
P\left(s_i,A_iU_{i,m}\right)
-
P\left(s_i,A_iL_{i,m}\right)
\right].
\end{aligned}
\label{eq:gamma_cell_regularized}
\end{equation}

The implementation uses the regularized representation in
\eqref{eq:gamma_cell_regularized}. The multiplicative factor is evaluated on
the logarithmic scale, and the difference between the two Gamma probabilities
is computed using stable logarithmic subtraction. For the final checkerboard
cell, the upper-tail Gamma probability is evaluated directly. These
calculations reduce numerical cancellation and underflow when the interval
probability is small.

\subsubsection{Derivatives of the marginal and cell probabilities}

Efficient likelihood optimization requires the derivatives of both the
marginal probability and the checkerboard cell probabilities with respect to
the model parameters. These derivatives are obtained by differentiating the
analytical expressions derived above.

For \(d=1\), let
\[
p_i
=
\mathbb P(N_i=n_i\mid\mathbf X_i=\mathbf x_i).
\]

The derivative with respect to \(\mu_i\) is
\[
\frac{\partial p_i}{\partial\mu_i}
=
p_i
\left[
\frac{n_i}{\mu_i}
-
\frac{n_i+\phi^{-1}}
{\mu_i+\phi^{-1}}
\right].
\]

To obtain the derivative with respect to \(\phi\), first differentiate with
respect to the Gamma shape parameter
\[
\alpha=\phi^{-1}.
\]
The corresponding derivative is
\[
\frac{\partial p_i}{\partial\alpha}
=
p_i
\left[
\psi(n_i+\alpha)
-
\psi(\alpha)
+
\log\left(
\frac{\alpha}{\alpha+\mu_i}
\right)
+
1
-
\frac{\alpha+n_i}{\alpha+\mu_i}
\right],
\]
where \(\psi(\cdot)\) denotes the digamma function. Since
\[
\frac{\partial\alpha}{\partial\phi}
=
-\frac{1}{\phi^2}
=
-\frac{\alpha}{\phi},
\]
it follows that
\[
\frac{\partial p_i}{\partial\phi}
=
-
\frac{\alpha}{\phi}
\frac{\partial p_i}{\partial\alpha}.
\]

For \(d>1\), the full joint likelihood depends on the checkerboard cell
probabilities \(I_i^{(m)}(n_i)\). Consequently, joint maximum likelihood
estimation requires their derivatives with respect to the model parameters.
These derivatives are unnecessary under IFM estimation, where only the
marginal likelihood is optimized. Because both the Gamma density and the
checkerboard boundaries depend on the model parameters, differentiation must
account for both sources of parameter dependence. 

For a generic parameter
\(\vartheta\in\{\mu_i,\phi\}\), Leibniz's rule gives
\begin{equation}
\begin{aligned}
\frac{\partial I_i^{(m)}(n_i)}{\partial\vartheta}
={}&
h_i(U_{i,m})
f_\Lambda(U_{i,m};\mu_i,\phi)
\frac{\partial U_{i,m}}{\partial\vartheta}
\\
&-
h_i(L_{i,m})
f_\Lambda(L_{i,m};\mu_i,\phi)
\frac{\partial L_{i,m}}{\partial\vartheta}
\\
&+
\int_{L_{i,m}}^{U_{i,m}}
h_i(\lambda)
\frac{\partial f_\Lambda(\lambda;\mu_i,\phi)}
{\partial\vartheta}
\,d\lambda,
\end{aligned}
\label{eq:gamma_leibniz}
\end{equation}
where
\[
h_i(\lambda)
=
\frac{e^{-\lambda}\lambda^{n_i}}
{\Gamma(n_i+1)}.
\]

The boundary derivatives are induced by the Gamma quantile function and follow
from
\[
F_\Lambda(L_{i,m};\mu_i,\phi)
=
\frac{m-1}{d},
\qquad
F_\Lambda(U_{i,m};\mu_i,\phi)
=
\frac{m}{d}.
\]

The interior terms in \eqref{eq:gamma_leibniz} can again be reduced to
differences of incomplete Gamma functions. Therefore, numerical quadrature is
not required for evaluating the score functions.

Most derivative components are evaluated analytically. The only numerical
differentiation required concerns derivatives of the incomplete Gamma
function with respect to its shape argument. These terms are evaluated using
a central finite-difference approximation. Thus, the resulting score is
semi-analytical rather than being obtained by finite-differencing the complete
likelihood contribution.

Finally, since
\[
\mu_i
=
\exp(\mathbf x_i^\top\boldsymbol\beta),
\qquad
\frac{\partial\mu_i}
{\partial\boldsymbol\beta}
=
\mu_i\mathbf x_i,
\]
the score with respect to the regression coefficients is
\[
\frac{\partial I_i^{(m)}(n_i)}
{\partial\boldsymbol\beta}
=
\frac{\partial I_i^{(m)}(n_i)}
{\partial\mu_i}
\mu_i\mathbf x_i.
\]

The same chain-rule expression applies to the marginal probability when
\(d=1\). Optimization is performed on the logarithmic dispersion scale,
\(\eta_\phi=\log\phi\), and hence
\[
\frac{\partial I_i^{(m)}}{\partial\eta_\phi}
=
\phi
\frac{\partial I_i^{(m)}}{\partial\phi}.
\]

The implementation also permits
\[
\mathbb V(\Lambda_i\mid\mathbf X_i)
=
\phi\mu_i^p,
\qquad
p\in\{1,3\}.
\]
In this case, the Gamma shape and rate parameters are replaced by
\[
\alpha_i
=
\frac{\mu_i^{2-p}}{\phi},
\qquad
r_i
=
\frac{\mu_i^{1-p}}{\phi}.
\]
The same analytical probability formulas remain applicable after substituting
these parameters. The score functions additionally account for the dependence
of \(\alpha_i\) on \(\mu_i\), through
\[
\frac{\partial\alpha_i}{\partial\mu_i}
=
(2-p)\frac{\alpha_i}{\mu_i}.
\]

\subsection{Lognormal Mixing}
\label{marginal_lognormal}

Unlike the Gamma mixing distribution, neither the marginal probability nor the
checkerboard cell probabilities admit closed-form expressions under Lognormal
mixing. Both quantities are therefore evaluated numerically using
Gauss--Hermite quadrature after transforming the integration interval to the
real line.

The latent intensity is parameterized as
\[
\Lambda_i\mid\mathbf X_i=\mathbf x_i
\sim
\operatorname{Lognormal}\left(
\log\mu_i-\frac{\sigma^2}{2},
\sigma^2
\right),
\ 
\mu_i=\exp(\mathbf x_i^\top\boldsymbol\beta),
\]
where
\[
\sigma^2=\log(1+\phi).
\]

This parameterization ensures that
\[
\mathbb E(\Lambda_i\mid\mathbf X_i=\mathbf x_i)
=
\mu_i,
\qquad
\mathbb V(\Lambda_i\mid\mathbf X_i=\mathbf x_i)
=
\mu_i^2\phi.
\]

The corresponding quantile function is
\[
Q_\Lambda(u;\mu_i,\phi)
=
\exp\left[
\log\mu_i
-\frac{\sigma^2}{2}
+\sigma\Phi^{-1}(u)
\right],
\quad 0<u<1,
\]
where
\(
\sigma=\sqrt{\log(1+\phi)}.
\)

Define
\[
g_i(u;\mu_i,\phi)
=
\frac{
\exp\!\left[-Q_\Lambda(u;\mu_i,\phi)\right]
Q_\Lambda(u;\mu_i,\phi)^{n_i}
}{
\Gamma(n_i+1)
}.
\]

\subsubsection{Marginal probability}

When the checkerboard consists of a single cell, \(d=1\), the cell probability
coincides with the marginal probability:
\[
I_i(n_i)
=
\mathbb P(N_i=n_i\mid\mathbf X_i=\mathbf x_i).
\]

Using the quantile representation,
\[
\mathbb P(N_i=n_i\mid\mathbf X_i=\mathbf x_i)
=
\int_0^1
g_i(u;\mu_i,\phi)\,du.
\]

Applying the transformation
\[
u=\Phi(\sqrt{2}x),
\qquad x\in\mathbb R,
\]
gives
\[
\mathbb P(N_i=n_i\mid\mathbf X_i=\mathbf x_i)
=
\frac{1}{\sqrt{\pi}}
\int_{-\infty}^{\infty}
g_i\left(
\Phi(\sqrt{2}x);
\mu_i,\phi
\right)
e^{-x^2}\,dx.
\]

Let \(\{x_r,w_r\}_{r=1}^R\) denote the nodes and weights of an
\(R\)-point Gauss--Hermite quadrature rule. The marginal probability is
approximated by
\begin{equation}
\mathbb P(N_i=n_i\mid\mathbf X_i=\mathbf x_i)
\approx
\frac{1}{\sqrt{\pi}}
\sum_{r=1}^R
w_r
g_i\left(
\Phi(\sqrt{2}x_r);
\mu_i,\phi
\right).
\label{eq:lnorm_marginal_gh}
\end{equation}

Equivalently, because
\[
\Phi^{-1}\left\{\Phi(\sqrt{2}x_r)\right\}
=
\sqrt{2}x_r,
\]
the Lognormal quantile evaluated at the \(r\)th quadrature node is
\[
Q_{ir}
=
\exp\left[
\log\mu_i
-\frac{\sigma^2}{2}
+\sigma\sqrt{2}x_r
\right].
\]

In the implementation, the quadrature terms are evaluated on the logarithmic
scale, and their weighted sum is computed using the log-sum-exp identity.
This reduces numerical underflow when the marginal probabilities are very
small.

\subsubsection{Checkerboard cell probabilities}

For the \(m\)th checkerboard cell, the corresponding probability is
\[
I_i^{(m)}(n_i)
=
\int_{(m-1)/d}^{m/d}
g_i(u;\mu_i,\phi)\,du.
\]

Introducing the transformation
\[
u
=
\frac{m-1}{d}
+
\frac{1}{d}\Phi(\sqrt{2}x),
\qquad
x\in\mathbb R,
\]
yields
\[
I_i^{(m)}(n_i)
=
\frac{1}{d\sqrt{\pi}}
\int_{-\infty}^{\infty}
g_i\left(
\frac{m-1}{d}
+
\frac{1}{d}\Phi(\sqrt{2}x);
\mu_i,\phi
\right)
e^{-x^2}\,dx.
\]

Hence, the Gauss--Hermite approximation is
\begin{equation}
I_i^{(m)}(n_i)
\approx
\frac{1}{d\sqrt{\pi}}
\sum_{r=1}^R
w_r
g_i\left(
u_{mr};
\mu_i,\phi
\right),
\label{eq:lnorm_cell_gh}
\end{equation}
where
\[
u_{mr}
=
\frac{m-1}{d}
+
\frac{1}{d}\Phi(\sqrt{2}x_r).
\]

Let
\[
q_{mr}
=
\Phi^{-1}(u_{mr}).
\]
The corresponding Lognormal quantile is
\[
Q_{imr}
=
\exp\left[
\log\mu_i
-\frac{\sigma^2}{2}
+\sigma q_{mr}
\right].
\]

Therefore, \eqref{eq:lnorm_cell_gh} can be written explicitly as
\[
I_i^{(m)}(n_i)
\approx
\frac{1}{d\sqrt{\pi}}
\sum_{r=1}^R
w_r
\frac{
\exp(-Q_{imr})Q_{imr}^{n_i}
}{
\Gamma(n_i+1)
}.
\]

In the implementation, the transformed probabilities \(u_{mr}\) and their
standard Normal quantiles \(q_{mr}\) are precomputed for all checkerboard
cells and quadrature nodes. The values of \(u_{mr}\) are restricted to a
small interior interval \([\varepsilon,1-\varepsilon]\) to avoid infinite
Normal quantiles at the boundaries. The quadrature sums are then evaluated
on the logarithmic scale using stable log-sum-exp calculations.

\subsubsection{Derivatives of the marginal and cell probabilities}

Efficient likelihood optimization requires derivatives of both the marginal
probability and the checkerboard cell probabilities with respect to the model
parameters. Since the quadrature nodes, weights, and transformed probabilities
\(u_{mr}\) do not depend on \(\mu_i\) or \(\phi\), differentiation can be
performed directly under the quadrature summation.

For a fixed \(u\), write
\[
Q_i(u)
=
Q_\Lambda(u;\mu_i,\phi).
\]
Since
\[
g_i(u;\mu_i,\phi)
=
\frac{
\exp[-Q_i(u)]Q_i(u)^{n_i}
}{
\Gamma(n_i+1)
},
\]
its derivative with respect to a generic parameter
\(\vartheta\in\{\mu_i,\phi\}\) is
\[
\frac{\partial g_i(u;\mu_i,\phi)}
{\partial\vartheta}
=
g_i(u;\mu_i,\phi)
\left[n_i-Q_i(u)\right]
\frac{\partial\log Q_i(u)}
{\partial\vartheta}.
\]

For the Lognormal quantile,
\[
\log Q_i(u)
=
\log\mu_i
-\frac{\sigma^2}{2}
+\sigma\Phi^{-1}(u).
\]

Because
\[
\sigma^2=\log(1+\phi)
\]
does not depend on \(\mu_i\),
\[
\frac{\partial\log Q_i(u)}
{\partial\mu_i}
=
\frac{1}{\mu_i}.
\]

It follows that
\[
\frac{\partial g_i(u;\mu_i,\phi)}
{\partial\mu_i}
=
g_i(u;\mu_i,\phi)
\frac{n_i-Q_i(u)}{\mu_i}.
\]

For the dispersion parameter,
\[
\frac{\partial\sigma^2}{\partial\phi}
=
\frac{1}{1+\phi},
\qquad
\frac{\partial\sigma}{\partial\phi}
=
\frac{1}{2\sigma(1+\phi)}.
\]

Therefore,
\[
\frac{\partial\log Q_i(u)}
{\partial\phi}
=
-\frac{1}{2(1+\phi)}
+
\frac{\Phi^{-1}(u)}
{2\sigma(1+\phi)},
\]
or equivalently,
\[
\frac{\partial\log Q_i(u)}
{\partial\phi}
=
\frac{
\Phi^{-1}(u)/\sigma-1
}{
2(1+\phi)
}.
\]

Thus,
\[
\frac{\partial g_i(u;\mu_i,\phi)}
{\partial\phi}
=
g_i(u;\mu_i,\phi)
\left[n_i-Q_i(u)\right]
\frac{
\Phi^{-1}(u)/\sigma-1
}{
2(1+\phi)
}.
\]

For the \(m\)th checkerboard cell, differentiation of
\eqref{eq:lnorm_cell_gh} gives
\[
\frac{\partial I_i^{(m)}(n_i)}
{\partial\mu_i}
\approx
\frac{1}{d\sqrt{\pi}}
\sum_{r=1}^R
w_r
g_i(u_{mr};\mu_i,\phi)
\frac{n_i-Q_{imr}}{\mu_i},
\]
and
\[
\frac{\partial I_i^{(m)}(n_i)}
{\partial\phi}
\approx
\frac{1}{d\sqrt{\pi}}
\sum_{r=1}^R
w_r
g_i(u_{mr};\mu_i,\phi)
\left(n_i-Q_{imr}\right)
\frac{
q_{mr}/\sigma-1
}{
2(1+\phi)
}.
\]

The same expressions apply to the marginal probability when \(d=1\), in
which case
\[
u_{1r}
=
\Phi(\sqrt{2}x_r).
\]

For \(d>1\), the full joint likelihood depends on the checkerboard cell
probabilities \(I_i^{(m)}(n_i)\). Consequently, joint maximum likelihood
estimation requires their derivatives with respect to the marginal parameters.
These cell-probability derivatives are unnecessary under IFM estimation,
where the marginal parameters are estimated solely from the marginal
likelihood.

Finally, since
\[
\mu_i
=
\exp(\mathbf x_i^\top\boldsymbol\beta),
\qquad
\frac{\partial\mu_i}
{\partial\boldsymbol\beta}
=
\mu_i\mathbf x_i,
\]
the derivative with respect to the regression coefficients is
\[
\frac{\partial I_i^{(m)}(n_i)}
{\partial\boldsymbol\beta}
=
\frac{\partial I_i^{(m)}(n_i)}
{\partial\mu_i}
\mu_i\mathbf x_i.
\]

The same chain-rule expression applies to the marginal probability when
\(d=1\). Optimization is performed on the logarithmic dispersion scale,
\(\eta_\phi=\log\phi\), and hence
\[
\frac{\partial I_i^{(m)}(n_i)}
{\partial\eta_\phi}
=
\phi
\frac{\partial I_i^{(m)}(n_i)}
{\partial\phi}.
\]

The implementation also permits
\[
\mathbb V(\Lambda_i\mid\mathbf X_i)
=
\phi\mu_i^p,
\qquad
p\in\{1,3\}.
\]
In this case,
\[
\sigma_i^2
=
\log\left(1+\phi\mu_i^{p-2}\right),
\]
so the derivatives additionally account for the dependence of
\(\sigma_i^2\) on \(\mu_i\). The same Gauss--Hermite quadrature representation
remains applicable.

\subsection{Inverse Gaussian Mixing}
\label{marginal_ig}

For the Inverse Gaussian mixing distribution, the marginal probability admits
a closed-form expression involving the modified Bessel function of the second
kind. In contrast, the checkerboard cell probabilities do not generally admit
closed-form expressions and are evaluated numerically using Gauss--Legendre
quadrature over finite intervals.

\subsubsection{Marginal probability}

When the checkerboard consists of a single cell, \(d=1\), the cell probability
coincides with the marginal probability:
\[
I_i(n_i)
=
\mathbb P
(
N_i=n_i
\mid
\mathbf X_i=\mathbf x_i
).
\]

Using the density representation,
\[
\mathbb P
(
N_i=n_i
\mid
\mathbf X_i=\mathbf x_i
)
=
\int_0^\infty
\frac{
e^{-\lambda}
\lambda^{n_i}
}
{
\Gamma(n_i+1)
}
f_\Lambda(\lambda;\mu_i,\phi)
\,d\lambda.
\]

Substituting the Inverse Gaussian density yields
\[
\begin{aligned}
&
\mathbb P
(
N_i=n_i
\mid
\mathbf X_i=\mathbf x_i
)
=
\frac{
\exp(\phi^{-1}/\mu_i)
}
{
\Gamma(n_i+1)
}
\left(
\frac{\phi^{-1}}
{2\pi}
\right)^{1/2}
\\
& \qquad\qquad \qquad
\times
\int_0^\infty
\lambda^{n_i-\frac32}
\exp
\left[
-
\left(
1+\frac{\phi^{-1}}{2\mu_i^2}
\right)\lambda
-
\frac{\phi^{-1}}
{2\lambda}
\right]
d\lambda.
\end{aligned}
\]

Let
\[
\nu_i
=
n_i-\frac12,
\qquad
A_i
=
1+\frac{\phi^{-1}}{2\mu_i^2},
\qquad
B
=
\frac{\phi^{-1}}2.
\]

Using the identity
\[
\int_0^\infty
x^{\nu-1}
\exp
\left(
-Ax-\frac{B}{x}
\right)
dx
=
2
\left(
\frac BA
\right)^{\nu/2}
K_\nu
\!\left(
2\sqrt{AB}
\right),
\]
where \(K_\nu(\cdot)\) denotes the modified Bessel function of the second kind,
gives
\begin{equation}
\begin{aligned}
&
\mathbb P
(
N_i=n_i
\mid
\mathbf X_i=\mathbf x_i
)
\\
& \qquad =
\frac{2}{\Gamma(n_i+1)}
\left(
\frac{\phi^{-1}}{2\pi}
\right)^{1/2}
\exp
\left(
\frac{\phi^{-1}}{\mu_i}
\right)
\times
\left(
\frac BA_i
\right)^{\nu_i/2}
K_{\nu_i}
\!\left(
2\sqrt{A_iB}
\right).
\end{aligned}
\label{eq:ig_marginal_probability}
\end{equation}

Since \(n_i\) is a nonnegative integer,
\[
\nu_i
=
n_i-\frac12
\]
is always a half-integer. Therefore,
\(K_{\nu_i}\) admits the finite expansion
\[
K_{m+\frac12}(x)
=
\sqrt{\frac{\pi}{2x}}
e^{-x}
\sum_{k=0}^{m}
\frac{
(m+k)!
}
{
k!(m-k)!(2x)^k
},
\qquad
m=0,1,\ldots.
\]

The implementation evaluates this finite expansion directly rather than
calling a general-purpose Bessel routine. All multiplicative terms are
computed on the logarithmic scale, and the finite series is evaluated using
the log-sum-exp identity to reduce numerical overflow and underflow. This
implementation is substantially faster and more stable than evaluating the
general modified Bessel function numerically.

\subsubsection{Checkerboard cell probabilities}

For \(d>1\), the \(m\)th checkerboard cell probability is
\[
I_i^{(m)}(n_i)
=
\int_{L_{im}}^{U_{im}}
\frac{
e^{-\lambda}
\lambda^{n_i}
}
{
\Gamma(n_i+1)
}
f_\Lambda(\lambda;\mu_i,\phi)
\,d\lambda,
\]
where
\[
L_{im}
=
Q_\Lambda
\!\left(
\frac{m-1}{d};
\mu_i,\phi
\right),
\qquad
U_{im}
=
Q_\Lambda
\!\left(
\frac{m}{d};
\mu_i,\phi
\right),
\]
and \(Q_\Lambda(\cdot)\) denotes the Inverse Gaussian quantile function.

Unlike the Gamma and Lognormal cases, no analytical expression is available
for either the quantile function or the checkerboard cell probabilities.
Consequently, both quantities are evaluated numerically.

For the first \(d-1\) checkerboard cells, the integration limits are finite.
Applying the linear transformation
\[
\lambda
=
\frac{U_{im}+L_{im}}{2}
+
\frac{U_{im}-L_{im}}{2}x,
\qquad
-1\le x\le1,
\]
gives
\[
\begin{aligned}
I_i^{(m)}(n_i)
=
\frac{U_{im}-L_{im}}{2}
\int_{-1}^{1}
&
\frac{
e^{-\lambda(x)}
\lambda(x)^{n_i}
}
{
\Gamma(n_i+1)
}
\times
f_\Lambda
(
\lambda(x);
\mu_i,\phi
)
dx.
\end{aligned}
\]

Using an \(R\)-point Gauss--Legendre quadrature rule,
\[
\{x_r,w_r\}_{r=1}^{R},
\]
yields
\begin{equation}
\begin{aligned}
I_i^{(m)}(n_i)
\approx
\frac{U_{im}-L_{im}}{2}
\sum_{r=1}^{R}
w_r
&
\frac{
e^{-\lambda_{imr}}
\lambda_{imr}^{n_i}
}
{
\Gamma(n_i+1)
}
\times
f_\Lambda
(
\lambda_{imr};
\mu_i,\phi
),
\end{aligned}
\label{eq:ig_cell_probability}
\end{equation}
where
\[
\lambda_{imr}
=
\frac{U_{im}+L_{im}}{2}
+
\frac{U_{im}-L_{im}}{2}x_r.
\]

Because the final checkerboard interval is unbounded,
\[
(L_{id},\infty),
\]
direct numerical integration over the infinite interval is avoided.
Instead, the last cell probability is computed as
\begin{equation}
I_i^{(d)}(n_i)
=
I_i(n_i)
-
\sum_{m=1}^{d-1}
I_i^{(m)}(n_i),
\label{eq:ig_last_cell}
\end{equation}
where \(I_i(n_i)\) denotes the marginal probability given in
Section~\ref{marginal_ig}.

This construction guarantees that the checkerboard cell probabilities sum
exactly to the marginal probability while avoiding numerical integration over
an unbounded interval.

In the implementation, the Gauss--Legendre quadrature nodes and weights are
precomputed once and reused for all observations. Numerical inversion of the
Inverse Gaussian distribution is performed only at the checkerboard
boundaries. To improve numerical stability, probabilities smaller than a
prescribed tolerance are truncated to a small positive constant before
logarithms are taken.

\subsubsection{Derivatives of the marginal and cell probabilities}

Efficient likelihood optimization requires derivatives of both the marginal
probability and the checkerboard cell probabilities with respect to the model
parameters.

\paragraph{Marginal probability}

For \(d=1\), the marginal probability is evaluated using the closed-form
representation in \eqref{eq:ig_marginal_probability}. Since
\(n_i-\tfrac12\) is a half-integer, both the modified Bessel function and its
derivative admit finite-sum representations. Consequently, the derivatives of
the marginal probability with respect to \(\mu_i\) and \(\phi\) are evaluated
analytically.

Let
\[
\nu_i=n_i-\frac12,
\qquad
A_i
=
1+\frac{\phi^{-1}}{2\mu_i^2},
\qquad
B=\frac{\phi^{-1}}{2},
\]
and
\[
x_i=2\sqrt{A_iB}.
\]

The logarithm of the marginal probability can be written as
\[
\begin{aligned}
\log I_i(n_i)
&=
-\frac12
\left[
\log\phi+\log(2\pi)
\right]
-\log\Gamma(n_i+1)
+\frac{1}{\phi\mu_i}
+\log 2
\\
& \qquad+
\frac{\nu_i}{2}
\left(
\log B-\log A_i
\right)
+
\log K_{\nu_i}(x_i).
\end{aligned}
\]

For a generic parameter
\(\vartheta\in\{\mu_i,\phi\}\),
\[
\begin{aligned}
\frac{\partial\log I_i(n_i)}
{\partial\vartheta}
&=
\frac12
\frac{\partial\log(\phi^{-1})}
{\partial\vartheta}
+
\frac{\partial(\phi^{-1}/\mu_i)}
{\partial\vartheta}
\\
& \ +
\frac{\nu_i}{2}
\left[
\frac{\partial\log B}{\partial\vartheta}
-
\frac{\partial\log A_i}{\partial\vartheta}
\right]
+
\frac{\partial\log K_{\nu_i}(x_i)}
{\partial x_i}
\frac{\partial x_i}{\partial\vartheta}.
\end{aligned}
\]

Because \(\nu_i\) is a half-integer, write
\[
K_{\nu_i}(x)
=
\sqrt{\frac{\pi}{2x}}
e^{-x}
P_i(x),
\]
where
\[
P_i(x)
=
\sum_{k=0}^{m_i}
c_{ik}(2x)^{-k},
\qquad
m_i=
\begin{cases}
0, & n_i=0,\\
n_i-1, & n_i\geq 1,
\end{cases}
\]
and
\[
c_{ik}
=
\frac{(m_i+k)!}
{k!(m_i-k)!}.
\]

It follows that
\[
\frac{\partial\log K_{\nu_i}(x)}
{\partial x}
=
-1-\frac{1}{2x}
+
\frac{P_i'(x)}{P_i(x)}.
\]

Since
\[
P_i'(x)
=
-\frac{1}{x}
\sum_{k=0}^{m_i}
k\,c_{ik}(2x)^{-k},
\]
we obtain
\[
\frac{\partial\log K_{\nu_i}(x)}
{\partial x}
=
-1
-
\frac{1}{2x}
-
\frac{1}{x}
\frac{
\sum_{k=0}^{m_i}
k\,c_{ik}(2x)^{-k}
}{
\sum_{k=0}^{m_i}
c_{ik}(2x)^{-k}
}.
\]

For the natural Inverse Gaussian parameterization,
\[
\lambda_{\mathrm{IG}}=\phi^{-1},
\]
the required derivatives are
\[
\frac{\partial A_i}{\partial\mu_i}
=
-\frac{1}{\phi\mu_i^3},
\qquad
\frac{\partial A_i}{\partial\phi}
=
-\frac{1}{2\phi^2\mu_i^2},
\]
and
\[
\frac{\partial B}{\partial\mu_i}=0,
\qquad
\frac{\partial B}{\partial\phi}
=
-\frac{1}{2\phi^2}.
\]

Moreover,
\[
\frac{\partial x_i}{\partial\vartheta}
=
\frac{x_i}{2}
\left[
\frac{\partial\log A_i}{\partial\vartheta}
+
\frac{\partial\log B}{\partial\vartheta}
\right].
\]

Therefore, the derivatives of the marginal probability are
\[
\frac{\partial I_i(n_i)}
{\partial\mu_i}
=
I_i(n_i)
\frac{\partial\log I_i(n_i)}
{\partial\mu_i},
\]
and
\[
\frac{\partial I_i(n_i)}
{\partial\phi}
=
I_i(n_i)
\frac{\partial\log I_i(n_i)}
{\partial\phi}.
\]

In the implementation, the finite sums defining \(P_i(x)\) and its
logarithmic derivative are evaluated simultaneously using normalized
log-sum-exp weights. Hence, the analytical gradient requires essentially the
same finite summation as the marginal probability itself and introduces only
a negligible additional computational cost.

\paragraph{Checkerboard cell probabilities}

For \(d>1\), the full joint likelihood depends on the checkerboard cell
probabilities
\[
I_i^{(m)}(n_i)
=
\int_{L_{im}}^{U_{im}}
h_i(\lambda;\mu_i,\phi)\,d\lambda,
\]
where
\[
h_i(\lambda;\mu_i,\phi)
=
\frac{
e^{-\lambda}\lambda^{n_i}
}{
\Gamma(n_i+1)
}
f_\Lambda(\lambda;\mu_i,\phi),
\]
and
\[
L_{im}
=
Q_\Lambda
\left(
\frac{m-1}{d};\mu_i,\phi
\right),
\qquad
U_{im}
=
Q_\Lambda
\left(
\frac{m}{d};\mu_i,\phi
\right).
\]

Let \(\vartheta\in\{\mu_i,\phi\}.\)

Since both the integrand and the integration boundaries depend on
\(\vartheta\), Leibniz's rule gives
\begin{equation}
\begin{aligned}
\frac{\partial I_i^{(m)}(n_i)}
{\partial\vartheta}
={}&
\int_{L_{im}}^{U_{im}}
h_i(\lambda;\mu_i,\phi)
\frac{
\partial\log f_\Lambda(\lambda;\mu_i,\phi)
}{
\partial\vartheta
}
\,d\lambda
\\
&+
h_i(U_{im};\mu_i,\phi)
\frac{\partial U_{im}}{\partial\vartheta}
-
h_i(L_{im};\mu_i,\phi)
\frac{\partial L_{im}}{\partial\vartheta}.
\end{aligned}
\label{eq:ig_cell_derivative}
\end{equation}

The Poisson kernel
\[
\frac{
e^{-\lambda}\lambda^{n_i}
}{
\Gamma(n_i+1)
}
\]
does not depend on the marginal parameters. Therefore, the derivative of the
integrand is determined entirely by the derivative of the Inverse Gaussian
density.

For the natural Inverse Gaussian parameterization
\(
\lambda_{\mathrm{IG}}=\phi^{-1},
\) the log-density is
\[
\log f_\Lambda(\lambda;\mu_i,\phi)
=
-\frac12\log(2\pi\phi)
-\frac32\log\lambda
-
\frac{
(\lambda-\mu_i)^2
}{
2\phi\mu_i^2\lambda
}.
\]

Its derivative with respect to \(\mu_i\) is
\begin{equation}
\frac{
\partial\log f_\Lambda(\lambda;\mu_i,\phi)
}{
\partial\mu_i
}
=
\frac{\lambda-\mu_i}
{\phi\mu_i^3}.
\label{eq:ig_logdensity_derivative_mu}
\end{equation}

Similarly,
\begin{equation}
\frac{
\partial\log f_\Lambda(\lambda;\mu_i,\phi)
}{
\partial\phi
}
=
-\frac{1}{2\phi}
+
\frac{
(\lambda-\mu_i)^2
}{
2\phi^2\mu_i^2\lambda
}.
\label{eq:ig_logdensity_derivative_phi}
\end{equation}

Substituting
\eqref{eq:ig_logdensity_derivative_mu}
and
\eqref{eq:ig_logdensity_derivative_phi}
into
\eqref{eq:ig_cell_derivative}
provides the analytical derivatives of the integrand contribution.

The remaining terms involve derivatives of the Inverse Gaussian quantile
function. Let
\[
q_i(p)
=
Q_\Lambda(p;\mu_i,\phi),
\]
which satisfies
\[
F_\Lambda(q_i(p);\mu_i,\phi)=p.
\]
Implicit differentiation gives
\begin{equation}
\frac{\partial q_i(p)}
{\partial\vartheta}
=
-
\frac{
\displaystyle
\frac{\partial
F_\Lambda(q_i(p);\mu_i,\phi)}
{\partial\vartheta}
}{
f_\Lambda(q_i(p);\mu_i,\phi)
},
\qquad
\vartheta\in\{\mu_i,\phi\}.
\label{eq:ig_quantile_derivative}
\end{equation}

To obtain explicit expressions, define
\[
s_i(x)
=
\sqrt{\frac{\phi^{-1}}{x}},
\]
\[
a_{i1}(x)
=
s_i(x)
\left(
\frac{x}{\mu_i}-1
\right),
\]
and
\[
a_{i2}(x)
=
-
s_i(x)
\left(
\frac{x}{\mu_i}+1
\right).
\]

The Inverse Gaussian distribution function can then be written as
\[
F_\Lambda(x;\mu_i,\phi)
=
\Phi(a_{i1}(x))
+
\exp
\left(
\frac{2}{\phi\mu_i}
\right)
\Phi(a_{i2}(x)).
\]

For fixed \(x\), the derivatives of the auxiliary quantities with respect to
\(\mu_i\) are
\[
\frac{\partial s_i(x)}{\partial\mu_i}=0,
\quad
\frac{\partial a_{i1}(x)}{\partial\mu_i}
=
-
s_i(x)
\frac{x}{\mu_i^2},
\quad
\frac{\partial a_{i2}(x)}{\partial\mu_i}
=
s_i(x)
\frac{x}{\mu_i^2}.
\]

Moreover,
\[
\frac{\partial}{\partial\mu_i}
\left(
\frac{2}{\phi\mu_i}
\right)
=
-
\frac{2}{\phi\mu_i^2}.
\]

Therefore,
\begin{equation}
\begin{aligned}
\frac{\partial
F_\Lambda(x;\mu_i,\phi)}
{\partial\mu_i}
={}&
\varphi(a_{i1}(x))
\left[
-
s_i(x)\frac{x}{\mu_i^2}
\right]
\\
&+
\exp
\left(
\frac{2}{\phi\mu_i}
\right)
\Bigg\{
-
\frac{2}{\phi\mu_i^2}
\Phi(a_{i2}(x))
\\
&\hspace{30mm}
+
\varphi(a_{i2}(x))
s_i(x)\frac{x}{\mu_i^2}
\Bigg\},
\end{aligned}
\label{eq:ig_cdf_derivative_mu}
\end{equation}
where \(\varphi(\cdot)\) denotes the standard normal density.

For the derivative with respect to \(\phi\),
\[
\frac{\partial s_i(x)}{\partial\phi}
=
-\frac{s_i(x)}{2\phi},
\]
so that
\[
\frac{\partial a_{i1}(x)}{\partial\phi}
=
-
\frac{a_{i1}(x)}{2\phi},
\quad
\frac{\partial a_{i2}(x)}{\partial\phi}
=
-
\frac{a_{i2}(x)}{2\phi}.
\]

In addition,
\[
\frac{\partial}{\partial\phi}
\left(
\frac{2}{\phi\mu_i}
\right)
=
-
\frac{2}{\phi^2\mu_i}.
\]

Hence,
\begin{equation}
\begin{aligned}
\frac{\partial
F_\Lambda(x;\mu_i,\phi)}
{\partial\phi}
={}&
-
\frac{
a_{i1}(x)\varphi(a_{i1}(x))
}{
2\phi
}
\\
&+
\exp
\left(
\frac{2}{\phi\mu_i}
\right)
\Bigg\{
-
\frac{2}{\phi^2\mu_i}
\Phi(a_{i2}(x))
\\
&\hspace{30mm}
-
\frac{
a_{i2}(x)\varphi(a_{i2}(x))
}{
2\phi
}
\Bigg\}.
\end{aligned}
\label{eq:ig_cdf_derivative_phi}
\end{equation}

Combining
\eqref{eq:ig_quantile_derivative},
\eqref{eq:ig_cdf_derivative_mu},
and
\eqref{eq:ig_cdf_derivative_phi}
gives
\[
\frac{\partial L_{im}}{\partial\vartheta}
=
\left.
\frac{\partial q_i(p)}{\partial\vartheta}
\right|_{p=(m-1)/d},
\]
and
\[
\frac{\partial U_{im}}{\partial\vartheta}
=
\left.
\frac{\partial q_i(p)}{\partial\vartheta}
\right|_{p=m/d}.
\]

For the first cell,
\[
L_{i1}=0,
\]
and therefore
\[
\frac{\partial L_{i1}}{\partial\mu_i}
=
\frac{\partial L_{i1}}{\partial\phi}
=
0.
\]

The integral term in
\eqref{eq:ig_cell_derivative}
is evaluated using the same Gauss--Legendre quadrature rule as the cell
probability itself. Specifically, for \(m=1,\ldots,d-1\),
\[
\begin{aligned}
\frac{\partial I_i^{(m)}(n_i)}
{\partial\vartheta}
\approx{}&
\frac{U_{im}-L_{im}}{2}
\sum_{r=1}^{R}
w_r
h_i(\lambda_{imr};\mu_i,\phi)
\\
&\times
\frac{
\partial\log
f_\Lambda(\lambda_{imr};\mu_i,\phi)
}{
\partial\vartheta
}
\\
&+
h_i(U_{im};\mu_i,\phi)
\frac{\partial U_{im}}{\partial\vartheta}
-
h_i(L_{im};\mu_i,\phi)
\frac{\partial L_{im}}{\partial\vartheta},
\end{aligned}
\]
where
\[
\lambda_{imr}
=
\frac{U_{im}+L_{im}}{2}
+
\frac{U_{im}-L_{im}}{2}x_r.
\]

The final checkerboard cell is evaluated by subtraction:
\[
I_i^{(d)}(n_i)
=
I_i(n_i)
-
\sum_{m=1}^{d-1}
I_i^{(m)}(n_i).
\]
Consequently,
\begin{equation}
\frac{
\partial I_i^{(d)}(n_i)
}{
\partial\vartheta
}
=
\frac{
\partial I_i(n_i)
}{
\partial\vartheta
}
-
\sum_{m=1}^{d-1}
\frac{
\partial I_i^{(m)}(n_i)
}{
\partial\vartheta
},
\qquad
\vartheta\in\{\mu_i,\phi\}.
\label{eq:ig_last_cell_derivative}
\end{equation}

Therefore, both the checkerboard cell probabilities and their derivatives are
evaluated using the same Gauss--Legendre quadrature rule, resulting in nearly
identical computational complexity.

\subsection{Pareto Mixing}
\label{marginal_pareto}

For the Pareto mixing distribution, we consider only the default
variance parameterization
\[
\mathbb V(\Lambda_i\mid\mathbf X_i)
=
(1+2\phi)\mu_i^2,
\]

Unlike the Gamma and Inverse Gaussian mixing distributions, neither the resulting marginal probability mass function nor the checkerboard cell probabilities admit convenient closed-form expressions.

\subsubsection{Marginal probability}

The marginal probability corresponds to the special case
\(d=1\), namely
\[
I_i(n_i)
=
\int_0^1
\frac{
e^{-Q_\Lambda(u)}
Q_\Lambda(u)^{n_i}
}
{\Gamma(n_i+1)}
\,du.
\]

Using the transformation
\[
u
=
1-e^{-t},
\qquad
0\le t<\infty,
\]
gives
\[
I_i(n_i)
=
\int_0^\infty
\frac{
e^{-Q_\Lambda(t)}
Q_\Lambda(t)^{n_i}
}
{\Gamma(n_i+1)}
e^{-t}\,dt,
\]
where
\[
Q_\Lambda(t)
=
\beta_i
\left(
e^{t/\alpha_i}-1
\right).
\]

The transformed integral is evaluated numerically using
Gauss--Laguerre quadrature.
\subsubsection{Checkerboard cell probabilities}

For the \(m\)th checkerboard cell,
\[
I_i^{(m)}(n_i)
=
\int_{(m-1)/d}^{m/d}
\frac{
e^{-Q_\Lambda(u)}
Q_\Lambda(u)^{n_i}
}
{\Gamma(n_i+1)}
\,du.
\]

Applying the transformation
\[
1-u
=
\left(1-\frac{m}{d}\right)
+
\frac1d e^{-t},
\qquad
0\le t<\infty,
\]
yields
\[
du
=
\frac1d
e^{-t}\,dt,
\]
and therefore
\[
I_i^{(m)}(n_i)
=
\frac1d
\int_0^\infty
\frac{
e^{-Q_{\Lambda,m}(t)}
Q_{\Lambda,m}(t)^{n_i}
}
{\Gamma(n_i+1)}
e^{-t}\,dt,
\]
where
\[
Q_{\Lambda,m}(t)
=
\beta_i
\left[
\left(
1-\frac{m}{d}
+
\frac1d e^{-t}
\right)^{-1/\alpha_i}
-1
\right].
\]

Hence, every checkerboard cell probability is evaluated using the same
Gauss--Laguerre quadrature rule as the marginal probability. In particular,
when \(d=1\), the above expression reduces to the marginal probability
presented previously.
\subsubsection{Derivatives of the marginal and cell probabilities}

The marginal probability corresponds to the special case \(d=1\).
Therefore, the derivatives of both the marginal and checkerboard cell
probabilities are obtained from the same numerical procedure.

For the transformed integral, let
\[
h(Q)
=
\frac{e^{-Q}Q^{n_i}}
{\Gamma(n_i+1)},
\]
where
\[
Q
=
\beta_i
\left[
\left(
1-\frac{m}{d}
+
\frac1d e^{-t}
\right)^{-1/\alpha_i}
-1
\right].
\]

Differentiating the integrand with respect to a generic parameter
\(
\vartheta\in\{\mu_i,\phi\}
\)
gives
\[
\frac{\partial h(Q)}
{\partial\vartheta}
=
h(Q)
(n_i-Q)
\frac{\partial\log Q}
{\partial\vartheta}.
\]

Since
\[
\log Q
=
\log\beta_i
+
\log
\left[
\left(
1-\frac{m}{d}
+
\frac1d e^{-t}
\right)^{-1/\alpha_i}
-1
\right],
\]
we obtain
\[
\frac{\partial\log Q}
{\partial\vartheta}
=
\frac1{\beta_i}
\frac{\partial\beta_i}
{\partial\vartheta}
-
\frac{
r
}{
\alpha_i
(1-e^{-r})
}
\frac{\partial\alpha_i}
{\partial\vartheta},
\]
where
\[
r
=
-\frac1{\alpha_i}
\log
\left(
1-\frac{m}{d}
+
\frac1d e^{-t}
\right).
\]

The derivatives
\[
\frac{\partial\alpha_i}{\partial\mu_i},
\qquad
\frac{\partial\alpha_i}{\partial\phi},
\qquad
\frac{\partial\beta_i}{\partial\mu_i},
\qquad
\frac{\partial\beta_i}{\partial\phi}
\]
follow directly from the moment parameterization relating
\((\mu_i,\phi)\) to \((\alpha_i,\beta_i)\).

Consequently, the derivatives of both the marginal and checkerboard cell
probabilities are evaluated by replacing the integrand in the
Gauss--Laguerre quadrature with its analytical derivative. Therefore,
both the probabilities and their derivatives are computed using the same
Gauss--Laguerre quadrature rule and require essentially the same
computational cost.

\subsection{Weibull Mixing}
\label{marginal_weibull}

For the Weibull mixing distribution, we consider the general variance
parameterization
\[
\mathbb V(\Lambda_i\mid\mathbf X_i)
=
\phi\mu_i^{p},
\qquad
p\in\{1,2,3\},
\]
corresponding to the \texttt{var\_power} option in the software
implementation.

The Weibull distribution is parameterized by its shape parameter
\(\kappa_i\) and scale parameter \(\lambda_i\), whose values are uniquely
determined from \((\mu_i,\phi)\) through the moment equations
\[
E(\Lambda_i)=\mu_i,
\qquad
\mathbb V(\Lambda_i)
=
\phi\mu_i^{p}.
\]

Equivalently, letting
\[
a_i=\frac1{\kappa_i},
\]
the scale parameter satisfies
\[
\lambda_i
=
\frac{\mu_i}
{\Gamma(1+a_i)},
\]
where \(a_i\) is obtained by solving the corresponding moment equation.

The Weibull quantile function is
\[
Q_\Lambda(u)
=
\lambda_i
\left[
-\log(1-u)
\right]^{a_i},
\qquad
0<u<1.
\]

Unlike the Gamma and Inverse Gaussian mixing distributions, the marginal
mixed-Poisson probability and the checkerboard cell probabilities do not
admit convenient closed-form expressions and are therefore evaluated
numerically.

\subsubsection{Marginal probability}

The marginal probability corresponds to the special case
\(d=1\),
\[
I_i(n_i)
=
\int_0^1
\frac{
e^{-Q_\Lambda(u)}
Q_\Lambda(u)^{n_i}
}
{\Gamma(n_i+1)}
\,du.
\]

Using the transformation
\[
u
=
1-e^{-t},
\qquad
0\le t<\infty,
\]
gives
\[
I_i(n_i)
=
\int_0^\infty
\frac{
e^{-Q_\Lambda(t)}
Q_\Lambda(t)^{n_i}
}
{\Gamma(n_i+1)}
e^{-t}
\,dt,
\]
where
\[
Q_\Lambda(t)
=
\lambda_i
t^{a_i}.
\]

The transformed integral is evaluated using Gauss--Laguerre quadrature.

\subsubsection{Checkerboard cell probabilities}

For the \(m\)th checkerboard cell,
\[
I_i^{(m)}(n_i)
=
\int_{(m-1)/d}^{m/d}
\frac{
e^{-Q_\Lambda(u)}
Q_\Lambda(u)^{n_i}
}
{\Gamma(n_i+1)}
\,du.
\]

Applying the transformation
\[
1-u
=
1-\frac{m}{d}
+
\frac1d e^{-t},
\qquad
0\le t<\infty,
\]
yields
\[
du
=
\frac1d
e^{-t}\,dt,
\]
and therefore
\[
I_i^{(m)}(n_i)
=
\frac1d
\int_0^\infty
\frac{
e^{-Q_{\Lambda,m}(t)}
Q_{\Lambda,m}(t)^{n_i}
}
{\Gamma(n_i+1)}
e^{-t}
\,dt,
\]
where
\[
Q_{\Lambda,m}(t)
=
\lambda_i
\left[
-\log
\left(
1-\frac{m}{d}
+
\frac1d e^{-t}
\right)
\right]^{a_i}.
\]

Hence, every checkerboard cell probability is evaluated using the same
Gauss--Laguerre quadrature rule as the marginal probability. In particular,
when \(d=1\), the above expression reduces to the marginal probability.

\section{Additional Simulation Results} \label{simulation_result}

This appendix provides supplementary simulation results for the Gaussian factor checkerboard copula. The estimated loading matrices are reported in the main text. Tables~\ref{tab:appendix_J1} and \ref{tab:appendix_J2} compare the true and estimated implied correlation matrices under one-factor and two-factor specifications, respectively. The results demonstrate that the proposed estimation procedure accurately recovers the underlying dependence structure.

\begin{table}[ht]
\centering
\small
\renewcommand{\arraystretch}{1.2} 
\caption{True and estimated implied correlation matrices under the Gaussian factor checkerboard copula (\(J=1\)).}
\label{tab:appendix_J1}
\resizebox{\columnwidth}{!}{%
\begin{tabular}{cc}
\toprule
\textbf{True correlation matrix} 
&
\textbf{Estimated correlation matrix}
\\[1.3mm]
\midrule
$\displaystyle
\begin{bmatrix}
 1.000 & 0.586 & -0.586 & 0.483 & -0.483\\
 0.586 & 1.000 & -0.390 & 0.321 & -0.321\\
-0.586 & -0.390 & 1.000 & -0.321 & 0.321\\
 0.483 & 0.321 & -0.321 & 1.000 & -0.265\\
-0.483 & -0.321 & 0.321 & -0.265 & 1.000
\end{bmatrix}
$

&

$\displaystyle
\begin{bmatrix}
 1.000 & 0.596 & -0.608 & 0.463 & -0.508\\
 0.596 & 1.000 & -0.412 & 0.314 & -0.345\\
-0.608 & -0.412 & 1.000 & -0.320 & 0.351\\
 0.463 & 0.314 & -0.320 & 1.000 & -0.268\\
-0.508 & -0.345 & 0.351 & -0.268 & 1.000
\end{bmatrix}
$
\\ 
\bottomrule
\end{tabular}
}
\end{table}

\begin{table}[ht]
\centering

\renewcommand{\arraystretch}{1.2} 
\caption{True and estimated implied correlation matrices under the Gaussian factor checkerboard copula (\(J=2\)).}
\label{tab:appendix_J2}
\resizebox{\columnwidth}{!}{%
\begin{tabular}{cc}
\toprule
\textbf{True correlation matrix}
&
\textbf{Estimated correlation matrix}
\\[1.3mm] \midrule
$\displaystyle
\begin{bmatrix}
 1.000 & 0.262 & 0.496 & -0.514 & -0.514\\
 0.262 & 1.000 & -0.436 & 0.398 & -0.691\\
 0.496 & -0.436 & 1.000 & -0.656 & 0.100\\
-0.514 & 0.398 & -0.656 & 1.000 & -0.068\\
-0.514 & -0.691 & 0.100 & -0.068 & 1.000
\end{bmatrix}
$
&
$\displaystyle
\begin{bmatrix}
 1.000 & 0.262 & 0.490 & -0.517 & -0.528\\
 0.262 & 1.000 & -0.437 & 0.369 & -0.681\\
 0.490 & -0.437 & 1.000 & -0.678 & 0.085\\
-0.517 & 0.369 & -0.678 & 1.000 & -0.029\\
-0.528 & -0.681 & 0.085 & -0.029 & 1.000
\end{bmatrix}
$
\\ 
\bottomrule
\end{tabular}
}
\end{table}

\section{Dataset used in the numerical application}
\label{appendix:data}
The empirical studies employ four datasets from different application domains.
The MEPS dataset contains healthcare utilization records together with
individual-level demographic, socioeconomic, and health-related characteristics.
The terrorism dataset records the numbers of attacks by Fulani extremists and
Boko Haram across geographical grid cells in Nigeria, along with population,
terrain, and spatial-location variables. The VHLSS dataset contains information
on outpatient and inpatient healthcare utilization, demographic characteristics,
socioeconomic status, health insurance coverage, and geographical regions.
Finally, the single-cell RNA sequencing dataset contains gene-expression counts
and cell-line information and is used to assess the proposed model in a
moderately high-dimensional setting.

The MEPS, terrorism, and single-cell RNA sequencing datasets are distributed
with the \texttt{GJRM}, \texttt{bizicount}, and \texttt{PLNmodels} packages,
respectively. The processed VHLSS dataset is publicly available as a CSV file
from the GitHub repository accompanying this paper. 

All code used to reproduce the simulation studies and empirical analyses is publicly available from the GitHub repository associated with this paper. The MEPS, \texttt{terror}, and \texttt{scRNA} datasets can be loaded directly from the corresponding \textsf{R} packages using the scripts provided in the GitHub repository. The VHLSS dataset is provided as a CSV file in the same repository. Since this dataset requires data cleaning and preprocessing before model fitting, the corresponding preprocessing scripts are also included. Detailed instructions for reproducing all analyses are available in the repository.

\begin{table*}[!t]
\centering
\scriptsize
\renewcommand{\arraystretch}{1.2}
\setlength{\tabcolsep}{4pt}

\caption{Regression specifications used in the empirical studies.}
\label{tab:datasets_regression_specifications}

\begin{threeparttable}

\begin{tabular}{
p{1.5cm}
p{1.8cm}
p{4.2cm}
p{4.0cm}
p{4.0cm}
}

\toprule

\textbf{Dataset}
&
\textbf{Response}
&
\textbf{Count model}
&
\textbf{Zero model}
&
\textbf{Purpose}

\\

\midrule


\multirow{2}{*}{MEPS}

&
\texttt{dvisit}

&
\makecell[l]{
\texttt{bmi + income + age}\\
\texttt{+ education + ethnicity2}\\
\texttt{+ ethnicity3 + ethnicity4}\\
\texttt{+ region2 + region3 + region4}\\
\texttt{+ gender + hypertension}\\
\texttt{+ hyperlipidemia}
}

&
--

&
\multirow{2}{3.5cm}[-2.5em]{
Evaluation of the proposed bivariate mixed Poisson model and comparison with \texttt{GJRM}.
}

\\
\addlinespace[1ex]
\cmidrule(lr){2-4}
&

\texttt{ndvisit}

&
\makecell[l]{
\texttt{bmi + income + age}\\
\texttt{+ education + ethnicity2}\\
\texttt{+ ethnicity3 + ethnicity4}\\
\texttt{+ region2 + region3 + region4}\\
\texttt{+ gender + hypertension}\\
\texttt{+ hyperlipidemia}
}

&
--

&

\\
\addlinespace[1ex]
\midrule


\multirow{2}{*}{Terror}

&
\texttt{att.ful}

&
\makecell[l]{
\texttt{pop + mtns}\\
\texttt{+ xcoord*ycoord}
}

&
\makecell[l]{
\texttt{pop + mtns}\\
\texttt{+ xcoord*ycoord}
}

&
\multirow{2}{3.5cm}[-1.0em]{
Evaluation of the proposed
bivariate zero-inflated model
and comparison with
\texttt{bizicount}.
}

\\
\addlinespace[2ex]
\cmidrule(lr){2-4}
&

\texttt{att.bok}

&
\makecell[l]{
\texttt{pop + mtns}\\
\texttt{+ xcoord*ycoord}
}

&
\makecell[l]{
\texttt{pop + mtns}\\
\texttt{+ xcoord*ycoord}
}

&

\\
\addlinespace[1ex]
\midrule


\multirow{2}{*}{VHLSS}

&
\texttt{FRE\_OUT}

&
\makecell[l]{
\texttt{GENDER + RELATE\_new\_OUT}\\
\texttt{+ GROUP\_AGE + MARRIED\_new\_OUT}\\
\texttt{+ EDU\_new + WORKING + REGION}\\
\texttt{+ INSURANCE\_STATUS}\\
\texttt{+ INCOME\_LEVEL\_new}\\
\texttt{+ SOCIAL\_ALLOWANCE\_MONTHLY}
}

&
\makecell[l]{
\texttt{GENDER + RELATE\_new\_OUT}\\
\texttt{+ GROUP\_AGE + MARRIED\_new\_OUT}\\
\texttt{+ INTERNET\_OR\_NOT + EDU\_new}\\
\texttt{+ INSURANCE\_STATUS + REGION}\\
\texttt{+ SOCIAL\_ALLOWANCE\_MONTHLY}
}

&
\multirow{2}{3.5cm}[-2.5em]{
Evaluation of the proposed bivariate zero-inflated mixed Poisson model using administrative healthcare data.
}

\\
\addlinespace[2ex]
\cmidrule(lr){2-4}
&

\texttt{FRE\_IN}

&
\makecell[l]{
\texttt{GENDER + RELATE\_new\_IN}\\
\texttt{+ GROUP\_AGE\_new\_IN}\\
\texttt{+ EDUCATION + WORKING}\\
\texttt{+ REGION}\\
\texttt{+ SOCIAL\_ALLOWANCE\_MONTHLY}
}

&
\makecell[l]{
\texttt{GENDER + RELATE\_new\_IN}\\
\texttt{+ GROUP\_AGE\_new\_IN}\\
\texttt{+ MARRIED\_STATUS}\\
\texttt{+ INTERNET\_OR\_NOT}\\
\texttt{+ EDUCATION}\\
\texttt{+ INSURANCE\_STATUS}\\
\texttt{+ INCOME\_LEVEL\_new}\\
\texttt{+ SOCIAL\_ALLOWANCE\_MONTHLY}
}

&

\\
\addlinespace[1ex]
\midrule


\multirow{3}{*}{scRNA}

&
Gene 1

&
\texttt{cell\_line}

&
--

&
\multirow{3}{3.5cm}{
Evaluation of the proposed multivariate mixed Poisson model in moderately high-dimensional settings.
}

\\

&

Gene 2

&
\texttt{cell\_line}

&
--

&

\\

&

Gene $m$, $m\ge3$

&
\texttt{cell\_line}

&
\texttt{cell\_line}

&

\\
\addlinespace[1ex]
\bottomrule

\end{tabular}

\begin{tablenotes}
\footnotesize
\item
MEPS: Medical Expenditure Panel Survey;
VHLSS: Vietnam Household Living Standard Survey.
For zero-inflated models (Terror and VHLSS), separate regression equations are specified for the count and zero-inflation components.
\end{tablenotes}

\end{threeparttable}

\end{table*}

\newpage

\bibliographystyle{apalike}
\bibliography{refe}

\end{document}